\documentclass[12pt]{article}
\usepackage{natbib}
\usepackage{setspace}
\usepackage[utf8]{inputenc}
\usepackage{amsmath}% http://ctan.org/pkg/amsmath
\usepackage{graphicx}
\usepackage{comment}
\usepackage{enumerate}
\usepackage{mathtools}
\usepackage{amssymb}
\usepackage{url}
\usepackage{comment}
\newcommand{\RNum}[1]{\uppercase\expandafter{\romannumeral #1\relax}}
\usepackage{float}
\usepackage[letterpaper, margin=1in]{geometry}
\newcommand{\indep}{\perp \!\!\! \perp}
\usepackage{algpseudocode}
\usepackage{algorithm}
\usepackage{lipsum}
\usepackage{stix}
\DeclareMathAlphabet\mathbfcal{LS2}{stixcal}{b}{n}
\usepackage{amsmath}
\usepackage{amsthm}
\usepackage{prodint}
\usepackage{enumitem}
\usepackage{bbm} % for indicator function \mathbbm{1}
\usepackage{tikz, pgfplots}
\usetikzlibrary{positioning,fit,calc}

\usetikzlibrary{arrows.meta}

\allowdisplaybreaks

\usepackage{titlesec}
\titlespacing*{\section}{0pt}{1.2ex plus .4ex minus .3ex}{0.6ex plus .2ex}
\titlespacing*{\subsection}{0pt}{1.0ex plus .3ex minus .2ex}{0.4ex plus .2ex}
\titlespacing*{\subsubsection}{0pt}{0.8ex plus .3ex minus .2ex}{0.3ex plus .2ex}
\newtheorem{assumption}{Assumption}
\newtheorem{prop}{Proposition}
\newtheorem{lemma}{Lemma}

\newtheorem{assumptionalt}{Assumption}[assumption]

\newenvironment{assumptionp}[1]{
  \renewcommand\theassumptionalt{#1}
  \assumptionalt
}{\endassumptionalt}

\title{Causal mediation analysis for longitudinal and survival data in continuous time using Bayesian non-parametric joint models}

\author{%
  Saurabh Bhandari$^{1}$,\quad
  Michael J.~Daniels$^{2,*}$,\quad and
  Juned Siddique$^{3}$ \\[1ex]
  $^{1}$Department of Public Health Sciences, University of Chicago,
  Chicago, Illinois, U.S.A. \\
  $^{2}$Department of Statistics, University of Florida,
  Gainesville, Florida, U.S.A. \\
  $^{3}$Department of Preventive Medicine, Northwestern University,
  Chicago, Illinois, U.S.A. \\[1ex]
  $^{*}$\textit{email:} \texttt{daniels@ufl.edu}
}

\date{\today}

\begin{document}

\maketitle

\begin{abstract}
     Observational cohort data is an important source of information for understanding the causal effects of treatments on survival and the degree to which these effects are mediated through changes in disease-related risk factors. However, these analyses are often complicated by irregular data collection intervals, longitudinal mediators, and the need to account for longitudinal confounders. We propose a causal mediation framework that jointly models longitudinal exposures, confounders, mediators, and time-to-event outcomes as continuous functions of age. This framework for longitudinal risk factor trajectories enables statistical inference even at ages where the subject's risk factor measurements are unavailable.  The observed data distribution in our framework is modeled using an enriched Dirichlet process mixture (EDPM) model. Using data from the Atherosclerosis Risk in Communities cohort study, we apply our methods to assess how medication prescribed to target cardiovascular disease (CVD) risk factors affects the time-to-CVD death.
\end{abstract}

% --- Main-text-only: tighten vertical space around display math (equation, align, $$). Reverted before \appendix. ---
\setlength{\abovedisplayskip}{2pt plus 1pt minus 1pt}
\setlength{\belowdisplayskip}{2pt plus 1pt minus 1pt}
\setlength{\abovedisplayshortskip}{2pt plus 1pt minus 1pt}
\setlength{\belowdisplayshortskip}{2pt plus 1pt minus 1pt}

\section{Introduction}\label{Sec1}

In longitudinal observational studies, researchers are often interested in understanding the impact of time-varying exposures on time-to-event (survival) outcomes. To track the development of specific health outcomes or events, such studies typically collect a rich set of participant covariate information, some of which may act as mediators of the relationship between the exposures and the outcomes, while others serve as confounders. Drawing causal conclusions in this setting is not straightforward. We need to carefully define the underlying data structure and develop statistical methods that allow us to identify causal effects from the observed data while accounting for time-dependent confounding. In addition, the irregular timing of data collection within and across study participants introduces an additional layer of difficulty for both modeling and inference.

\noindent To address these challenges, we develop a framework to identify and estimate the direct and indirect causal effects of exposure on a right-censored survival outcome in a continuous-time setting. We consider a longitudinal data structure involving a time-varying exposure and mediator, both subject to time-varying confounding. We assume that subjects' treatment and other explanatory variables are updated at common time intervals, but the subjects need not enter the study at the same time. A review of the causal mediation literature relevant to our framework is provided in Section \ref{Sec11}.

\noindent We model the observed data distribution using an enriched Dirichlet process mixture (EDPM) model \citep{wade2011enriched, wade2014improving}. To the best of our knowledge, this is the first use of the EDPM for joint modeling of longitudinal and survival data.  Bayesian non-parametric (BNP) mixture models, such as the Dirichlet process mixture (DPM) model, have been separately used to model longitudinal \citep{quintana2016bayesian,dunson2009bayesian} and survival \citep{zhao2015dirichlet,hanson2004bayesian} data. A relevant work that assigns EDP prior to regression parameters of longitudinal data in continuous time is available in \cite{zeldow2021functional}. We extend the existing models by incorporating multiple longitudinal variables (treatment, confounder, and mediator) and survival outcomes into a single framework to solve a causal mediation problem.

\noindent The literature on treatment effects estimation in continuous time highlights two major approaches to formulate the causal framework. One approach involves the martingale theory, where the possible observations of a subject are represented in terms of a counting process model \citep{rytgaard2022continuous, lok2008statistical, roysland2011martingale}. An alternative approach involves modeling the survival outcome and/or longitudinal exposures, mediators, and confounders in continuous time using parametric/semi-parametric models, assuming a fixed age grid where the explanatory variables may update, or the event of interest may occur \citep{zeng2022causal, sun2022estimating}. In related work, \cite{sun2022estimating} propose an estimator of the mean potential outcome under the assumption that time-dependent variables are updated at a set of common time grids within a dynamic treatment regime setting. In a continuous-time framework that does not include mediators, they develop a G-computation-based method for treatment effect estimation, assuming that the longitudinal predictors are updated only at a finite and common set of time points across the population. In other words, they assume that the predictors' histories are piecewise constant, with updates occurring at these shared time grids. We utilize this idea in the context of mediation analysis in continuous time with a survival outcome.

\noindent In this article, we model the conditional distributions of longitudinal exposure, confounder, and mediator as smooth functions of age. This setup allows us to draw from their predictive distributions at any age, including ages at which these variables may not have been observed. From a scientific perspective, this allows us to address questions about how treatment affects the outcome across different ages or age groups of interest. We use this framework to study how blood pressure medication influences time to cardiovascular disease (CVD) death, and to decompose this effect into its direct and indirect pathways.

\subsection{Relevant background on longitudinal causal mediation analysis}\label{Sec11}

The most prevalent causal estimands in the mediation analysis literature are the natural causal effects \citep{pearl2014interpretation}. In the point exposure setting, these effects are defined as differences in expected potential outcomes resulting from contrasting exposure and mediator regimes. They are referred to as natural effects because the potential outcomes are constructed as nested counterfactuals arising from two parallel worlds: the exposure is fixed at one level across the entire population, while the mediator is set not to a constant value, but to the value it would have naturally taken had the exposure been set to a different level. In the longitudinal setting, however, causal estimands defined using such nested counterfactuals suffer from identification issues due to time-varying confounding. The task of adjusting for time-varying confounding becomes further complicated in the presence of survival outcomes when the causal parameters of interest are natural direct and indirect effects.

\noindent Previous work \citep{lin2017mediation, zheng2017longitudinal,zeng2022causal,didelez2019defining} on causal mediation analysis with longitudinal variables discusses two key challenges in defining natural causal effects in this setting, which we briefly mention here. First, when the longitudinal confounders are affected by past exposures---the so-called "recanting witness"---the natural direct and indirect effects cannot be identified from the observed data \citep{avin2005identifiability}. Second, when analyzing the effect of exposure on survival outcome, if subjects die before the end of the study, their so-called "cross-world counterfactual" mediator
values may be poorly defined. This is because defining the natural causal effects requires us to construct potential survival outcomes as nested counterfactuals arising from two parallel worlds for the same subject. In a survival setting, such a definition is problematic because it allows for the possibility that a subject may survive longer in "one world" than in the other
"counterfactual world", making the value assigned to the mediator process ill-defined in the latter, where the subject dies earlier.

\noindent  Several approaches have been proposed to address these challenges. These include separable effects \citep{didelez2019defining}, assumptions that extend mediator trajectories beyond observed survival times \citep{zeng2022causal}, and interventional (randomized) mediation effects that operate on mediator distributions rather than fixed values \citep{vanderweele2017mediation, lin2017mediation, zheng2017longitudinal, wang2025targeted}. In this article, we adopt the interventional framework, defining causal parameters through static interventions on the exposure and stochastic interventions on the mediator, following \cite{zheng2017longitudinal} and \cite{wang2025targeted}. The concepts of interventional direct and indirect effects, as well as randomized interventions, have been extensively discussed in previous work by the authors \citep{bhandari2025bayesian}. For detailed information, we refer readers to this paper.

\noindent Prior work has also explored causal mechanisms through joint analysis of longitudinal and survival data \citep{liu2018exploring, zheng2022quantifying, caubet2023bayesian, zhang2021mediation,shardell2018joint}. However, these methods typically rely on parametric modeling approaches for the observed data, assume a discrete-time structure for longitudinal measurements, omit mediators in the causal analysis, fail to incorporate longitudinal exposures and confounders, or capture time-varying confounding through shared parameters such as random effects. In contrast, our framework uses flexible Bayesian models for the joint analysis of longitudinal exposures, confounders, mediators, and time-to-event outcomes, all defined as continuous functions of age.

\noindent A recent study by \cite{kateline2025continuous} develops a continuous-time mediation framework that assumes a cross-sectional exposure and no effect of the mediator on future time-varying confounders. In contrast to our framework, their identification strategy relies on a cross-world independence assumption to identify path-specific effects. They also adopt a differential equation–based parametric modeling approach for the observed data and use a parametric bootstrap for inference, both of which may be sensitive to model misspecification. In contrast, our framework accommodates time-varying exposures, captures feedback between mediators and confounders, and uses a flexible Bayesian nonparametric approach for modeling the observed data.

\noindent The remainder of the article is structured as follows. In Section 2, we describe our motivating example and present a summary of two cohort studies that we combine in our analysis. In Section 3, we introduce our notation and data structure. Section 4 describes our causal framework. In Section 5, we present our proposed observed data models together with the G-computation algorithm. In Section 6, we conduct simulation studies to compare the performance of our proposed model with a standard Bayesian parametric model. Section 7 presents the results of our analysis, which is followed by a brief discussion in Section 8.

\section{Motivating study}\label{Sec2}
Our proposed methods examine the causal effect of blood pressure (BP) medication on time-to-death due to cardiovascular disease (CVD) at any specific age, accounting for time-varying mediators and confounders in the Atherosclerosis Risk in Communities (ARIC) cohort study. ARIC is a cohort study with a total of $15,792$ participants from four geographically diverse communities (ARIC Investigators, 1989) \nocite{aric1989atherosclerosis}. Our analysis uses a de-identified, limited-access version of the dataset from the NHLBI BioLINNC repository.

\noindent A key challenge in our setting is that ARIC participants are followed at irregularly spaced clinic visits and enter the study at different baseline ages. In ARIC, the first four clinic exams (in which participant data were collected) were conducted three years apart, whereas the subsequent three exams occurred four years apart. This irregular visit schedule, together with heterogeneity in baseline age at study entry, complicates the estimation of treatment effects at a fixed age across participants.

\noindent  To address this challenge, our proposed method conceptualizes a hypothetical prospective cohort study in which all participants are enrolled at the same baseline age (for example, age 45) and followed with longitudinal assessments occurring at regular, predetermined ages (for example, 45, 50, 55, 60, and so forth). In such an idealized design, the causal question of interest would be straightforward: what is the effect of antihypertensive medication on survival probability at age 55, decomposed into direct and indirect pathways through blood pressure? The common enrollment age ensures that predictor histories are directly comparable across individuals at each age on the grid. Our framework emulates this idealized study using data from ARIC. 

\noindent We use smooth trajectory functions for longitudinal predictors to impute each participant's predictor values at any age of interest. This imputation effectively reconstructs the data structure that would have been observed had all participants been enrolled at a common age and assessed on a shared age grid. Consequently, our causal effect estimates can be interpreted as if they arose from the hypothetical same-age cohort design, with the age grid $a_1 < a_2 < \cdots < a_K$  (formalized in Section~\ref{Sec3}) representing the common ages at which treatment and risk factor updates are conceptualized to occur.

\noindent Figure~\ref{Fig:DAG_Application} illustrates the causal structure underlying our application. In this simplified directed acyclic graph (DAG), antihypertensive medication use is a time-varying treatment, mean blood pressure is a time-varying mediator, and smoking status is a time-varying confounder that affects future variables while also being influenced by prior treatment. The outcome of interest is time to CVD-related death, summarized in our analysis using the survival probability at a given age. Our estimation strategy explicitly respects this temporal structure.

\begin{figure}[H]
\singlespacing
\centering
\begin{tikzpicture}[
    scale = 0.8, every node/.style={circle, draw=black, thick, minimum size=0.8cm, font=\footnotesize, align=center},
    arrow/.style={->, thick, >=Stealth}
]

% Nodes on the horizontal axis
\node (Za) at (-1.5,0) {Antihypertensive \\ Medication};
\node (Ma) at (4,0) {Mean Blood \\ Pressure};
\node (Ti) at (10,0) {Time-to-death \\
by CVD};

% Node above Ma and Ti
\node (La) at (7,4) {Smoking Status};

% Horizontal arrows
% Arrows from M(a) to ALL future covariates
\draw[arrow] (Ma) to (Ti);
% Arrows from L(a) to ALL future covariates

\draw[arrow] (La) -- (Ma);
\draw[arrow] (La) -- (Ti);

% Curved arrow from Z(a) ALL future covariates
\draw[arrow] (Za) -- (Ma);
\draw[arrow, bend left=25] (Za) to (La);
\draw[arrow, bend left=-25] (Za) to (Ti);

\end{tikzpicture}
\caption{Directed acyclic graph illustrating the causal relationships among antihypertensive medication use, mean blood pressure, smoking status, and time to cardiovascular disease death.}
\label{Fig:DAG_Application}
\end{figure}

\noindent We apply our proposed method to estimate the causal effects of time-varying antihypertensive medication on age-specific survival, decomposed into interventional direct and indirect effects through mean BP. Mean BP, defined as the average of systolic blood pressure (SBP) and diastolic blood pressure (DBP) measured at each study visit, serves as the mediator, while smoking status acts as a time-varying confounder of the treatment--mediator, treatment--outcome, and mediator--outcome relationships. Baseline confounders included in our study are biological sex, race, body mass index, educational attainment, smoking status (never, former, current), diabetes status, baseline age, and the total-to-HDL cholesterol ratio. All causal effects are defined and estimated conditional on this baseline covariate set.

\noindent For the data analysis, we restrict the study population to participants who were hypertensive at baseline, defined as having systolic blood pressure (SBP) $\geq 140$ mmHg or diastolic blood pressure (DBP) $\geq 90$ mmHg at the first clinic visit. Accordingly, the causal effects of interest pertain to individuals for whom hypertension is already present. This restriction is imposed because these individuals are more likely to receive antihypertensive treatment, ensuring adequate variation in treatment exposure for causal effect estimation. The final analytic sample consists of $2{,}552$ participants, of whom $1{,}218$ are male and $1{,}334$ are female by biological sex. The age range of subjects at the first visit spans 42 to 66 years, while at the fifth visit the corresponding age range is 67 to 90 years (see Figure \ref{Fig:ARIC_age_distributions}).

\begin{figure}[H]
\singlespacing
\centering
\includegraphics[scale = 0.4]{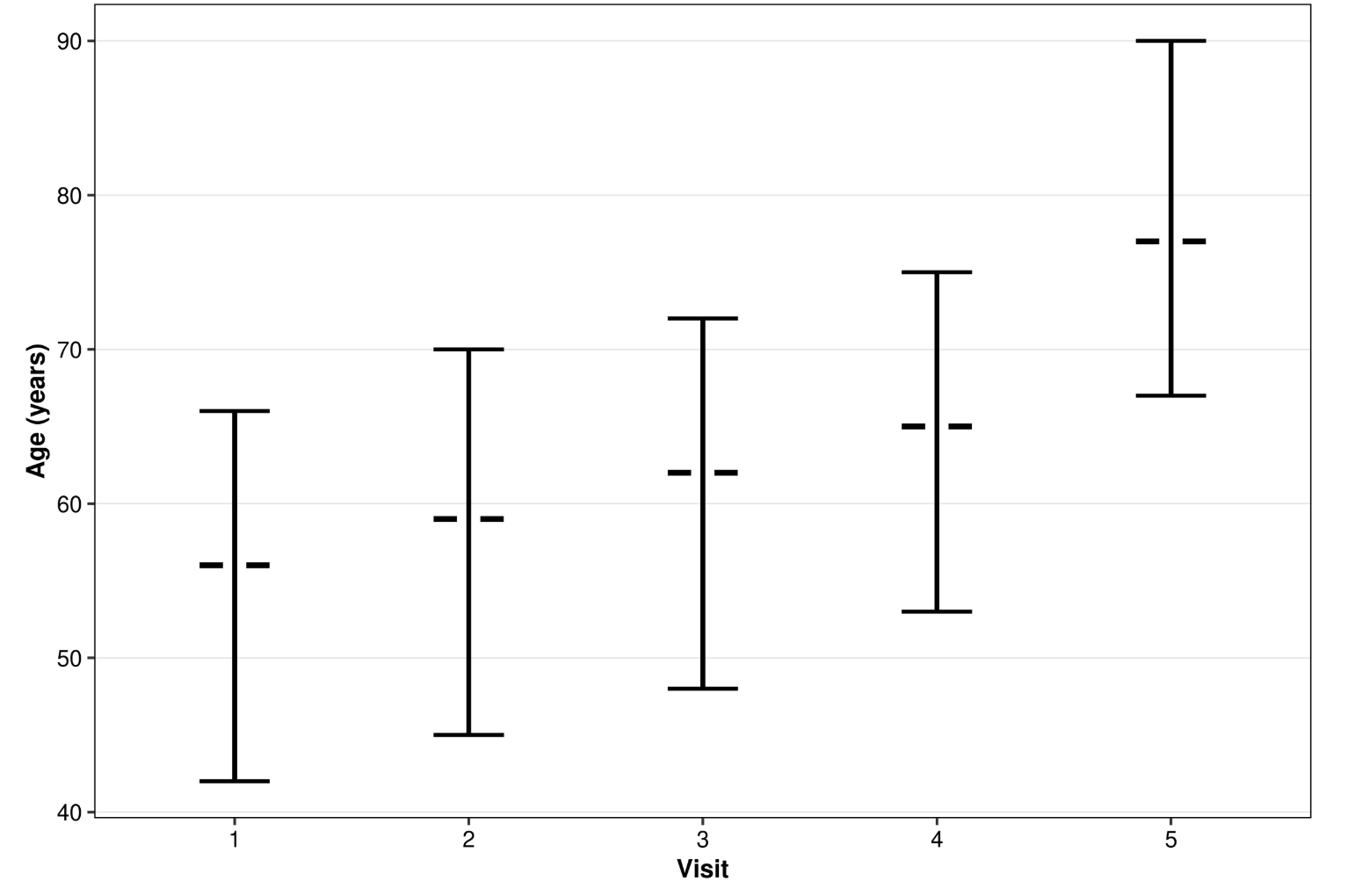}
\caption{Age distribution across ARIC study visits for baseline hypertensive subjects. Vertical bars represent the observed age range (minimum to maximum), with dashed horizontal lines indicating the median age at each visit.}
\label{Fig:ARIC_age_distributions}
\end{figure}

\section{Notation and setup}\label{Sec3}
Suppose we have data of the form: $T_{i}, C_{i},  \delta_{i}, \textbf{Z}_{i}(a),\textbf{L}_{i}(a),\textbf{M}_{i}(a),\textbf{L}_{i,0}$ where $a \in [0, \mathcal{A}]$ denotes continuous age for participants $i = 1,\ldots,n$  in the study. A simplified directed acyclic graph (DAG) representing the temporal relationships among these variables is shown in Figure \ref{Fig:DAG}.

\begin{figure}[H]
\singlespacing
\centering
\begin{tikzpicture}[
    every node/.style={circle, draw=black, thick, minimum size=1.4cm, font=\footnotesize, align=center},
    arrow/.style={->, thick, >=Stealth}
]

% Nodes on the horizontal axis
\node (L0) at (-4.5,0) {$L\textsubscript{i,0}$};
\node (Za) at (-1.5,0) {$Z_i(a)$};
\node (Ma) at (1,0) {$M_i(a)$};
\node (Za1) at (3.5,0) {$Z_i(a + \Delta)$};
\node (Ma1) at (6,0) {$M_i(a + \Delta)$};
\node (Ti) at (10,0) {$(T_i,\delta_i)$};

% Node above Ma and Ti
\node (La) at (2,3) {$L_i(a)$};
\node (La1) at (7,3) {$L_i(a + \Delta)$};

% Curved and straight short arrows from L0
\draw[arrow, bend left=30] (L0) to +(1.5,1);  % upward-right
\draw[arrow, bend right=30] (L0) to +(+1.5,-1); % downward-right
\draw[arrow] (L0) -- +(1.5,0);                 % straight right

% Straight short arrows from Ma1
\draw[arrow] (Ma1) -- +(1.5,0);
\draw[arrow, bend left=-35] (Ma1) to (Ti);
% Dots after Ma1
\node[draw=none, fill=none, inner sep=0pt] at (7.75, 0) {$\cdots$};
\draw[arrow] (8, 0) -- (Ti);
% Dots before Z(a)
\node[draw=none, fill=none, inner sep=0pt] at (-2.5, 0) {$\cdots$};

% Horizontal arrows
% Arrows from M(a) to ALL future covariates
\draw[arrow] (Ma) -- (Za1);
\draw[arrow, bend left=25] (Ma) to (La1);
\draw[arrow, bend left=-35] (Ma) to (Ma1);
\draw[arrow, bend left=-35] (Ma) to (Ti);
% Arrows from L(a) to ALL future covariates

\draw[arrow] (La) -- (Ma);
\draw[arrow] (La) -- (Za1);
\draw[arrow] (La) -- (Ma1);
\draw[arrow] (La) -- (Ti);
\draw[arrow] (La) -- (La1);
% Arrows from La1 to ALL future covariates

\draw[arrow] (La1) -- (Ma1);
\draw[arrow] (La1) -- (Ti);

% Curved arrow from Z(a) ALL future covariates
\draw[arrow] (Za) -- (Ma);
\draw[arrow, color = blue, bend left=25] (Za) to (La);
\draw[arrow, color = blue, bend left=60] (Za) to (La1);
\draw[arrow, bend left=-35] (Za) to (Za1);
\draw[arrow, bend left=-35] (Za) to (Ma1);
\draw[arrow, bend left=-35] (Za) to (Ti);

% Curved arrow from Za1 ALL future covariates
\draw[arrow] (Za1) -- (Ma1);
\draw[arrow, color = blue, bend left=25] (Za1) to (La1);
\draw[arrow, bend left=-35] (Za1) to (Ti);

\end{tikzpicture}
\caption{A directed acyclic graph (DAG) representing the temporal relationships among baseline and longitudinal variables, as well as the survival outcomes. The blue curved arrows illustrate the recanting witness issue, where longitudinal confounders are influenced by past exposures. The change in longitudinal processes from age $a$ to $a+\Delta$ represents variation over a small time interval.}
\label{Fig:DAG}
\end{figure}

\noindent In this representation, $T_i$ denotes the true event time (in terms of age), $C_i$ represents the censoring time (also in age), and $\delta_i$ is the event indicator that distinguishes event ($\delta_i =1$) from right censoring  ($\delta_i =0$) for individual $i$. Individuals with $C_i \leq T_i$ are right-censored, so we only observe $\tilde{T}_i = min\big(T_{i}, C_{i}\big)$ and the event indicator $\delta_i$. For subject $i$, we make observations at $n_i$ visit ages $ \{a_{i,j} \in [0, \mathcal{A}], j = 1, 2, \ldots, n_i\}$, and the interval between two consecutive visit ages can differ within and across the subjects. In other words, at each visit age $a_{i,j}$, each subject $i$ $(i = 1, 2, \ldots , n)$ is assigned to a treatment $\bigl(Z_{i}(a_{i,j}) = 1\bigr)$ or a control $\bigl( Z_{i}(a_{i,j}) = 0 \bigr)$ group. Similarly, at each $a_{i,j}$, we measure time-dependent confounders and  mediators, $L_{i}(a_{i,j})$ and $M_{i}(a_{i,j})$, respectively. Thus, each individual $i$ has the following observed data in a bounded interval $[0,\mathcal{A}]$:
\begin{equation}\label{Eq:DataStructure}
    \begin{aligned}
        \begin{bmatrix}
            \textbf{L}_{i,0},
            Z_i(a_{i1}),
            L_i(a_{i1}),
            M_i(a_{i1}),
            \ldots,
            Z_i(a_{i,n_i}) ,
            L_i(a_{i,n_i}),
            M_i(a_{i,n_i}),
            \big(\tilde{T}_i,\delta_i\big)
        \end{bmatrix} , \text{ where } a_{i,n_i} \leq \mathcal{A}.
    \end{aligned}
\end{equation}
In Equation (\ref{Eq:DataStructure}), $\textbf{L}_{i,0}$ is a vector representing baseline confounders, such as age at entry (baseline age), sex, race, and education level. 

\noindent We use the bold letters to denote process until age for the longitudinal exposure $\bigl(\textbf{Z}_{i}(a) \equiv \{Z_{i}(\omega): \omega \leq a \in [0, \mathcal{A}]\}\bigr)$, longitudinal confounder $\bigl(\textbf{L}_{i}(a) \equiv \{L_{i}(\omega): \omega \leq a \in [0, \mathcal{A}]\}\bigr)$ and the longitudinal mediator $\bigl(\textbf{M}_{i}(a) \equiv \{M_{i}(\omega): \omega \leq a \in [0, \mathcal{A}]\}\bigr)$. We assume that any paths of $\textbf{Z}$, $\textbf{L}$, and $\textbf{M}$ on $[0,\mathcal{A}]$ are Lipschitz continuous. The Lipschitz continuity assumption ensures that these processes are well-behaved, preventing extreme fluctuations or discontinuities.

\noindent Let $a_1 < a_2 < \dots < a_{k} < \dots < a_{K} < \dots$ be the age grid where at least one event occurs in each interval $(a_{k-1}, a_{k}]$, and consider a fixed survival age $a \in \mathbb{R}^{+}$ such that $a_{K} \leq a \leq a_{K+1}$. Note that the ages on this grid, denoted by $a_k$, can be distinct from the visit ages, denoted by $a_{i,j}$, which, in turn, can be different from the fixed age $a$. Below, we briefly describe these quantities.

\noindent The fixed survival age $a \in \mathbb{R}^{+}$ is a continuous variable at which we want to compute causal effects and conduct statistical inference. The predictors may or may not be measured at this value. The age grid $a_1 < a_2 < \dots < a_{k} < \dots < a_{K} < \dots$ is a conceptual framework where we assume that treatment and risk factor updates are possible, with the requirement that at least one event takes place within each interval  $(a_{k-1}, a_{k}]$. The interval widths are not required to be equal across the grid. To satisfy this requirement, the intervals may be spaced sufficiently far apart---for example, every 5 years. Based on our proposed models, we can make predictions at any age within this grid. This grid plays a crucial role in developing our causal mediation framework, as described in Section \ref{Sec4}. Similarly, the visit age $a_{i,j}$ represents the age of subject $i$ at the $j^{th}$ risk factor assessment ($a_{i,n_i} \coloneqq$ the age of individual $i$ at their final risk factor assessment before death or censoring). At each $a_{i,j}$, where $j \in \{1,\dots,n_i\}$, we measure the time-varying exposures, confounders, and mediators—denoted as $Z_{i}(a_{i,j})$, $L_{i}(a_{i,j})$, and $M_{i}(a_{i,j})$, respectively—which are well-defined as long as individual $i$ is alive at age $a_{i,j}$. Finally, for subject $i$, let $t_i$ denote their observed age at the time of the event, i.e., an observation of the random variable $T_i$. We use the notation $t_i$ rather than an age-based notation to distinguish the observed event/censoring age from the visit ages $a_{i,j}$ and the values $a_k$ in the age grid.

\noindent As an example (see Figure \ref{Fig:Age_grid_example}), suppose we wish to compute direct and indirect causal effects at age $a = 60$. We can define an age grid as $a_1 = 45 < a_2 = 50 < a_3 = 55 < a_4 = 60 (= a)$, setting $K = 4$. Among the $N$ participants in the study, consider two subjects, A and B. Subject A, who dies at age $58$, has risk factors measured at visit ages $a_{1,1} = 49$, $a_{1,2} = 54$, and $a_{1,3} = 57$. On the other hand, Subject B, who remains alive until the end of the study, has risk factors measured at visit ages $a_{2,1} = 55$, $a_{2,2} = 58$, $a_{2,3} = 61$, and $a_{2,4} = 64$. Note that the intervals between consecutive visit ages vary both within and across Subjects A and B, whereas the intervals between consecutive ages in the grid vary only within each subject. The observed age at death or censoring, $t_i$, for Subjects A and B are $58$ and $64$, respectively.

\begin{figure}[H]
\singlespacing
\centering
\includegraphics[width = \textwidth]{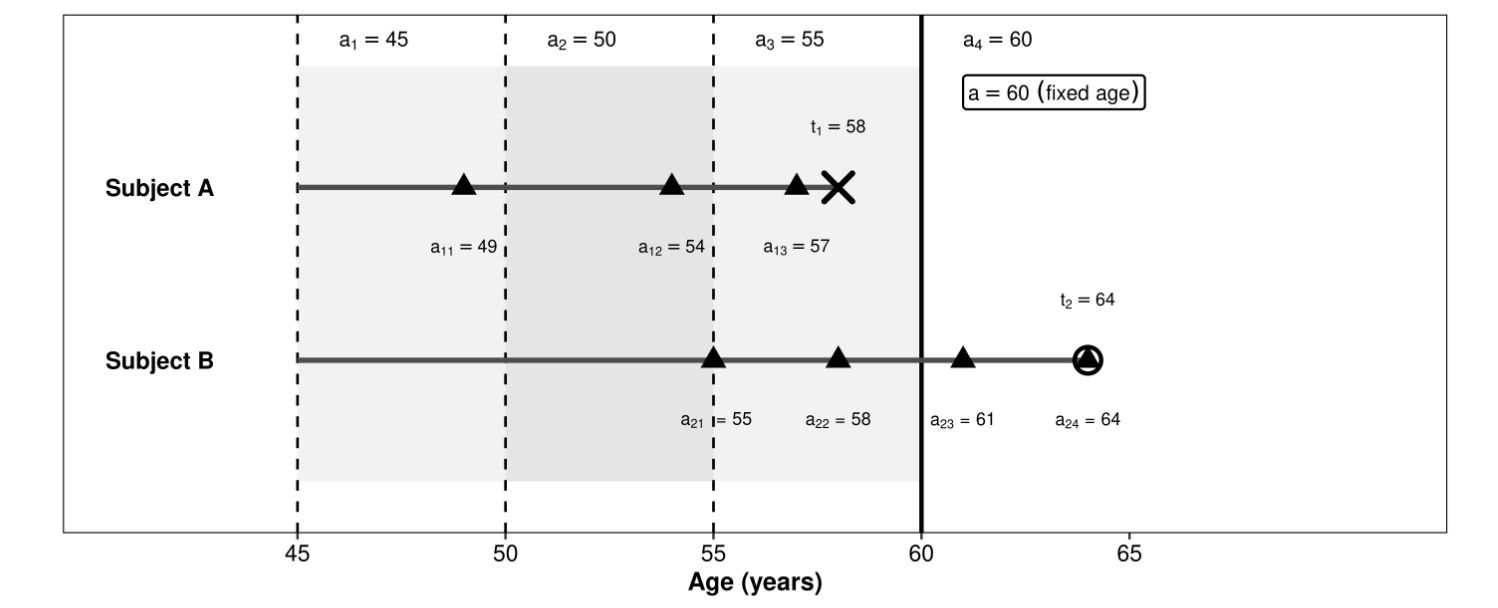}
\caption{Illustration of the age grid, visit ages, and event times. Vertical dashed lines represent the age grid $a_1 = 45 < a_2 = 50 < a_3 = 55 < a_4 = 60$. The solid vertical line at age 60 represents the fixed age $a$ at which causal effects are computed. Subject A (top) has risk factor measurements at visit ages $a_{1,1} = 49$, $a_{1,2} = 54$, and $a_{1,3} = 57$, and dies at age $t_1 = 58$. Subject B (bottom) has measurements at visit ages $a_{2,1} = 55$, $a_{2,2} = 58$, $a_{2,3} = 61$, and $a_{2,4} = 64$, and is censored at age $t_2 = 64$.}
\label{Fig:Age_grid_example}
\end{figure}

\section{Causal estimands and identification}\label{Sec4}

\subsection{Potential outcomes}\label{Sec41}
For any variable $W$, let $W_{(u,v)}$ denote the counterfactual value of $W$ when $U$ and $V$ are set---possibly, contrary to the observed values---to $u$ and $v$, respectively. Following the arguments from \cite{zheng2017longitudinal} and \cite{wang2025targeted}, we define counterfactuals by joint intervention on exposure and mediator history. This joint intervention is established by a static intervention on the exposure variables, followed by a random intervention on the mediator variables. Suppose that $\textbf{z}(a)$ and $\textbf{z}_{*}(a)$ represent two possible exposure regimes. Define $\textbf{z}_{*}(a^{-})$ and $\textbf{m}(a^{-})$ as the exposure process $\textbf{z}_{*}$ and mediator process $\textbf{m}$, respectively, up to, but not including, age $a$. Let 
\begin{equation}
    \begin{aligned}
        \mathcal{G}_{\textbf{z}_{*}(a)}\Bigl(m(a)\big|\boldsymbol{\ell}(a),\textbf{m}(a^{-}), \boldsymbol{\ell}_{0} \Bigr)
         = &  F\Bigl(M_{\textbf{z}_{*}(a)}(a) = m(a)\big|\textbf{L}_{\textbf{z}_{*}(a)}(a) = \boldsymbol{\ell}(a),
         \textbf{M}_{\textbf{z}_{*}(a^{-})}(a^{-}) = \textbf{m}(a^{-}),\textbf{L}_{0} =  \boldsymbol{\ell}_{0}\Bigr)
    \end{aligned} 
\end{equation}
denote the conditional distribution of mediators in the world where the exposure is set to $\textbf{Z}(a) = \textbf{z}_{*}(a)$. This conditional distribution provides a random draw $M(a) \sim \mathcal{G}_{\textbf{z}_{*}(a)}\Bigl(m(a)\big|\boldsymbol{\ell}(a),\textbf{m}(a^{-}), \boldsymbol{\ell}_{0} \Bigr)$ within each stratum $\bigl( \boldsymbol{\ell}(a),\textbf{m}(a^{-}), \boldsymbol{\ell}_{0}\bigr)$. Throughout this article, we will denote \\ $\mathcal{G}_{z_{*}(a)} = \mathcal{G}_{z_{*}(a)}\Bigl(m(a)\big|\boldsymbol{\ell}(a),\textbf{m}(a^{-}), \boldsymbol{\ell}_{0} \Bigr) $ and $\mathbfcal{G}_{\textbf{z}_{*}(a)} =  \mathcal{G}_{\textbf{z}_{*}([0:a])}$ for convenience.

\noindent Now, consider a joint intervention to (statically) set $\textbf{Z}(a) = \textbf{z}(a)$ in the population and randomly draw $\textbf{M}(a)$ from  $\mathcal{G}_{\textbf{z}_{*}(a)}$ defined above. Resulting from this intervention,  the counterfactual time-varying confounder  at age $a$ is defined as $L_{(\textbf{z}(a), \textbf{m}(a))}
(a) $. Similarly, we define $T_{(\textbf{z}(a), \textbf{m}(a))}$ to be the potential time to event under the exposure regime $\textbf{z}(a)$, and mediator history $\textbf{m}(a)$. Using this, we define the potential survival function at age $a$ conditional on no previous event under the exposure process $\textbf{z}(a)$ and the mediator process taking the value as if the individual is under the regime $\textbf{z}_{*}(a)$ $\Bigl(\textbf{m}(a) \sim \mathbfcal{G}_{\textbf{z}_{*}(a)}\Bigr)$, as $\mathcal{S}_{\textbf{z},\textbf{z}_{*}}(a) = P\bigl(T_{i,(\textbf{z}(a), \textbf{m}(a))} > a\bigr)$.

\subsection{Causal estimands}\label{Sec42}
We define the Interventional Direct Effects (IDE) and the  Interventional Indirect Effects (IIE) as $IDE(a)  = \mathcal{S}_{\textbf{z},\textbf{z}_{*}}(a) - \mathcal{S}_{\textbf{z}_{*},\textbf{z}_{*}}(a)$ and $IIE(a) = \mathcal{S}_{\textbf{z},\textbf{z}}(a) - \mathcal{S}_{\textbf{z},\textbf{z}_{*}}(a),$ respectively. This implies that for an individual at age $a$, the interventional direct and indirect effects are given by:
\begin{equation}    
    \begin{aligned}
    IDE(a) =&  P\big[T_{i,(\textbf{z}(a), \mathbfcal{G}_{\textbf{z}_{*}(a)})} > a\big] -  P\big[T_{i,(\textbf{z}_{*}(a), \mathbfcal{G}_{\textbf{z}_{*}(a)})} > a\big] \text{ and }
    \\
    IIE(a) =&  P\big[T_{i,(\textbf{z}(a), \mathbfcal{G}_{\textbf{z}(a)})} > a\big] - 
    P\big[T_{i,(\textbf{z}(a), \mathbfcal{G}_{\textbf{z}_{*}(a)})} > a\big].
    \end{aligned}
\end{equation}
These two contrasts decompose the total effects $(TE)$ of the exposure as $TE(a) = IDE(a) + IIE(a)= \mathcal{S}_{\textbf{z},\textbf{z}}(a) - \mathcal{S}_{\textbf{z}_{*},\textbf{z}_{*}}(a) 
    =  P\big[T_{i,(\textbf{z}(a), \mathbfcal{G}_{\textbf{z}(a)})} > a\big] - P\big[T_{i,(\textbf{z}_{*}(a), \mathbfcal{G}_{\textbf{z}_{*}(a)})} > a\big].$

\noindent In our application,  $\mathcal{S}_{\textbf{z},\textbf{z}_{*}}(a)$ represents the probability of surviving beyond age $a$ under a hypothetical intervention in which an individual's BP medication history is set to $\textbf{z}(a)$, while their CVD risk factor trajectory (the mediator) evolves as if they had followed the reference medication regime $\textbf{z}_{*}(a)$. At age $a$, $IDE(a)$ measures the interventional direct effect: the change in survival probability at age $a$ attributable to switching medication status from $\textbf{z}_{*}(a)$ to $\textbf{z}(a)$ through pathways that do not operate via CVD risk factors. This captures direct pharmacological effects of antihypertensive medications on mortality risk independent of BP reduction.

\noindent In contrast, the interventional indirect effect, $IIE(a)$, measures the change in survival probability attributable to medication-induced modifications of the CVD risk factor trajectory, holding the direct medication effect fixed at $\textbf{z}(a)$. This represents the portion of the treatment effect operating through the hypothesized biological mechanism: medication alters CVD risk factors (e.g., mean BP), which in turn affect survival. The total effect, $TE(a) = IDE(a) + IIE(a)$, represents the overall difference in survival probability comparing the treatment regime $\textbf{z}(a)$ to the reference regime $\textbf{z}_{*}(a)$, integrating both direct and mediated pathways.

\subsection{Identification assumptions}\label{Sec43}
\begin{assumption}\label{Assump1:Positivity}
    Positivity: \\
    For all $\boldsymbol{\ell}_0$, $\textbf{z}(a)$, $\boldsymbol{\ell}(a)$, and $\textbf{m}(a)$, we assume the following conditions hold:
    \begin{enumerate}[itemsep=0pt, parsep=0pt, topsep=2pt, partopsep=0pt]
        \item $f_{Z(a)|\textbf{Z}(a^{-}), \textbf{L}(a^{-}),\textbf{M}(a^{-}), L_0}\big(z(a)\big|\textbf{z}(a^{-}), \boldsymbol{\ell}(a^{-}),\textbf{m}(a^{-}),\ell_0\big) > 0$
        \item $f_{Z_{*}(a)|\textbf{Z}_{*}(a^{-}), \textbf{L}(a^{-}),\textbf{M}(a^{-}), L_0}\big(z_{*}(a)\big|\textbf{z}_{*}(a^{-}), \boldsymbol{\ell}(a^{-}),\textbf{m}(a^{-}),\ell_0\big) > 0$
        \item $f_{L(a)|\textbf{Z}(a), \textbf{L}(a^{-}),\textbf{M}(a^{-}), L_0}\big(\ell(a)\big|\textbf{z}(a), \boldsymbol{\ell}(a^{-}),\textbf{m}(a^{-}),\ell_0\big) > 0$
        \item $f_{M(a)|\textbf{Z}_{*}(a), \textbf{L}(a),\textbf{M}(a^{-}), L_0}\big(m(a)\big|\textbf{z}_{*}(a), \boldsymbol{\ell}(a),\textbf{m}(a^{-}),\ell_0\big) > 0$
        \item $f_{T_i, \delta_i|\textbf{Z}(a), \textbf{L}(a),\textbf{M}(a), L_0}\big(t_i, \delta_i \big|\textbf{z}(a), \boldsymbol{\ell}(a),\textbf{m}(a),\ell_0\big) > 0$
    \end{enumerate}
    where $\textbf{z}(a^{-})$, $\boldsymbol{\ell}(a^{-})$, and $\textbf{m}(a^{-})$ denote the exposure, confounder, and mediator processes up to, but not including, age $a$.
\end{assumption}

\noindent Assumptions \ref{Assump1:Positivity}.1. and \ref{Assump1:Positivity}.2. require exposures to be observed in each confounder and mediator stratum. Assumptions \ref{Assump1:Positivity}.3. and \ref{Assump1:Positivity}.4. require longitudinal confounders and mediators to be supported under $\textbf{z}(a)$ and $\textbf{z}_{*}(a)$, respectively. Similarly, Assumption \ref{Assump1:Positivity}.5. requires the survival outcomes to be observed in each exposure, confounder, and mediator stratum.
These positivity assumptions are necessary to ensure that the conditional densities used for the identification of the causal parameters are well-defined.

\begin{assumption}\label{Assump2:Consistency}
    Consistency:
    \begin{enumerate}[label=(\roman*), itemsep=0pt, parsep=0pt, topsep=2pt, partopsep=0pt]
            \item \label{Assump:2.(i)}  $L_{i,(\textbf{z}_{*}(a))}(a) = L_{i}(a)$, $M_{i,(\textbf{z}_{*}(a))}(a) = M_{i}(a)$, and $T_{i,(\textbf{z}_{*}(a))}(a) = T_i$ given $\textbf{Z}_{i}(a) = \textbf{z}_{*}(a)$,
            \item \label{Assump:2.(ii)} $L_{i,(\textbf{z}(a), \textbf{m}(a^{-}))}(a) = L_{i}(a)$ given $\textbf{Z}_{i}(a) = \textbf{z}(a)$, and $\textbf{M}_{i}(a^{-}) = \textbf{m}(a^{-})$,
            \item \label{Assump:2.(iii)} $T_{i,(\textbf{z}(a), \textbf{m}(a))} = T_i$ given $\textbf{Z}_{i}(a) = \textbf{z}(a)$, and $\textbf{M}_{i}(a) = \textbf{m}(a)$,
    \end{enumerate}
\end{assumption}
where $\textbf{m}_{i}(a^{-})$ denotes the mediator process up to, but not including, age $a$.

\begin{assumption}\label{Assump3:TrtIgnorability}
    Ignorability: \\
    Let $\textbf{z}(a^{+})$ and $\textbf{z}_{*}(a^{+})$ denote the future exposure processes in treatment and control regimes, respectively. There exists $\epsilon > 0$ such that for any $0 < \Delta < \epsilon$, with $a, a + \Delta \in [0, \mathcal{A}]$, the change in the treatment process from age $a$ to $a + \Delta$ is independent of the subsequent counterfactual variables conditional on the observed baseline confounders, and the treatment, confounder, and mediator processes up to age $a$:
    \begin{enumerate}[label=(\roman*), itemsep=0pt, parsep=0pt, topsep=2pt, partopsep=0pt]
        \item   $\big\{\textbf{L}_{i,(\textbf{z}_{*}(a^{+}))},\textbf{M}_{i,(\textbf{z}_{*}(a^{+}))},T_{i,(\textbf{z}_{*}(a^{+}))} \big\} \indep \big(Z_{i,*}(a+ \Delta) - Z_{i,*}(a)\big) \big| \textbf{M}(a), \textbf{L}_{i}(a),\textbf{Z}_{i,*}(a),\textbf{L}_0, $ and
        \item  $\big\{\textbf{L}_{i,(\textbf{z}(a^{+}), \textbf{m}(a^{+}))},\textbf{M}_{i,(\textbf{z}(a^{+}), \textbf{m}(a^{+}))},T_{i,(\textbf{z}(a^{+}), \textbf{m}(a^{+}))}\big\} \indep \big(Z_{i}(a+ \Delta) - Z_{i}(a)\big) \big| \textbf{M}_{i}(a), \textbf{L}_{i}(a),\textbf{Z}_{i}(a),\textbf{L}_{i,0}$.
    \end{enumerate}
\end{assumption}

\begin{assumption}\label{Assump4:MediatorIgnorability}
    Sequential Ignorability: \\
    Let $\textbf{m}(a^{+})$ denote the future mediator process. There exists $\epsilon > 0$ such that for any $0 < \Delta < \epsilon$, with $a, a + \Delta \in [0, \mathcal{A}]$, the change in the mediator process from age $a$ to $a + \Delta$ is independent of the subsequent counterfactual variables conditional on the observed baseline confounders, and the treatment, confounder, and mediator processes up to age $a$: \\
        $\big\{L_{i,(\textbf{z}(a^{+}), \textbf{m}(a^{+}))},T_{i,(\textbf{z}(a^{+}), \textbf{m}(a^{+}))}\big\} \indep \big(M_{i}(a+ \Delta) - M_{i}(a)\big) \big| \textbf{M}_{i}(a), \textbf{L}_{i}(a),\textbf{Z}_{i}(a),\textbf{L}_{i,0}$.
\end{assumption}

\noindent Assumption \ref{Assump3:TrtIgnorability} states that in a sufficiently small time interval, there are no unmeasured confounders in the relationship between the treatment assignment and all subsequent counterfactuals, conditional on the predictor history. Similarly, Assumption \ref{Assump4:MediatorIgnorability} asserts an equivalent no-unmeasured-confounding condition for mediators in a sufficiently small time interval, given the predictor history.

\noindent Assumptions \ref{Assump3:TrtIgnorability} and \ref{Assump4:MediatorIgnorability} are strong and may be violated in real-world applications. Figure~\ref{Fig:DAG2} presents a DAG illustrating a potential violation of Assumption~\ref{Assump3:TrtIgnorability}. The red arrows starting from the unmeasured confounder, denoted by $U_i$, indicate that changes in the exposure process over an age window between two closely spaced time points, $ a $ and $ a + \Delta $, are not independent of the potential survival time, even after conditioning on baseline covariates. A similar illustration of a potential violation of Assumption~\ref{Assump4:MediatorIgnorability} can be generated by redirecting the red arrows in Figure~\ref{Fig:DAG2} from pointing to $ Z(a) $ and $ Z(a + \Delta) $ to instead point to $ M(a) $ and $ M(a + \Delta) $.

\begin{figure}[H]
\singlespacing
\centering
\begin{tikzpicture}[
    every node/.style={circle, draw=black, thick, minimum size=1.4cm, font=\small, align=center},
    arrow/.style={->, thick, >=Stealth}
]

% Nodes on the horizontal axis
\node (L0) at (-4.5,0) {$L\textsubscript{i,0}$};
\node (Za) at (-1.5,0) {$Z_i(a)$};
\node (Ma) at (1,0) {$M_i(a)$};
\node (Za1) at (3.5,0) {$Z_i(a + \Delta)$};
\node (Ma1) at (6,0) {$M_i(a + \Delta)$};
\node (Ti) at (10,0) {$(T_i,\delta_i)$};

% Node above Ma and Ti
\node (La) at (2,3) {$L_i(a)$};
\node (La1) at (7,3) {$L_i(a + \Delta)$};

% Node Ui
\node (Ui) at (-2.5,3) {$U_i$};
\draw[arrow, color = red] (Ui) to (Za);
\draw[arrow, color = red, bend left=25] (Ui) to (Za1);
\draw[arrow, color = red, bend left=60] (Ui) to (Ti);

% Curved and straight short arrows from Ui

% Curved and straight short arrows from L0
\draw[arrow, bend left=30] (L0) to +(1.5,1);  % upward-right
\draw[arrow, bend right=30] (L0) to +(+1.5,-1); % downward-right
\draw[arrow] (L0) -- +(1.5,0);                 % straight right

% Straight short arrows from Ma1
\draw[arrow] (Ma1) -- +(1.5,0);
\draw[arrow, bend left=-35] (Ma1) to (Ti);
% Dots after Ma1
\node[draw=none, fill=none, inner sep=0pt] at (7.75, 0) {$\cdots$};
\draw[arrow] (8, 0) -- (Ti);
% Dots before Z(a)
\node[draw=none, fill=none, inner sep=0pt] at (-2.5, 0) {$\cdots$};

% Horizontal arrows
% Arrows from M(a) to ALL future covariates
\draw[arrow] (Ma) -- (Za1);
\draw[arrow, bend left=25] (Ma) to (La1);
\draw[arrow, bend left=-35] (Ma) to (Ma1);
\draw[arrow, bend left=-35] (Ma) to (Ti);
% Arrows from L(a) to ALL future covariates

\draw[arrow] (La) -- (Ma);
\draw[arrow] (La) -- (Za1);
\draw[arrow] (La) -- (Ma1);
\draw[arrow] (La) -- (Ti);
\draw[arrow] (La) -- (La1);
% Arrows from La1 to ALL future covariates

\draw[arrow] (La1) -- (Ma1);
\draw[arrow] (La1) -- (Ti);

% Curved arrow from Z(a) ALL future covariates
\draw[arrow] (Za) -- (Ma);
\draw[arrow,  bend left=25] (Za) to (La);
\draw[arrow,  bend left=60] (Za) to (La1);
\draw[arrow, bend left=-35] (Za) to (Za1);
\draw[arrow, bend left=-35] (Za) to (Ma1);
\draw[arrow, bend left=-35] (Za) to (Ti);

% Curved arrow from Za1 ALL future covariates
\draw[arrow] (Za1) -- (Ma1);
\draw[arrow,  bend left=25] (Za1) to (La1);
\draw[arrow, bend left=-35] (Za1) to (Ti);

\end{tikzpicture}
\caption{A DAG illustrating a potential violation of Assumption \ref{Assump3:TrtIgnorability}, where an unmeasured confounder
$U_i$ influences both the longitudinal exposures and the survival outcomes.}
\label{Fig:DAG2}
\end{figure}

\noindent Assumptions \ref{Assump3:TrtIgnorability} and \ref{Assump4:MediatorIgnorability} involve cross-world counterfactuals and are therefore generally untestable, but they are essential for identification and widely used in causal mediation analysis. Our formulation, defined over small time intervals to accommodate continuous-time data, is adapted from the sequential ignorability assumption in \cite{zeng2022causal}. However, unlike their setting with a point exposure and a continuous-time mediator, our framework allows for both exposure and mediator to vary continuously over time.

\noindent In our motivating example, Assumption \ref{Assump3:TrtIgnorability} requires that for the individual, the history of BP medication, time-varying (and baseline) predictors, and CVD risk factor measurements fully accounts for confounding in the relationship between change in medication status within a sufficiently small age interval and subsequent time-varying predictors, CVD risk factor measurements, and survival outcomes. Likewise, Assumption \ref{Assump4:MediatorIgnorability} requires that the history of medication, time-varying (and baseline) predictors, and CVD risk factors, fully accounts for confounding in the relationship between change in CVD risk factor measurements within a sufficiently small age interval and subsequent predictors and survival outcomes.

\begin{assumption}\label{Assump5:IndepCensoring}
    Independent Censoring: \\
    The censoring time is independent of all variables, including baseline confounders, time-varying treatments, confounders, and mediators, and the time-to-event outcome: \\
    $C_i \indep \big\{ T_i, \textbf{M}_{i}(a), \textbf{L}_{i}(a),\textbf{Z}_{i}(a),\textbf{L}_{i,0} \big\},$  for any $a \in [0, \mathcal{A}]$.
\end{assumption}

\noindent Although the independent censoring assumption is strong, it is reasonable in our motivating example. In the ARIC cohort study, death data are updated from the National Death Index \citep{NCHS_NDI} every few years. At each data update, the date of the most recent recorded death serves as a proxy censoring time for individuals who are still alive at that point. Importantly, because this form of censoring does not arise from loss to follow-up or participant-specific reasons, it can be viewed as administrative censoring. In other words, people who are alive on the date of the most recent recorded death are censored not because of dropout but because the outcome assessment ceases at that point. Therefore, the assumption of independent censoring is reasonable.

\subsection{Identification}\label{Sec44}
\begin{prop}
    Let $a_1 < a_2 < \ldots < a_{k} < \ldots < a_{K} < \ldots$ be the age grid where risk factor updates are possible, and at least one event takes place within each interval  $(a_{k-1}, a_{k}]$. Consider a fixed age $a$ such that $a_{K} \leq a \leq a_{K+1}$. Under the assumptions in the previous sub-section, the $IDE, IIE, \text{ and } TE$ can be identified from the observed data using:
   \begin{equation}\label{Eq:NonParamIdentification}
       \begin{aligned}
        \mathcal{S}_{\textbf{z},\textbf{z}_{*}}(a) 
        & =  \int_{\textbf{m}_i(a)} \int_{\boldsymbol{\ell}_i(a)} \int_{\boldsymbol{\ell}_{0}} \int_{\textbf{b}_i}  \\& 
        \prod_{k = 1}^{K} \bigg\{ Pr \Bigl[\tilde{T}_i> a_{k} \big|  \tilde{T}_i \geq a_{k-1},\textbf{Z}_i(a_{k}) = \textbf{z}_i(a_{k}), \textbf{L}_i(a_{k}) = \boldsymbol{\ell}_i(a_{k}),   \textbf{M}_i(a_{k}) = \textbf{m}_i(a_{k}),   
        \textbf{L}_{i,0}= \boldsymbol{\ell}_{i,0} \Bigr] 
        \times \\& 
         f\Bigl(\mathbf{M}_i(a_{k}) = \textbf{m}_i(a_{k})\big| \tilde{T}_i \geq a_{k-1},   \textbf{Z}_i(a_{k}) = \textbf{z}_{i,*}(a_{k}),
         \textbf{L}(a_{k}) =
         \boldsymbol{\ell}_i(a_{k}), \mathbf{M}_i(a_{k-1}) = \textbf{m}_i(a_{k-1}), \textbf{L}_{i,0}= \boldsymbol{\ell}_{i,0}, b_i^M \Bigr)
          \times \\&
        f\Bigl(\mathbf{L}_i(a_{k}) = \boldsymbol{\ell}_i(a_{k})\big|\tilde{T}_i \geq a_{k-1}, \textbf{Z}_i(a_{k}) = \textbf{z}_{i}(a_{k}), \textbf{L}(a_{k-1}) =
         \boldsymbol{\ell}_i(a_{k-1}), \mathbf{M}_i(a_{k-1}) = \textbf{m}_i(a_{k-1}), \textbf{L}_{i,0}= \boldsymbol{\ell}_{i,0}, \\& b_i^L\Bigr) \bigg\}
         \times 
         f(\textbf{L}_{0}) \text{ }
         d\textbf{m}_i(a)
         d\boldsymbol{\ell}_i(a) d\boldsymbol{\ell}_{i,0} d\textbf{b}_i
    \end{aligned}
   \end{equation}
 where, for any random variables $A$ and $B$, $f_{A}(.)$ and  $f_{A|B}(.)$ denote the density function of $A$ and the conditional density of $A$ given $B$, respectively. 
\end{prop}
\begin{proof}
    See Appendix A.
\end{proof}

\section{Model specification and estimation}\label{Sec5}
\subsection{Joint model for the outcome, mediator, confounder, and exposure}\label{Sec52}

To estimate the causal effects introduced in Section \ref{Sec42}, we need to estimate the joint distribution $F\big(T_i, \delta_i, \textbf{M}_{i}, \textbf{L}_{i}, \textbf{Z}_{i}, \textbf{L}_{i,0}  \big)$. We propose an enriched Dirichlet process mixture (EDPM) model \citep{wade2011enriched, wade2014improving}, which induces a nested clustering structure. At the outer cluster, subjects are grouped according to shared survival model parameters, while within each outer cluster, subjects are further partitioned into inner clusters that share parameters governing the longitudinal exposure, confounder, mediator, and baseline covariate processes. The resulting hierarchical model is
\begin{equation}\label{eq:EDPMmodel}
\resizebox{\textwidth}{!}{$
    \begin{aligned}
       T_i, \delta_i \big| T_i > a_{i,n_i}, \textbf{M}_{i}(a_{i,n_i}), \textbf{L}_{i}(a_{i,n_i}), \textbf{Z}_{i}(a_{i,n_i}), \textbf{L}_{i,0}; \boldsymbol{\beta}_i &\sim F_{t_i,\delta_i}\big(. \big| \textbf{m}_{i}(a_{i,n_i}), \boldsymbol{\ell}_{i}(a_{i,n_i}), \textbf{z}_{i}(a_{i,n_i}), \boldsymbol{\ell}_{i,0}; \boldsymbol{\beta}_i \big)
       \\
        M_{i}(a_{i,j})\big| T_i> a_{i,j}, \textbf{M}_{i}(a_{i,j-1}), \textbf{L}_{i} (a_{i,j}), \textbf{Z}_{i}(a_{i,j}), \textbf{L}_{i,0}, b_i^{M}; \boldsymbol{\theta}^{M}_{i}, \boldsymbol{\eta}^{M}_{i} &\sim F_{m(a_{i,j})} \big(.\big|a_{i,j}, \textbf{m}_{i}(a_{i,j-1}), \boldsymbol{\ell}_{i} (a_{i,j}), \textbf{z}_{i}(a_{i,j}), \boldsymbol{\ell}_{i,0},b_i^{M}; \boldsymbol{\theta}^{M}_{i}, \boldsymbol{\eta}^{M}_{i}\big)
        \\
        L_{i}(a_{i,j})\big|T_i> a_{i,j}, \textbf{M}_{i}(a_{i,j-1}), \textbf{L}_{i} (a_{i,j-1}), \textbf{Z}_{i}(a_{i,j}), \textbf{L}_{i,0},b_i^{L}; \boldsymbol{\theta}^{L}_{i}, \boldsymbol{\eta}^{L}_{i} &\sim F_{\ell(a_{i,j})} \big(.\big|a_{i,j}, \textbf{m}_{i}(a_{i,j-1}), \boldsymbol{\ell}_{i} (a_{i,j-1}), \textbf{z}_{i}(a_{i,j}), \boldsymbol{\ell}_{i,0},b_i^{L}; \boldsymbol{\theta}^{L}_{i}, \boldsymbol{\eta}^{L}_{i}\big)
        \\
        Z_{i}(a_{i,j})\big|T_i> a_{i,j}, \textbf{M}_{i}(a_{i,j-1}), \textbf{L}_{i} (a_{i,j-1}), \textbf{Z}_{i}(a_{i,j-1}), \textbf{L}_{i,0}, b_i^{Z}; \boldsymbol{\theta}^{Z}_{i}, \boldsymbol{\eta}^{Z}_{i} &\sim F_{z(a_{i,j})} \big(.\big|a_{i,j}, \textbf{m}_{i}(a_{i,j-1}), \boldsymbol{\ell}_{i} (a_{i,j-1}), \textbf{z}_{i}(a_{i,j-1}), \boldsymbol{\ell}_{i,0}, b_i^{Z}; \boldsymbol{\theta}^{Z}_{i}, \boldsymbol{\eta}^{Z}_{i}\big)
        \\ \textbf{L}_{i,0};\boldsymbol{\theta}^{\textbf{L}_0}_{i} &\sim F_{\boldsymbol{\ell}_{0}}\big(.\big| \boldsymbol{\theta}^{\textbf{L}_0}_{i} \big), \\
        (\boldsymbol{\beta}_i, \boldsymbol{\theta}_{i})|H &\sim H  \\
        H &\sim EDP(\alpha^{\beta}, \alpha^{\theta|\beta}, H_0),
    \end{aligned}
$}
\end{equation}
where, for each subject $i$,  $\boldsymbol{\theta}_{i} = \big(\boldsymbol{\theta}^{M}_{i}, \boldsymbol{\eta}^{M}_{i}, \boldsymbol{\theta}^{L}_{i}, \boldsymbol{\eta}^{L}_{i}, \boldsymbol{\theta}^{Z}_{i}, \boldsymbol{\eta}^{Z}_{i},   \boldsymbol{\theta}^{\textbf{L}_0}_{i}\big)$. The regression coefficients $\boldsymbol{\theta}^{M}_i$, $\boldsymbol{\theta}^{L}_i$, and $\boldsymbol{\theta}^{Z}_i$ are shared across all visit ages within each inner ($\theta$-level) cluster. Age-dependent variation is captured through the spline coefficients $\boldsymbol{\eta}^{M}_i$, $\boldsymbol{\eta}^{L}_i$, and $\boldsymbol{\eta}^{Z}_i$,  which are also cluster-specific. Their explicit formulation is provided in~\eqref{eq:Localmodels}. We assume that the components of the baseline covariate vector $\textbf{L}_{i,0}$ are locally independent within clusters. 

\noindent The EDPM has two sets of concentration parameters: the first consists of a single parameter, $\alpha^{\beta}$, corresponding to outer ($\beta$-level) clusters, and the second consists of multiple parameters, $\alpha^{\theta|\beta}$, corresponding to inner ($\theta$-level) clusters within each $\beta$-level cluster.  The number of $\beta$-level clusters and the number of $\theta$-level clusters within each $\beta$-level cluster depend on $\alpha^{\beta}$ and $\alpha^{\theta|\beta}$, respectively. The lower values of the concentration parameters indicate fewer clusters at both levels. By the square-breaking construction of the EDP, if $H \sim EDP(\alpha^{\beta}, \alpha^{\theta|\beta}, H_0)$, then
\[
H = \sum_{r=1}^{\infty} \sum_{s=1}^{\infty} \gamma_{r} \gamma_{s|r} \, \delta_{\beta_r^{*},\theta_{s|r}^{*}},
\]
where, at the $\beta$-level, the weights are given by
$\gamma_r = \gamma_r^{'}\prod_{t < r} \bigl(1 - \gamma_t^{'}\bigr), \quad \gamma_t^{'} \sim \text{Beta}(1, \alpha^{\beta}), \quad \beta^{*}_r \overset{iid}{\sim} H_{0\beta},$
and, conditionally within each $\beta$-level cluster, the weights are
$\gamma_{s|r} = \gamma_{s|r}^{'}\prod_{t < s} \bigl(1 - \gamma_{t|r}^{'}\bigr), \quad \gamma_{t|r}^{'} \sim \text{Beta}(1, \alpha^{\theta|\beta}), \quad \theta^{*}_{s|r} \overset{iid}{\sim} H_{0\theta|\beta}.$ Following \citet{burns2023truncation}, we implement the EDPM in this article via a truncation approximation. Complete details of the approximation and the associated posterior computation are provided in Appendix~\ref{AppendixC}.
    
\noindent For notational clarity, let $t_i$ denote a random sample from the outcome distribution and $t$ denote a fixed age of interest. Within each $\beta-$level cluster, we assume a Cox regression model: 
\begin{equation}\label{Eq:CoxModel}
    \lambda_i\bigl(t\big|\boldsymbol{\ell}_{i,0}, \boldsymbol{\beta}_i\bigr) = \lambda_0\bigl(t\bigr)exp\bigl\{ \boldsymbol{\beta}_i \boldsymbol{\ell}_{i,0}^\top \bigr\},
\end{equation}
where the baseline hazard, $\lambda_0\bigl(t\bigr)$, is assumed to be piecewise constant over a partition composed of $B$ disjoint intervals. This specification yields a piecewise exponential model for the survival outcome, which can be conveniently approximated using a Poisson model for computational efficiency (see Chapter 13 of \cite{christensen2010bayesian}). Further details on the Poisson approximation are provided in Appendix D.

\noindent We model the observed time-varying variables $Z_i(a_{ij})$, $L_i(a_{ij})$, and $M_i(a_{ij})$ using penalized thin-plate splines with $D$ pre-specified knots at $q_1 \leq \ldots \leq q_D$, following \cite{zeldow2021functional}, to accommodate nonlinearities across age and enable prediction at any age $a$. We denote by $\mathbfcal{B}_i$ the resulting $n_i \times D$ basis matrix for subject $i$ and by $\mathbfcal{B}_i(a_{ij})$ its $j^{\text{th}}$ row, which we reference in equation~\eqref{eq:Localmodels} below. Full construction of $\mathbfcal{B}_i$ from the cubic basis functions $\Bigl\{|a_{ij} - q_d|^3\Bigr\}_{d=1}^{D}$ and the associated penalty matrix $\Omega_D$ is provided in Appendix~\ref{AppendixD}. Specifically, we assume the following generalized linear models (GLMs) for 
$F_{m(a_{i,j})} \big(.\big)$, 
$F_{\ell(a_{i,j})} \big(.\big)$, and $F_{z(a_{i,j})} \big(.\big)$ for all $j = 1, \ldots, n_i$:
\begin{equation}\label{eq:Localmodels}
     \begin{aligned}
         M_{i}(a_{i,j})\big|\boldsymbol{\ell}_{i,0}, b_i^{M}; \boldsymbol{\theta}^{M}_i, \boldsymbol{\eta}^{M}_i &\sim N \big(\boldsymbol{\theta}^{M}_i\boldsymbol{\ell}^\top_{i,0}  + b_i^M + \boldsymbol{\eta}^{M}_i \mathbfcal{B}_{i}(a_{i,j})^\top, \sigma_{M}^{2} \big),
        \\
        L_{i}(a_{i,j})\big|\boldsymbol{\ell}_{i,0},  b_i^{L}; \boldsymbol{\theta}^{L}_i, \boldsymbol{\eta}^{L}_i &\sim N \big(\boldsymbol{\theta}^{L}_i\boldsymbol{\ell}^\top_{i,0}  + b_i^L + \boldsymbol{\eta}^{L}_i \mathbfcal{B}_{i}(a_{i,j})^\top, \sigma_{L}^{2} \big),
        \\
      Z_{i}(a_{i,j})\big|\boldsymbol{\ell}_{i,0},  b_i^{Z}; \boldsymbol{\theta}^{Z}_i, \boldsymbol{\eta}^{Z}_i &\sim Bern\bigl(p_{i,j}^{Z}\bigr), \quad  probit\big(p_{i,j}^{Z}\big) = \boldsymbol{\theta}^{Z}_i\boldsymbol{\ell}^\top_{i,0}  + b_i^Z + \boldsymbol{\eta}^{Z}_i \mathbfcal{B}_{i}(a_{i,j})^\top,
    \end{aligned}
\end{equation}
where $\boldsymbol{\eta}^{M}_i  \coloneqq \bigl( \eta^{M}_{i,1}, \ldots, \eta^{M}_{i,D}\bigr)$, $\boldsymbol{\eta}^{L}_i \coloneqq \bigl( \eta^{L}_{i,1}, \ldots, \eta^{L}_{i,D}\bigr)$, and $\boldsymbol{\eta}^{Z}_i \coloneqq \bigl( \eta^{Z}_{i,1}, \ldots, \eta^{Z}_{i,D}\bigr)$ are coefficients on the spline parameters for subject $i$ corresponding to each of the $D$ spline knots, and $\mathbfcal{B}_{i}(a_{i,j})$ is the $1\times D$ row of the spline basis matrix evaluated at visit age $a_{i,j}$ as defined in Appendix~\ref{AppendixD}. For binary time-varying confounders and mediators, we modify their respective models in (\ref{eq:Localmodels}) to:
\begin{align*}
    M_{i}(a_{i,j})\big|\boldsymbol{\ell}_{i,0},  b_i^{M}; \boldsymbol{\theta}^{M}_i, \boldsymbol{\eta}^{M}_i &\sim Bern\bigl(p_{i,j}^{M}\bigr), \quad  probit\big(p_{i,j}^{M}\big) = \boldsymbol{\theta}^{M}_i\boldsymbol{\ell}^\top_{i,0}  + b_i^M + \boldsymbol{\eta}^{M}_i \mathbfcal{B}_{i}(a_{i,j})^\top, \\
    L_{i}(a_{i,j})\big|\boldsymbol{\ell}_{i,0},  b_i^{L}; \boldsymbol{\theta}^{L}_i, \boldsymbol{\eta}^{L}_i &\sim Bern\bigl(p_{i,j}^{L}\bigr), \quad  probit\big(p_{i,j}^{L}\big) = \boldsymbol{\theta}^{L}_i\boldsymbol{\ell}^\top_{i,0}  + b_i^L + \boldsymbol{\eta}^{L}_i \mathbfcal{B}_{i}(a_{i,j})^\top, \quad j = 1, \ldots, n_i.
\end{align*}

\noindent The random intercepts $b_i^{M}$,  $b_i^{L}$, and $b_i^{Z}$  are assumed to be normally distributed as $b_i^{M} \sim N\big(0,\tau_{M}^{2}\big)$, $b_i^{L} \sim N\big(0,\tau_{L}^{2}\big)$, and $b_i^{Z} \sim N\big(0,\tau_{Z}^{2}\big)$. These random effects are not included in the EDP prior, so they do not depend on clusters. The baseline covariates are assumed to be locally independent, specifying  Normal distributions for continuous baseline covariates, and  Bernoulli distributions for binary baseline covariates with their corresponding parameters. We provide further details on EDP base measures and prior specifications in Appendix B. Note that a key advantage of the EDPM model is that although parametric models are assumed within each cluster, the variables are globally dependent, allowing for complex non-linear relationships. 

\noindent In our ARIC application, the outer-level clustering captures unobserved heterogeneity in survival risk profiles: participants assigned to the same $\beta$-cluster share similar baseline hazard functions and covariate effects on CVD mortality, effectively identifying latent subpopulations with distinct survival characteristics (e.g., groups with inherently high versus low cardiovascular mortality risk that cannot be fully explained by observed predictors).

\noindent The inner-level clustering, nested within each $\beta$-cluster, captures heterogeneity in the longitudinal predictor dynamics as participants in the same $\theta$-cluster exhibit similar patterns of BP progression, smoking behavior changes, and medication adherence over time. This nested structure allows, for example, a subgroup of participants who share a common high-risk survival profile ($\beta$-cluster) to be further differentiated based on whether their BP trajectories are stable, increasing, or well-controlled by medication ($\theta$-clusters). The hierarchical EDP prior induces global dependence across all model components while permitting flexible, nonparametric heterogeneity. This specification accommodates the irregular observation times characteristic of the ARIC study.

\subsection{G-computation algorithm}\label{Sec6}
The integrals in Equation (\ref{Eq:NonParamIdentification}) are not available in closed form. To address this, we propose a Monte Carlo integration-based G-computation algorithm, outlined in Algorithm \ref{Alg:GcompAlg}. 

\noindent To carry out Algorithm \ref{Alg:GcompAlg}, we require the joint density of the observed data and several conditional densities (the survival function and the predictive distributions of the time-varying mediator and confounder at each age in the grid) implied by the EDPM in equation~\eqref{eq:EDPMmodel}. The full derivations of these quantities, including the truncation approximation and the explicit form of the cluster-membership weights, are provided in Appendix~\ref{AppendixE}. 

\begin{algorithm}[H]
\caption{G-computation algorithm}\label{Alg:GcompAlg}
\resizebox{!}{0.68\textwidth}{%
\begin{minipage}{\linewidth}
\singlespacing

\begin{enumerate}
    \item  Fit the observed data models as specified in Section \ref{Sec5}.  For each model, burn in the first $B$ MCMC iterations and keep the next $Q$ iterations (i.e. keep iterations $B+1,\ldots, B+Q$). For each of the  $Q$ iterations, we have a sample from the posterior of observed data parameters for $L(a_{ij}), M(a_{ij}), \text{ and } T_i$.

    \item At the  $(B+q)^{th}$ iteration, compute a posterior sample for $\mathcal{S}_{\textbf{z}_{1},\textbf{z}_{2}}(a)$ where $\bigl(\textbf{z}_{1},\textbf{z}_{2}\bigr) = \big\{(\textbf{z},\textbf{z}_{*}),  (\textbf{z}_{*},\textbf{z}_{*}), (\textbf{z},\textbf{z})\big\}$ as follows:
        \begin{enumerate}
            \item  Given the $q$-th posterior draws of the EDP weights $\xi_r^{q}, \xi_{s|r}^{q}$ and of the baseline-covariate cluster parameters $\theta^{\textbf{L}_0,q}_{s|r}$, randomly draw $C^{*}$ \big(row\big) vectors as samples for $\textbf{L}_{0}$, say $\boldsymbol{\tilde{\ell}}^{c}_0$, $c \in \big\{1,\ldots, C^{*} \big\}$, from the mixture:
            $$f_{\textbf{L}_0} \big(\boldsymbol{\ell}_0\big) =  \sum_{r=1}^{N}\xi_r^{q} \sum_{s=1}^{M} \xi_{s|r}^{q} \times p\bigl( \boldsymbol{\ell}_0;\theta^{\textbf{L}_0,q}_{s|r}\bigr).$$

             \item Conditional on the $\big(B+q\big)^{th}$ posterior samples of the random effect variances,  $\tau_{M}^{2,q}$, $\tau_{L}^{2,q}$, and $\tau_{Z}^{2,q}$, randomly draw $C^{*}$ observations—corresponding to $C^{*}$ Monte Carlo samples—from
             \[
               N\big(0, \tau_{M}^{2,q}\big), \quad N\big(0, \tau_{L}^{2,q}\big), \quad \text{and} \quad N\big(0, \tau_{Z}^{2,q}\big)
             \]
             for  
             $\big\{b^{M,q,c}\big\}_{c =1}^{C^{*}},\; \big\{b^{L,q,c}\big\}_{c =1}^{C^{*}},\; \text{and}\; \big\{b^{Z,q,c}\big\}_{c =1}^{C^{*}}$
             respectively. 

             \item Define binary survival indicator $\big(s^{c}_{a_k}\big)$ and cumulative survival probability $\big(S^{c}(a_k)\big)$ at $a_k$. Initialize $s^{c}_{a_0}$ and $S^{c}(a_0)$: for each Monte Carlo sample $c \in \{1, \ldots, C^{*}\}$, set $s^{c}_{a_0}  = 1$ and  $S^{c}(a_0) = 1$ (all subjects alive at the beginning). 

            \item Sequentially repeat the following $C^{*}$ times for all $k \in \{1, \ldots, K\}$:
                \begin{enumerate}
                    \item If $s^{c}_{a_{k-1}}  = 1$ (subject alive at previous age in the grid):
                    \begin{itemize}
                        \item For regime $\textbf{Z}(a_k) = \textbf{z}_1(a_k)$, sample $L(a_k)$:
                        \begin{align*}
                            & f_L\big(\ell \big| \textbf{z}_1(a_k),\tilde{\boldsymbol{\ell}}^{c}(a_{k-1}), \tilde{\textbf{m}}^{c}(a_{k-1}),\tilde{\boldsymbol{\ell}}_{0}^{c}, b^{M,q,c}, b^{L,q,c},b^{Z,q,c}; \beta^{q}, \theta^{q}\big)
                            \\=& \sum_{r=1}^{N} w^{c}_r\bigl(\textbf{z}_1(a_k),\tilde{\boldsymbol{\ell}}^{c}(a_{k-1}), \tilde{\textbf{m}}^{c}(a_{k-1}),\tilde{\boldsymbol{\ell}}_{0}^{c}\bigr)\times p\bigl(\ell(a_k)\big| \tilde{\boldsymbol{\ell}}_{0}^{c}; b^{L,q,c}, \theta^{L,q}_{s|r}\bigr).
                        \end{align*}

                        \item For regime $\textbf{Z}(a_k)=\textbf{z}_{2}(a_k)$, sample $M(a_k)$:
                        \begin{align*}
                            & f_M\big(m \big|\textbf{z}_2(a_k),\tilde{\boldsymbol{\ell}}^{c}(a_k), \tilde{\textbf{m}}^{c}(a_{k-1}),\tilde{\boldsymbol{\ell}}_{0}^{c}; b^{M,q,c},b^{L,q,c},b^{Z,q,c}, \beta^{q}, \theta^{q}\big)
                            \\=& \sum_{r=1}^{N} w^{c}_r\bigl(\textbf{z}_2(a_k),\tilde{\boldsymbol{\ell}}^{c}(a_k), \tilde{\textbf{m}}^{c}(a_{k-1}),\tilde{\boldsymbol{\ell}}_{0}^{c}\bigr)\times p\bigl(m(a_k)\big| \tilde{\boldsymbol{\ell}}_{0}^{c}; b^{M,q,c}, \theta^{M,q}_{s|r}\bigr).
                        \end{align*}
                    \item Otherwise ($s^{c}_{a_{k-1}}=0$), subject $c$ is dead. Set all future $S^{c}(a_{k})=0$.
                    \end{itemize}
                    Both displayed mixture forms above use the same outer-cluster collapse derived in Appendix~\ref{AppendixE}.
                    \item For a fixed regime $\textbf{Z}(a_{k}) = \textbf{z}_1(a_{k}) $, conditional on $\tilde{\boldsymbol{\ell}}_{0}^{c}$, $\tilde{\boldsymbol{\ell}}^{c}(a_{k})$ and $\tilde{\textbf{m}}^{c}(a_{k})$,   compute the interval-specific conditional probability of surviving from age $a_{k-1}$ to age $a_k$, say  $p^{c}_{a_{k-1} \to a_k}$, as defined in (\ref{Eq:GcompSurvfunc1}):
                    \begin{align*}
                        p^{c}_{a_{k-1} \to a_k} =&  Pr\bigl(T > a_k \big| T \geq a_{k-1}, \textbf{z}_1(a_{k}),\tilde{\boldsymbol{\ell}}^{c}(a_{k}), \tilde{\textbf{m}}^{c}(a_{k}),\tilde{\boldsymbol{\ell}}_{0}^{c}, b^{M,q,c},b^{L,q,c},b^{Z,q,c}; \beta^{q}, \theta^{q} \bigr).
                    \end{align*}

                    \item Update the cumulative survival probability:
                    $S^{c}(a_k) = S^{c}(a_{k-1}) \times p^{c}_{a_{k-1} \to a_k}.$
                    \item Generate $s^{c}_{a_k} = Bernoulli\big(p^{c}_{a_{k-1} \to a_k}\big)$.

                \end{enumerate}

            \item Compute:
            $\mathcal{S}_{\textbf{z}_1,\textbf{z}_2}^{ (B+q)}(a_K)
                = \frac{1}{C^{*}} \sum_{c=1}^{C^{*}} S^{c}(a_K).$
        \end{enumerate}

    \item Repeat Step 2 for all $Q$ posterior samples from Step 1.
\end{enumerate}

\end{minipage}%
}
\end{algorithm}

\section{Simulation study}\label{Sec7}

We conduct simulation studies to compare the proposed EDPM methodology against a parametric Bayesian latent class joint model (BLCJM) under two scenarios: (i) correct model specification, where both the data-generating process and the BLCJM share the same functional form, and (ii) model misspecification, where the data-generating process includes interaction and nonlinear terms, or different number of latent classes, that the BLCJM cannot capture. In both scenarios, the EDPM model serves as a flexible alternative that can adapt to complex data structures through its nonparametric mixture framework. Each scenario is evaluated at two sample sizes, $n \in \{1{,}500, 2{,}500\}$, across $R = 500$ simulation replicates. All simulations were implemented in R and NIMBLE \citep{de2017programming} using custom MCMC algorithms. 

\noindent To anchor the simulation studies in features of the motivating cohort study, we first fit a two-class parametric BLCJM to the ARIC hypertensive-at-baseline population. We then treat the resulting class-specific posterior mean estimates, including baseline covariate effects, regression and spline coefficients for the exposure, confounder, and mediator, and piecewise constant baseline hazards, as the data-generating "ground truth" for simulating replicated datasets. Subject-specific visit-age sequences are drawn from the empirical ARIC distribution via the Bayesian bootstrap \citep{rubin1981bayesian}, ensuring that the simulated datasets retain the irregular visit patterns observed in ARIC. We summarize the performance of each method using bias, mean squared error, $95\%$ coverage probability, and average credible interval width. Additional details on the BLCJM fit to ARIC, the visit-age sampling procedure, and so on, are provided in Appendix~\ref{AppendixF}.

\subsection{Study 1: correct model specification}\label{Sec73}

In this study, we generate data from the parametric BLCJM model described in Appendix~\ref{AppendixF}. This study serves two purposes: (i) to verify that the BLCJM recovers the true causal effects when correctly specified, and (ii) to assess whether the EDPM maintains comparable performance without substantial efficiency loss due to overfitting. Table~\ref{Tab:TrueEffects_Study1} reports the true causal effects computed under the correctly specified DGP.

\noindent Table~\ref{Tab:Study1_Results} summarizes the performance of the EDPM and BLCJM under the correctly specified DGP~(\ref{Eq:SimDGMMarginalJoint}). As expected, both approaches exhibit negligible bias across all estimands, ages, and sample sizes, with MSE effectively zero throughout. The BLCJM attains slight over-coverage (0.98--1.00) with slightly wider intervals. The EDPM shows relatively larger bias, particularly for the direct and total effects. Despite having comparable interval widths, this bias leads to slight undercoverage for some of the direct and total effects. 

\begin{table}[H]
\singlespacing
\centering
\caption{True causal effects (survival probability scale) under the correctly specified data-generating process. True IDE, IIE, and TE are reported as $\times 100$.}
\label{Tab:TrueEffects_Study1}
\small
\begin{tabular}{c|ccc|ccc}
\hline
Age & $\mathcal{S}^{\text{true}}_{\textbf{z}, \textbf{z}}(a)$ & $\mathcal{S}^{\text{true}}_{\textbf{z}_{*}, \textbf{z}_{*}}(a)$ & $\mathcal{S}^{\text{true}}_{\textbf{z}, \textbf{z}_{*}}(a)$ & True IDE ($\times 100$) & True IIE ($\times 100$) & True TE ($\times 100$) \\
\hline
65 & 0.7536 & 0.7586 & 0.7539 & $-$0.47 & $-$0.02 & $-$0.49 \\
75 & 0.6958 & 0.7014 & 0.6959 & $-$0.55 & $-$0.01 & $-$0.56 \\
\hline
\end{tabular}
\end{table}

\begin{table}[H]
\singlespacing
\centering
\caption{Study 1 results: Comparison of EDPM and BLCJM under correct model specification with $n \in \{1{,}500, 2{,}500\}$ and $R = 500$ replications. Bias and MSE are reported as $\times 100$.}
\label{Tab:Study1_Results}
\small
\begin{tabular}{cccc|cccc}
\hline
Estimand & $n$ & Age & Model & Bias ($\times 100$) & MSE ($\times 100$) & 95\% Coverage & Average Interval Width \\
\hline
IDE & 1{,}500
& 65 & EDPM &  0.31 & 0.00 & 0.93 & 0.02 \\
& &    & BLCJM &  0.13 & 0.00 & 0.99 & 0.02 \\
& & 75 & EDPM &  0.36 & 0.01 & 0.91 & 0.03 \\
& &    & BLCJM &  0.20 & 0.00 & 0.99 & 0.02 \\
\cline{2-8}
& 2{,}500
& 65 & EDPM &  0.29 & 0.00 & 0.92 & 0.02 \\
& &    & BLCJM &  0.13 & 0.00 & 0.99 & 0.02 \\
& & 75 & EDPM &  0.35 & 0.01 & 0.88 & 0.03 \\
& &    & BLCJM &  0.22 & 0.00 & 0.99 & 0.02 \\
\hline
IIE & 1{,}500
& 65 & EDPM &  0.06 & 0.00 & 1.00 & 0.02 \\
& &    & BLCJM &  0.02 & 0.00 & 1.00 & 0.02 \\
& & 75 & EDPM &  0.07 & 0.00 & 1.00 & 0.02 \\
& &    & BLCJM &  0.01 & 0.00 & 1.00 & 0.02 \\
\cline{2-8}
& 2{,}500
& 65 & EDPM &  0.04 & 0.00 & 1.00 & 0.02 \\
& &    & BLCJM &  0.02 & 0.00 & 1.00 & 0.02 \\
& & 75 & EDPM &  0.04 & 0.00 & 1.00 & 0.02 \\
& &    & BLCJM &  0.01 & 0.00 & 1.00 & 0.02 \\
\hline
TE & 1{,}500
& 65 & EDPM &  0.37 & 0.01 & 0.93 & 0.03 \\
& &    & BLCJM &  0.15 & 0.00 & 0.99 & 0.02 \\
& & 75 & EDPM &  0.43 & 0.01 & 0.91 & 0.03 \\
& &    & BLCJM &  0.21 & 0.00 & 0.98 & 0.02 \\
\cline{2-8}
& 2{,}500
& 65 & EDPM &  0.33 & 0.00 & 0.90 & 0.02 \\
& &    & BLCJM &  0.15 & 0.00 & 0.99 & 0.02 \\
& & 75 & EDPM &  0.40 & 0.01 & 0.87 & 0.03 \\
& &    & BLCJM &  0.22 & 0.00 & 0.99 & 0.02 \\
\hline
\end{tabular}
\end{table}

\noindent  Overall, these results indicate that the BLCJM attains the expected efficiency under correct specification, while the EDPM maintains competitive performance with only a modest loss in coverage for some of the direct and total effects. This trade-off is relatively minor in this setting and is balanced by the potential gains in robustness under model misspecification, which we examine next.

\subsection{Study 2: performance under model misspecification}\label{Sec74}

In this study, we evaluate the performance of the EDPM relative to the BLCJM under two distinct forms of misspecification: (i) functional form misspecification, where the true DGP includes quadratic and interaction terms not captured by the linear functional form of baseline covariates as in~\eqref{Eq:SimDGMMarginalJoint}, and (ii) latent structure misspecification, where the true number of latent classes differs from the two classes assumed by the BLCJM. 

\subsubsection{Study 2a: functional form misspecification}\label{Sec74a}

In this scenario, we evaluate the robustness of both methods when the true data-generating process includes quadratic effects and interaction terms that violate the linear main-effects (of baseline covariates) assumption of the BLCJM. Specifically, we modify the linear baseline covariates in the confounder, mediator, and survival models to include the following additional terms:
\begin{itemize}
    \item \textbf{Quadratic effect of BMI:} A term $\theta^{(r)}_{BMI^2} L_{i,0,1}^2$, where $\theta^{(r)}_{BMI^2} = -0.3$, is added to the $r^{\text{th}}$ class-specific linear predictor. This captures potential diminishing or accelerating effects of BMI on the outcomes.
    \item \textbf{BMI $\times$ Sex interaction:} A term $\theta^{(r)}_{BMI \times Sex} \big\{ L_{i,0,1} \times L_{i,0,2}\big\}$, where $\theta^{(r)}_{BMI \times Sex} = 0.5$, is added to the $r^{\text{th}}$ class-specific linear predictor, allowing the effect of BMI to differ by biological sex.
\end{itemize}
 
\noindent Under this misspecified DGP, the BLCJM, which assumes only linear main effects of baseline covariates within each latent class, cannot fully capture the true data-generating mechanism. The true causal effects under this DGP are reported in Table~\ref{Tab:TrueEffects_Study2a}.

\noindent Table~\ref{Tab:Study2a_Results} summarizes performance under functional form misspecification for $n \in \{1{,}500, 2{,}500\}$. For the IDE and TE, we find that both approaches produce stable point estimates, with no evidence of severe distortion from omitting quadratic and interaction terms. That said, a consistent pattern emerges as the EDPM tends to have noticeably lower MSE and higher coverage across both ages and sample sizes. Specifically, the BLCJM credible intervals are systematically under-covered, whereas the EDPM intervals achieve substantially higher coverage. Still, coverage for both models remains below the nominal level. For bias, however, the results are mixed, with EDPM outperforming in some settings while BLCJM yields smaller bias in others.

\begin{table}[H]
\singlespacing
\centering
\caption{True causal effects (survival probability scale) under the misspecified data-generating process with quadratic and interaction terms (Study 2a), computed via Monte Carlo integration with $C^{*} = 10{,}000$ subjects. True IDE, IIE, and TE are reported as $\times 100$.}
\label{Tab:TrueEffects_Study2a}
\small
\begin{tabular}{c|ccc|ccc}
\hline
Age & $\mathcal{S}^{\text{true}}_{\textbf{z}, \textbf{z}}(a)$ & $\mathcal{S}^{\text{true}}_{\textbf{z}_{*}, \textbf{z}_{*}}(a)$ & $\mathcal{S}^{\text{true}}_{\textbf{z}, \textbf{z}_{*}}(a)$ & True IDE ($\times 100$) & True IIE ($\times 100$) & True TE ($\times 100$) \\
\hline
65 & 0.6222 & 0.6947 & 0.6233 & $-$7.14 & $-$0.11 & $-$7.25 \\
75 & 0.5238 & 0.6184 & 0.5248 & $-$9.36 & $-$0.10 & $-$9.46 \\
\hline
\end{tabular}
\end{table}

\begin{table}[H]
\singlespacing
\centering
\caption{Study 2a results: Comparison of EDPM and BLCJM under functional form misspecification (quadratic and interaction terms in DGP) with $n \in \{1{,}500, 2{,}500\}$ and $R = 500$ replications. Bias and MSE are reported as $\times 100$.}
\label{Tab:Study2a_Results}
\small
\begin{tabular}{cccc|cccc}
\hline
Estimand & $n$ & Age & Model & Bias ($\times 100$) & MSE ($\times 100$) & 95\% Coverage & Average Interval Width \\
\hline
IDE & 1{,}500
& 65 & EDPM &  $-$0.29 & 0.10 & 0.70 & 0.07 \\
& &    & BLCJM &  0.87 & 0.33 & 0.39 & 0.07 \\
& & 75 & EDPM &  0.51 & 0.17 & 0.67 & 0.08 \\
& &    & BLCJM &  1.84 & 0.50 & 0.34 & 0.08 \\
\cline{2-8}
& 2{,}500
& 65 & EDPM &  $-$0.51 & 0.06 & 0.74 & 0.06 \\
& &    & BLCJM &  0.08 & 0.26 & 0.43 & 0.06 \\
& & 75 & EDPM &  0.22 & 0.10 & 0.70 & 0.07 \\
& &    & BLCJM &  0.88 & 0.37 & 0.39 & 0.07 \\
\hline
IIE & 1{,}500
& 65 & EDPM &  0.15 & 0.00 & 1.00 & 0.02 \\
& &    & BLCJM &  0.06 & 0.00 & 1.00 & 0.02 \\
& & 75 & EDPM &  0.17 & 0.00 & 1.00 & 0.02 \\
& &    & BLCJM &  0.04 & 0.00 & 1.00 & 0.02 \\
\cline{2-8}
& 2{,}500
& 65 & EDPM &  0.12 & 0.00 & 1.00 & 0.02 \\
& &    & BLCJM &  0.05 & 0.00 & 1.00 & 0.02 \\
& & 75 & EDPM &  0.13 & 0.00 & 1.00 & 0.02 \\
& &    & BLCJM &  0.03 & 0.00 & 1.00 & 0.02 \\
\hline
TE & 1{,}500
& 65 & EDPM &  $-$0.14 & 0.10 & 0.70 & 0.07 \\
& &    & BLCJM &  0.93 & 0.34 & 0.40 & 0.07 \\
& & 75 & EDPM &  0.68 & 0.17 & 0.67 & 0.08 \\
& &    & BLCJM &  1.88 & 0.51 & 0.35 & 0.08 \\
\cline{2-8}
& 2{,}500
& 65 & EDPM &  $-$0.39 & 0.06 & 0.75 & 0.06 \\
& &    & BLCJM &  0.13 & 0.26 & 0.42 & 0.06 \\
& & 75 & EDPM &  0.35 & 0.10 & 0.71 & 0.07 \\
& &    & BLCJM &  0.90 & 0.37 & 0.40 & 0.07 \\
\hline
\end{tabular}
\end{table}

\noindent For the IIE, both methods behave similarly and essentially achieve ideal performance as bias is negligible, MSE is near zero, and coverage remains at the nominal level across all scenarios. As we increase the sample size from $1{,}500$ to $2{,}500$, we observe modest gains in bias and MSE for the IDE and TE under both models, but coverage improves only slightly. This reflects a familiar phenomenon under misspecification as posterior distributions concentrate more tightly, but around a pseudo-true value rather than the true data-generating parameter. Overall, we see that while neither approach fully recovers the ground truth effect estimates, the EDPM provides more reliable uncertainty quantification and better finite-sample accuracy, offering a practical advantage in settings where functional form assumptions are difficult to justify.

\subsubsection{Study 2b: latent class misspecification}\label{Sec74b}

In this scenario, we investigate the consequences of misspecifying the number of latent classes in the data-generating process. We now generate data from a DGP with ten latent classes. This means the BLCJM fitted with two latent classes is under-parameterized relative to the true data-generating mechanism. This represents a form of latent structure misspecification in which the true population heterogeneity is more complex than the model assumes. Table~\ref{Tab:TrueEffects_Study2b} reports the true causal effects under this DGP, and Table~\ref {Tab:Study2b_Results} summarizes the simulation results comparing the EDPM and BLCJM.

\begin{table}[H]
\singlespacing
\centering
\caption{True causal effects under the misspecified data-generating process with ten latent classes (Study 2b), computed via Monte Carlo integration with $C^{*} = 10{,}000$ subjects. True IDE, IIE, and TE are reported as $\times 100$.}
\label{Tab:TrueEffects_Study2b}
\small
\begin{tabular}{c|ccc|ccc}
\hline
Age & $\mathcal{S}^{\text{true}}_{\textbf{z}, \textbf{z}}(a)$ & $\mathcal{S}^{\text{true}}_{\textbf{z}_{*}, \textbf{z}_{*}}(a)$ & $\mathcal{S}^{\text{true}}_{\textbf{z}, \textbf{z}_{*}}(a)$ & True IDE ($\times 100$) & True IIE ($\times 100$) & True TE ($\times 100$) \\
\hline
65 & 0.5942 & 0.7698 & 0.6096 & $-$16.03 & $-$1.54 & $-$17.57 \\
75 & 0.4996 & 0.7297 & 0.5237 & $-$20.61 & $-$2.40 & $-$23.01 \\
\hline
\end{tabular}
\end{table}

\begin{table}[H]
\singlespacing
\centering
\caption{Study 2b results: Comparison of EDPM and BLCJM under latent class misspecification (ten-class DGP, two-class BLCJM) with $n \in \{1{,}500, 2{,}500\}$ and $R = 500$ replications. Bias and MSE are reported as $\times 100$.}
\label{Tab:Study2b_Results}
\small
\begin{tabular}{cccc|cccc}
\hline
Estimand & $n$ & Age & Model & Bias ($\times 100$) & MSE ($\times 100$) & 95\% Coverage & Average Interval Width \\
\hline
IDE & 1{,}500
& 65 & EDPM &  0.43 & 0.20 & 0.80 & 0.11 \\
& &    & BLCJM &  4.81 & 1.00 & 0.33 & 0.07 \\
& & 75 & EDPM &  1.47 & 0.28 & 0.78 & 0.13 \\
& &    & BLCJM &  7.05 & 1.60 & 0.32 & 0.07 \\
\cline{2-8}
& 2{,}500
& 65 & EDPM &  0.03 & 0.15 & 0.78 & 0.09 \\
& &    & BLCJM &  1.94 & 0.77 & 0.28 & 0.06 \\
& & 75 & EDPM &  0.82 & 0.23 & 0.75 & 0.10 \\
& &    & BLCJM &  3.45 & 1.19 & 0.29 & 0.07 \\
\hline
IIE & 1{,}500
& 65 & EDPM &  $-$0.06 & 0.01 & 0.97 & 0.03 \\
& &    & BLCJM &  0.75 & 0.02 & 0.40 & 0.02 \\
& & 75 & EDPM &  0.50 & 0.01 & 0.90 & 0.04 \\
& &    & BLCJM &  1.40 & 0.04 & 0.25 & 0.03 \\
\cline{2-8}
& 2{,}500
& 65 & EDPM &  $-$0.15 & 0.00 & 0.97 & 0.03 \\
& &    & BLCJM &  0.68 & 0.02 & 0.45 & 0.02 \\
& & 75 & EDPM &  0.37 & 0.01 & 0.91 & 0.03 \\
& &    & BLCJM &  1.34 & 0.04 & 0.28 & 0.03 \\
\hline
TE & 1{,}500
& 65 & EDPM &  0.37 & 0.21 & 0.80 & 0.12 \\
& &    & BLCJM &  5.56 & 1.20 & 0.33 & 0.07 \\
& & 75 & EDPM &  1.97 & 0.30 & 0.79 & 0.13 \\
& &    & BLCJM &  8.45 & 2.01 & 0.32 & 0.08 \\
\cline{2-8}
& 2{,}500
& 65 & EDPM &  $-$0.12 & 0.16 & 0.80 & 0.09 \\
& &    & BLCJM &  2.63 & 0.92 & 0.28 & 0.07 \\
& & 75 & EDPM &  1.18 & 0.23 & 0.76 & 0.10 \\
& &    & BLCJM &  4.79 & 1.46 & 0.29 & 0.08 \\
\hline
\end{tabular}
\end{table}

\noindent Under latent class misspecification, we see a clear breakdown of the parametric BLCJM. As age increases from $65$ to $75$, bias and MSE grow steadily for both the indirect and total effects, while empirical coverage remains persistently low and largely insensitive to sample size. Increasing $n$ reduces variability but does not correct the fundamental misalignment between the model and the data-generating process, indicating that the dominant source of error is structural rather than stochastic. In other words, the two-class specification fails to capture the underlying heterogeneity induced by the ten latent classes, leading to overly narrow credible intervals that are consistently centered away from the truth.

\noindent In contrast, the EDPM remains substantially more robust to this form of misspecification. We observe consistently smaller bias and MSE across all estimands and age groups, along with improved coverage that is much closer to nominal levels. For the indirect effect in particular, EDPM maintains near-nominal coverage despite small absolute bias, whereas the BLCJM substantially undercovers, highlighting its tendency to underestimate uncertainty when the latent class structure is misspecified.

\subsection{Summary of simulation findings}\label{Sec75}

Under correct model specification (Study~1), both the EDPM and the BLCJM recover the IDE, IIE, and TE with negligible bias and MSE. However, the picture changes model misspecification. Under functional form misspecification (Study~2a), both methods continue to produce stable point estimates, but their uncertainty quantification begins to separate. The BLCJM shows clear signs of miscalibration, with coverage for the IDE and TE falling well below nominal levels across ages and sample sizes. The EDPM performs more favorably as MSE is consistently smaller and coverage is noticeably closer to nominal, although still imperfect. The improvement is not dramatic, but it is consistent across scenarios and points to the benefit of allowing more flexible structure. 

\noindent Under latent class misspecification (Study~2b), the differences are more pronounced. When the data are generated from a richer latent structure than the BLCJM allows, the BLCJM struggles to recover the direct and total effects as bias, MSE are substantially larger, and coverage remains persistently low even as the sample size increases. In contrast, the EDPM adapts more effectively to this setting. It keeps bias comparatively small and maintains much lower MSE across all scenarios. Coverage for the IDE and TE remains below nominal but is substantially improved relative to the BLCJM, and for the IIE, it stays close to nominal across ages and sample sizes. As in Study~2a, increasing $n$ primarily sharpens point estimation, with limited impact on coverage for either method. Overall, these results suggest that the EDPM retains the good behavior of the BLCJM under correct specification while offering a more reliable fallback when key modeling assumptions, particularly those governing latent heterogeneity, are violated.

\noindent A likely reason for the remaining undercoverage in both approaches is that the reported credible intervals are calibrated under the working model rather than the true data-generating mechanism. When the model is misspecified, the posterior tends to concentrate around a pseudo-true value, and additional data primarily tighten this concentration rather than correct it. As a result, even though bias and MSE improve with larger samples, the intervals can remain systematically off-target. Addressing this would require methods that explicitly account for the sensitivity of the causal estimand to nuisance model misspecification, which is beyond the scope of the current analysis.

\section{Real data analysis results}\label{Sec8}

In this section, we present results from applying our proposed EDPM methodology to estimate the interventional direct, indirect, and total effect of longitudinal blood pressure (BP) medication on time to cardiovascular disease (CVD) death in the ARIC cohort study. Recall that mean BP, defined as the average of SBP and DBP measured at each study visit, serves as the mediator in our analysis. Similarly, we define longitudinal smoking status as a confounder of the relationships among treatment, mediator, and outcome. We adjust for the following baseline confounders: biological sex, race, body mass index (BMI), educational attainment, smoking status (never, former, current), diabetes status, baseline age, and the total-to-HDL cholesterol ratio.

\noindent The analysis is conducted for the hypertensive at baseline population (n = $2552$) as well as stratified by biological sex ($1218$ male and $1334$ female). Figure~\ref{Fig:Interventional_Effects} presents the posterior means and 95\% credible intervals for all three causal effects across ages 40 to 85 years, while Figure~\ref{Fig:Posterior_Prop_of_Positive_Effects} displays the corresponding posterior probabilities of positive effects. Complete numerical results are provided in Appendix~\ref{AppendixG}.

\begin{figure}[H]
\singlespacing
\centering
\includegraphics[width = \textwidth]{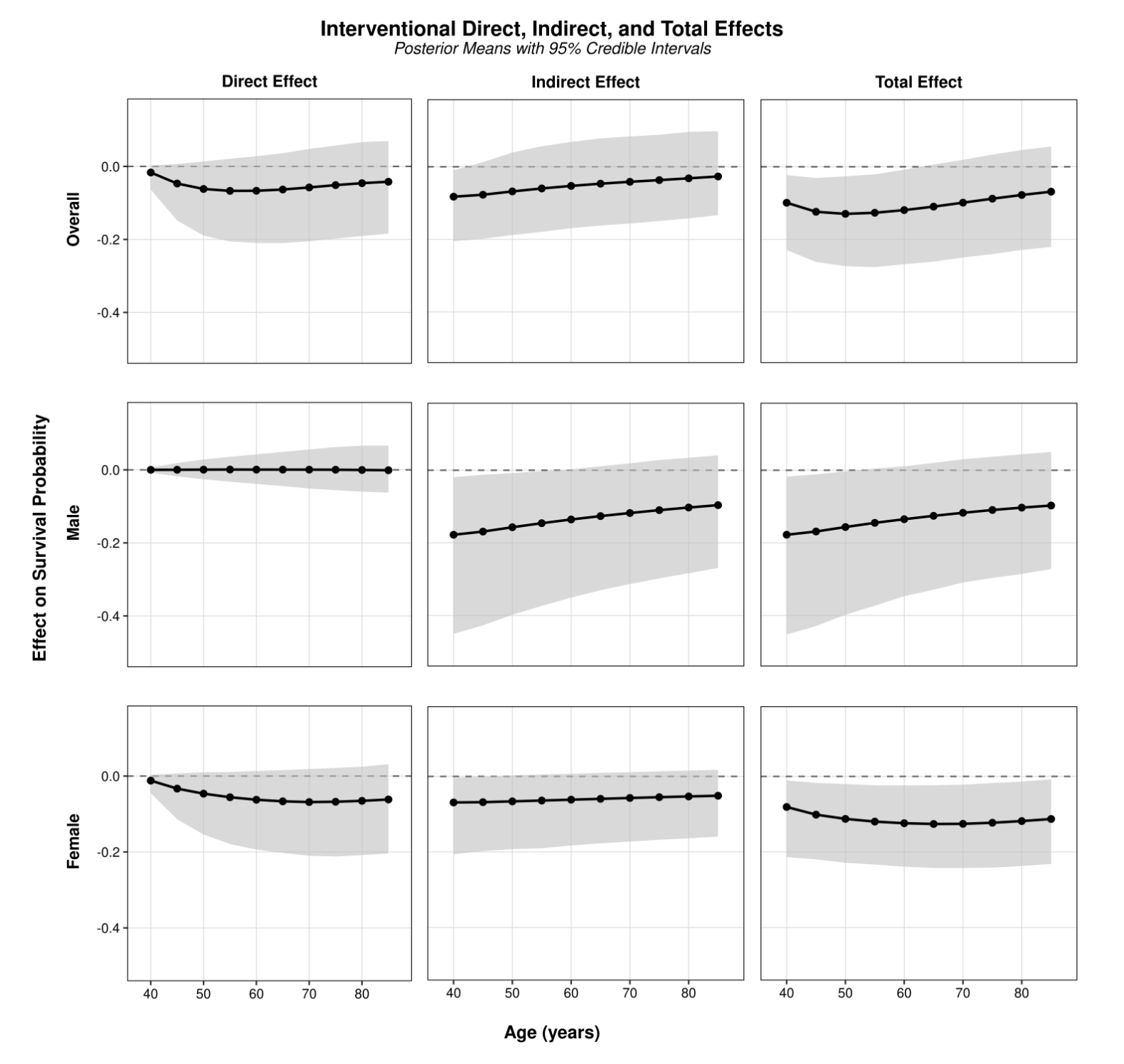}
\caption{Interventional direct, indirect, and total effects of blood pressure medication on survival probability across ages, stratified by biological sex. Each panel displays the posterior mean (solid line) and 95\% credible interval (shaded region). The dashed horizontal line indicates zero effect.}
\label{Fig:Interventional_Effects}
\end{figure}

\begin{figure}[H]
\singlespacing
\centering
\includegraphics[width = \textwidth]{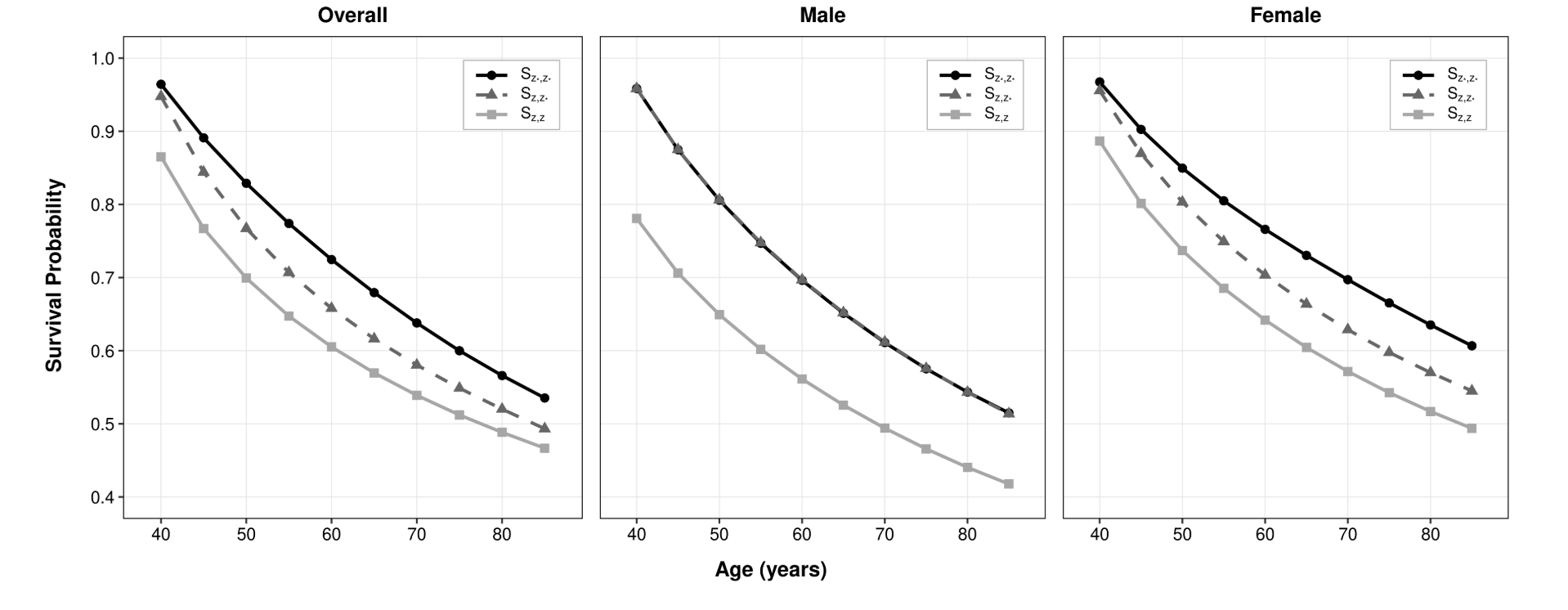}
\caption{Posterior mean potential survival probabilities under three treatment regimes: i) $S_{z_{*},z_{*}}$ (black), ii) $S_{z,z_{*}}$ (dark gray), and iii) $S_{z,z}$ (light gray), stratified by sex. Circle, triangle, and square denote estimates at ages 40--85 years for $S_{z_{*},z_{*}}$, $S_{z,z_{*}}$, and $S_{z,z}$, respectively. The gap between curves quantifies the interventional direct and indirect effects.}
\label{Fig:Potential_Surv_curves}
\end{figure}

\begin{figure}[H]
\singlespacing
\centering
\includegraphics[width = \textwidth]{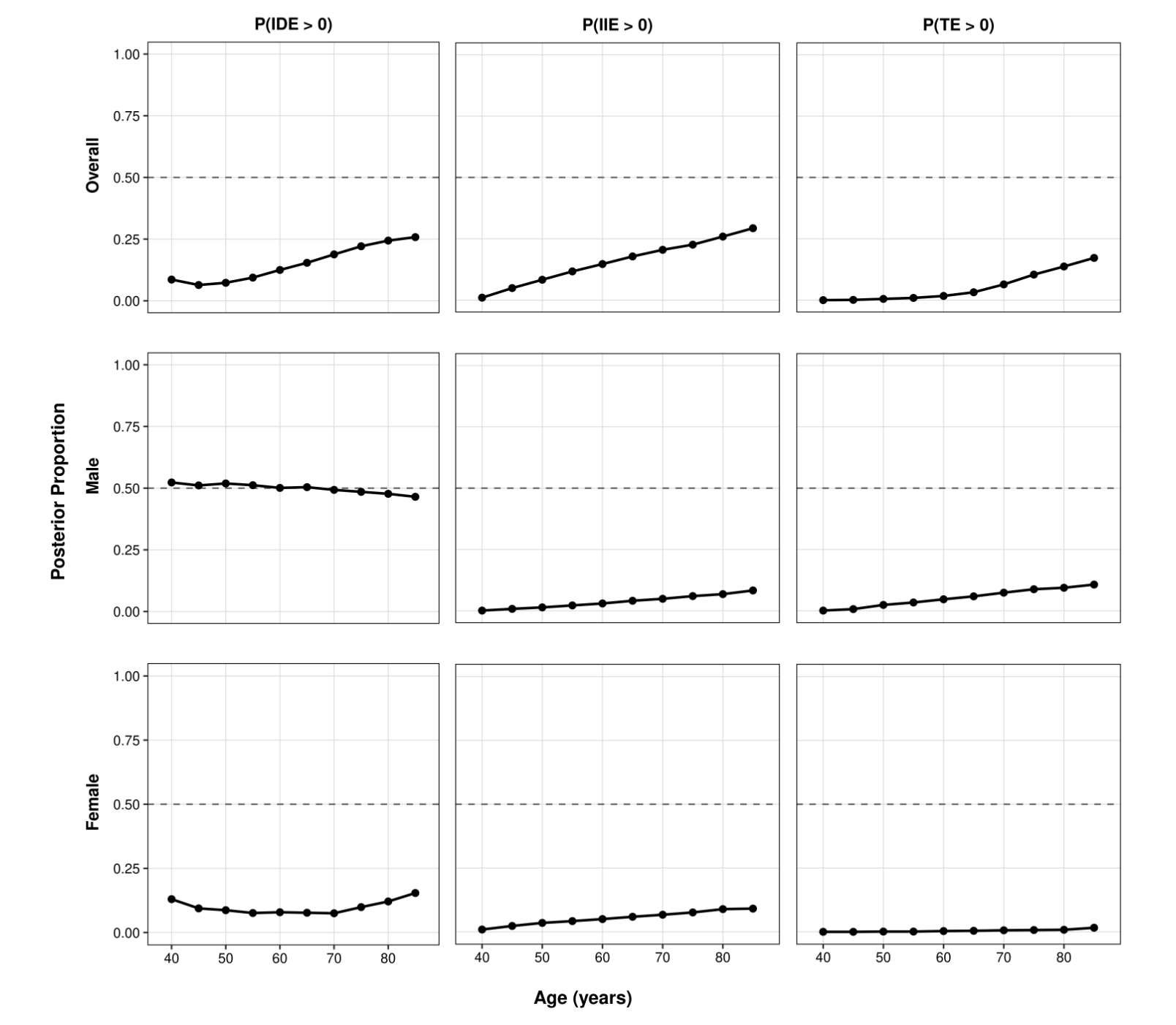}
\caption{Posterior probability of positive interventional direct, indirect, and total effects on survival probability across ages, stratified by biological sex. Probabilities were computed as the proportion of $Q$ saved MCMC iterations in which the effect estimate was positive. The dashed horizontal line indicates $P = 0.5$.}

\label{Fig:Posterior_Prop_of_Positive_Effects}
\end{figure}

\noindent For the spline-based longitudinal models, we place knots at ages $\{40, 45, 50, 55, 60, 65, 70, 75, 80, 85 \}$, and compute causal effect estimates on the same age grid. In principle, the spline knot locations and the ages at which causal effects are evaluated need not coincide. The fitted trajectory models can be used to impute predictor values at any age within the observed range. However, we chose to align the two grids for computational convenience: when the evaluation ages correspond exactly to the spline knots, the predicted risk factor values at these ages are determined directly by the spline basis functions without requiring interpolation.

\noindent For the piecewise constant baseline hazard in the survival submodel, we specify $B = 20$ disjoint age intervals spanning from minimum to the maximum observed event time in ARIC. We then fit the proposed EDPM model using four Markov chain Monte Carlo (MCMC) chains, running each for 20{,}000 iterations, and discarding the first 17{,}500 iterations as burn-in. From the remaining 2{,}500 iterations per chain, we retain every tenth iteration, yielding a total of 1{,}000 posterior samples across the four MCMC chains for inference. Finally, we compute the interventional direct, indirect, and total effect estimates using Algorithm~\ref{Alg:GcompAlg}, with 10{,}000 Monte Carlo samples at each retained iteration. Code for implementing the analysis is available at \url{https://github.com/SBstats/causal-mediation-continuous-time-EDPM}.

\subsection{Overall baseline hypertensive population}

Among individuals hypertensive at baseline, the posterior mean IDE of antihypertensive medication on survival probability is negative across the age range considered (Figure~\ref{Fig:Interventional_Effects}, top row). The magnitude is modest, with estimates of $-0.017$ (95\% CI: $-0.064$, $0.0025$) at age 40, reaching a peak around ages 55--65 (approximately $-0.067$), and attenuating to $-0.042$ (95\% CI: $-0.18$, $0.069$) at age 85. The credible intervals include zero at all ages, and the posterior probability of a positive direct effect increases with age, from $0.086$ at age 40 to $0.26$ at age 85 (Figure~\ref{Fig:Posterior_Prop_of_Positive_Effects}, top left panel). The evidence for a nonzero direct pathway is therefore weak and becomes less clear-cut at older ages.

\noindent The posterior mean IIE, representing the pathway through mean blood pressure, is also negative across ages but attenuates markedly. Estimates range from $-0.082$ (95\% CI: $-0.21$, $-0.0097$) at age 40 to $-0.027$ (95\% CI: $-0.13$, $0.098$) at age 85. The credible interval excludes zero only at age 40. Correspondingly, the posterior probability of a positive indirect effect rises from $0.010$ at age 40 to $0.29$ at age 85 (Figure~\ref{Fig:Posterior_Prop_of_Positive_Effects}, top middle panel), so the indirect effects are more clearly negative at younger ages and is no longer well separated from zero by age 85.

\noindent The posterior mean TE is negative across the full age range, with the largest magnitudes observed around ages 50--55 (approximately $-0.13$). The credible interval excludes zero at ages 40--60 but crosses zero at ages 65--85, and the posterior probability of a positive total effect rises from $0.000$ at age 40 to $0.17$ at age 85 (Figure~\ref{Fig:Posterior_Prop_of_Positive_Effects}, top right panel). The total effect is therefore more clearly distinguishable from zero at younger ages, and the evidence attenuates substantially with age.

\noindent Figure~\ref{Fig:Potential_Surv_curves} illustrates the potential survival curves. In the overall hypertensive at baseline population (left panel), potential survival under no treatment, $S_{z_{*},z_{*}}$, declines from $0.96$ at age 40 to $0.54$ at age 85, while potential survival under treatment, $S_{z,z}$, is uniformly lower (from $0.87$ to $0.47$). The intermediate curve, $S_{z,z_{*}}$, lies between the two. 

\subsection{Results stratified by biological sex }

The analysis stratified by biological sex reveals notable differences in how the total effect decomposes across sexes (Figures~\ref{Fig:Interventional_Effects} and~\ref{Fig:Posterior_Prop_of_Positive_Effects}), even though the overall magnitude of the total effect is similar.

\noindent Among males, the estimated IDE is essentially null across all ages (Figure~\ref{Fig:Interventional_Effects}, middle row). The posterior means are close to zero, with symmetric credible intervals that include zero and posterior probabilities near $0.50$ throughout (Figure~\ref{Fig:Posterior_Prop_of_Positive_Effects}, middle left panel). This pattern is consistent with no detectable direct effect from treatment to survival outside of the mediator.

\noindent By contrast, the IIE among males is negative and relatively large in magnitude, with estimates of $-0.18$ (95\% CI: $-0.45$, $-0.019$) at age 40 and $-0.096$ (95\% CI: $-0.27$, $0.041$) at age 85. The credible interval excludes zero at ages 40--55 but includes zero at ages 60 and above, and the posterior probability of a positive indirect effect rises from $0.0010$ at age 40 to $0.083$ at age 85 (Figure~\ref{Fig:Posterior_Prop_of_Positive_Effects}, middle panel). The uncertainty therefore widens appreciably with age. The total effect closely mirrors the indirect effect in both magnitude and pattern, reflecting the negligible contribution of the direct pathway. This decomposition is evident in the survival curves (Figure~\ref{Fig:Potential_Surv_curves}, middle panel), where $S_{z,z_{*}}$ and $S_{z_{*},z_{*}}$ are nearly identical across ages, so the difference between $S_{z_{*},z_{*}}$ and $S_{z,z}$ is attributable almost entirely to the mediated pathway through mean blood pressure.

\noindent Among females, both pathways contribute to the overall effect. The IDE is negative across ages, with estimates of $-0.012$ (95\% CI: $-0.045$, $0.0034$) at age 40, reaching a peak around age 70 (approximately $-0.069$), and attenuating to $-0.062$ (95\% CI: $-0.20$, $0.031$) at age 85 (Figure~\ref{Fig:Interventional_Effects}, bottom row). The credible intervals include zero at all ages, and posterior probabilities of a positive effect remain low (between $0.075$ and $0.15$), so the direct pathway is more likely negative than positive but is not clearly separated from zero.

\noindent The IIE is also negative, with estimates that attenuate modestly with age. The credible interval excludes zero only at ages 40 and 45. Posterior probabilities of a positive indirect effect are small throughout, rising from $0.0090$ at age 40 to $0.091$ at age 85 (Figure~\ref{Fig:Posterior_Prop_of_Positive_Effects}, bottom middle panel). The TE for females is negative across the full age range, with credible intervals excluding zero at all ten ages considered. The survival curves (Figure~\ref{Fig:Potential_Surv_curves}, right panel) reflect contributions from both pathways: $S_{z,z_{*}}$ lies below $S_{z_{*},z_{*}}$, consistent with a nonzero direct effect, and the additional gap between $S_{z,z_{*}}$ and $S_{z,z}$ reflects the indirect component.

\subsection{Interpretation}

We interpret these results with caution. In the ARIC study, antihypertensive treatment is not randomly assigned. Initiation and continuation of medication are driven by clinical considerations, such as elevated blood pressure, comorbidities, and overall cardiovascular risk, that are themselves predictive of survival. Even though we adjust for a rich set of baseline and time-varying predictors, it is unlikely that we fully capture all factors influencing both treatment decisions and mortality. As a result, the negative effects we estimate are likely due to residual confounding, where treated individuals are systematically at higher underlying risk.

\noindent More broadly, this analysis highlights what the EDPM framework brings to the table. By allowing flexible modeling of the longitudinal processes while targeting interpretable estimands, we are able to examine how they evolve with age and across subgroups.

\section{Discussion}\label{Sec9}

\noindent  We develop a novel Bayesian nonparametric framework for causal mediation analysis that jointly models time-varying exposures, confounders, mediators, and a survival outcome in continuous time. The enriched Dirichlet process mixture (EDPM) formulation provides a flexible representation of population heterogeneity and, as demonstrated in simulation, offers improved robustness to functional form and latent structure misspecification relative to a parametric finite-mixture alternative. In the ARIC hypertensive-at-baseline cohort, the proposed approach yields a decomposition of the effect of antihypertensive medication into direct and indirect pathways that varies by sex, while accommodating time-varying confounding and avoiding restrictive parametric assumptions on the observed-data distribution.

\noindent The validity of our analysis relies on sequential ignorability assumptions (Assumptions~\ref{Assump3:TrtIgnorability} and~\ref{Assump4:MediatorIgnorability}), which are not testable from observed data. In our setting, these assumptions require that, within sufficiently small time intervals, there are no unmeasured common causes of treatment and subsequent outcomes, nor of the mediator and subsequent outcomes, conditional on the observed history. These conditions could be violated if, for example, an acute cardiovascular event triggers both a change in medication and a worsening of the risk profile through pathways not captured by the measured covariates. The consistently negative total effects we observe are compatible with confounding by indication. This is a well-documented phenomenon in pharmacoepidemiological studies of antihypertensives \citep{bhandari2025bayesian, liu2015can}, and should not be interpreted as evidence that medication is harmful. Sensitivity analysis methods for unmeasured confounding in the mediation setting could be incorporated into the framework, though extending such methods to the continuous-time longitudinal setting presents its own challenges and is left for future work.

\noindent The proposed framework is computationally intensive. Posterior inference requires fitting a high-dimensional Bayesian model via MCMC, followed by g-computation using repeated Monte Carlo integration over posterior draws. While feasible for moderate sample sizes with parallel computing, the approach may not scale readily to substantially larger datasets without algorithmic modifications. Therefore, approximate inference methods, including variational approximations or subsampling-based MCMC, may offer viable alternatives but are not considered here.

\noindent Our causal estimands are defined through interventional (randomized) effects rather than natural direct and indirect effects. This choice is driven by the well-known identification difficulties that arise when natural effects are defined in the presence of treatment-affected confounders---the so-called "recanting witness" problem \citep{avin2005identifiability}. Recent work has revisited this problem from several angles: \cite{vo2026recanting} provide conditions under which natural effects can be partially recovered even in the presence of recanting witnesses, while \cite{domingo2026path} develop path-specific effect definitions that offer an alternative decomposition of treatment pathways. Our interventional effects sidestep the recanting witness issue entirely by replacing the subject's natural mediator value with a draw from the population mediator distribution under a specified treatment regime. The trade-off is interpretive as interventional effects answer a population-level question about the consequences of shifting the mediator distribution, rather than a subject-level question about the mediator a particular individual would have experienced under a different treatment. In settings where the population-level question is of primary scientific interest, as is often the case in public health and clinical epidemiology, this is a reasonable and arguably more actionable estimand.

\noindent There are several interesting directions for future extensions of our work. First, the current framework models a single mediator, so a natural extension is to incorporate multiple mediators with potentially complex causal ordering among them. Second, we assume that the survival outcome is subject only to right censoring. Accommodating competing risks, for example, distinguishing cardiovascular death from other causes, could be relevant in many applications. Finally, although the EDPM substantially reduces nuisance-model misspecification bias relative to a parametric finite mixture, the credible intervals it produces are not guaranteed to have nominal frequentist coverage when nuisance components are misspecified. Therefore, layering a semiparametric correction \citep{yiu2025semiparametric}, such as one-step estimation, or targeted maximum likelihood \citep{van2012targeted}, onto the proposed framework could provide a remedy and is an important direction for future work.

% ============================================================
\bibliographystyle{apalike}
\bibliography{ref}

@article{rytgaard2022continuous,
  title={{Continuous-time targeted minimum loss-based estimation of intervention-specific mean outcomes}},
  author={Rytgaard, Helene C and Gerds, Thomas A and van der Laan, Mark J},
  journal={The Annals of Statistics},
  volume={50},
  number={5},
  pages={2469--2491},
  year={2022},
  publisher={Institute of Mathematical Statistics}
}

@article{zeldow2021functional,
  title={{Functional clustering methods for longitudinal data with application to electronic health records}},
  author={Zeldow, Bret and Flory, James and Stephens-Shields, Alisa and Raebel, Marsha and Roy, Jason A},
  journal={Statistical Methods in Medical Research},
  volume={30},
  number={3},
  pages={655--670},
  year={2021},
  publisher={SAGE Publications Sage UK: London, England}
}

@article{shardell2018joint,
  title={{Joint mixed-effects models for causal inference with longitudinal data}},
  author={Shardell, Michelle and Ferrucci, Luigi},
  journal={Statistics in Medicine},
  volume={37},
  number={5},
  pages={829--846},
  year={2018},
  publisher={Wiley Online Library}
}

@article{zeng2022causal,
  title={{A Causal Mediation Model for Longitudinal Mediators and Survival Outcomes with an Application to Animal Behavior}},
  author={Zeng, Shuxi and Lange, Elizabeth C and Archie, Elizabeth A and Campos, Fernando A and Alberts, Susan C and Li, Fan},
  journal={Journal of Agricultural, Biological and Environmental Statistics},
  volume={28},
  number={2},
  pages={197--218},
  year={2022},
  publisher={Springer},
  doi={10.1007/s13253-022-00490-6}
}

@article{lin2017mediation,
  title={{Mediation analysis for a survival outcome with time-varying exposures, mediators, and confounders}},
  author={Lin, Sheng-Hsuan and Young, Jessica G and Logan, Roger and VanderWeele, Tyler J},
  journal={Statistics in Medicine},
  volume={36},
  number={26},
  pages={4153--4166},
  year={2017},
  publisher={Wiley Online Library}
}

@article{vanderweele2017mediation,
  title={{Mediation analysis with time varying exposures and mediators}},
  author={VanderWeele, Tyler J and Tchetgen Tchetgen, Eric J},
  journal={Journal of the Royal Statistical Society: Series B (Statistical Methodology)},
  volume={79},
  number={3},
  pages={917--938},
  year={2017},
  publisher={Wiley Online Library}
}

@article{zheng2017longitudinal,
  title={{Longitudinal mediation analysis with time-varying mediators and exposures, with application to survival outcomes}},
  author={Zheng, Wenjing and van der Laan, Mark},
  journal={Journal of Causal Inference},
  volume={5},
  number={2},
  pages={20160006},
  year={2017},
  publisher={De Gruyter},
  doi={10.1515/jci-2016-0006}
}

@article{rubin1981bayesian,
  title={{The Bayesian bootstrap}},
  author={Rubin, Donald B},
  journal={The Annals of Statistics},
  pages={130--134},
  year={1981},
  publisher={JSTOR}
}

@inproceedings{avin2005identifiability,
  title={{Identifiability of path-specific effects}},
  author={Avin, Chen and Shpitser, Ilya and Pearl, Judea},
  booktitle={Proceedings of the 19th International Joint Conference on Artificial Intelligence (IJCAI)},
  pages={357--363},
  year={2005}
}

@article{aric1989atherosclerosis,
  title={{The Atherosclerosis Risk in Communities (ARIC) Study: design and objectives}},
  author={{The ARIC Investigators}},
  journal={American Journal of Epidemiology},
  volume={129},
  number={4},
  pages={687--702},
  year={1989},
  publisher={Oxford University Press},
  doi={10.1093/oxfordjournals.aje.a115184}
}

@article{wade2014improving,
  title={{Improving prediction from Dirichlet process mixtures via enrichment}},
  author={Wade, Sara and Dunson, David B and Petrone, Sonia and Trippa, Lorenzo},
  journal={The Journal of Machine Learning Research},
  volume={15},
  number={1},
  pages={1041--1071},
  year={2014},
  publisher={JMLR. org}
}

@article{wade2011enriched,
  title={{An enriched conjugate prior for Bayesian nonparametric inference}},
  author={Wade, Sara and Mongelluzzo, Silvia and Petrone, Sonia},
  journal={Bayesian Analysis},
  volume={6},
  number={3},
  pages={359--385},
  year={2011},
  publisher={Institute of Mathematical Statistics},
  doi={10.1214/ba/1339616468}
}

@article{lok2008statistical,
  title={{Statistical modeling of causal effects in continuous time}},
  author={Lok, Judith J},
  journal={Annals of Statistics},
  pages={1464--1507},
  year={2008}
}

@article{didelez2019defining,
  title={{Defining causal mediation with a longitudinal mediator and a survival outcome}},
  author={Didelez, Vanessa},
  journal={Lifetime Data Analysis},
  volume={25},
  pages={593--610},
  year={2019},
  publisher={Springer}
}

@article{crainiceanu2005bayesian,
  title={{Bayesian analysis for penalized spline regression using WinBUGS}},
  author={Crainiceanu, Ciprian M. and Ruppert, David and Wand, Matthew P.},
  journal={Journal of Statistical Software},
  volume={14},
  number={14},
  pages={1--24},
  year={2005},
  doi={10.18637/jss.v014.i14}
}

@article{burns2023truncation,
  title={Truncation approximation for enriched dirichlet process mixture models},
  author={Burns, Natalie and Daniels, Michael J},
  journal={arXiv preprint arXiv:2305.01631},
  year={2023}
}

@article{roysland2011martingale,
  title={{A martingale approach to continuous-time marginal structural models}},
  author={R{\o}ysland, Kjetil},
  journal={Bernoulli},
  volume={17},
  number={3},
  pages={895--915},
  year={2011},
  publisher={Bernoulli Society for Mathematical Statistics and Probability},
  doi={10.3150/10-BEJ303}
}

@article{zhao2015dirichlet,
  title={{A Dirichlet process mixture model for survival outcome data: assessing nationwide kidney transplant centers}},
  author={Zhao, Lili and Shi, Jingchunzi and Shearon, Tempie H and Li, Yi},
  journal={Statistics in Medicine},
  volume={34},
  number={8},
  pages={1404--1416},
  year={2015},
  publisher={Wiley Online Library}
}

@book{christensen2010bayesian,
  title={{Bayesian ideas and data analysis: an introduction for scientists and statisticians}},
  author={Christensen, Ronald and Johnson, Wesley and Branscum, Adam and Hanson, Timothy E},
  year={2010},
  publisher={CRC press}
}

@article{hanson2004bayesian,
  title={{A Bayesian semiparametric AFT model for interval-censored data}},
  author={Hanson, Timothy and Johnson, Wesley O},
  journal={Journal of Computational and Graphical Statistics},
  volume={13},
  number={2},
  pages={341--361},
  year={2004},
  publisher={Taylor \& Francis}
}

@article{quintana2016bayesian,
  title={{Bayesian nonparametric longitudinal data analysis}},
  author={Quintana, Fernando A and Johnson, Wesley O and Waetjen, L Elaine and B. Gold, Ellen},
  journal={Journal of the American Statistical Association},
  volume={111},
  number={515},
  pages={1168--1181},
  year={2016},
  publisher={Taylor \& Francis}
}

@article{dunson2009bayesian,
  title={{Bayesian nonparametric hierarchical modeling}},
  author={Dunson, David B},
  journal={Biometrical Journal: Journal of Mathematical Methods in Biosciences},
  volume={51},
  number={2},
  pages={273--284},
  year={2009},
  publisher={Wiley Online Library}
}

@article{roy2018bayesian,
  title={{Bayesian nonparametric generative models for causal inference with missing at random covariates}},
  author={Roy, Jason and Lum, Kirsten J and Zeldow, Bret and Dworkin, Jordan D and Re III, Vincent Lo and Daniels, Michael J},
  journal={Biometrics},
  volume={74},
  number={4},
  pages={1193--1202},
  year={2018},
  publisher={Wiley Online Library}
}

@book{daniels2023bayesian,
  title={{Bayesian nonparametrics for causal inference and missing data}},
  author={Daniels, Michael J and Linero, Antonio and Roy, Jason},
  year={2023},
  publisher={Chapman and Hall/CRC}
}

@article{pearl2014interpretation,
  title={{Interpretation and identification of causal mediation}},
  author={Pearl, Judea},
  journal={Psychological Methods},
  volume={19},
  number={4},
  pages={459--481},
  year={2014},
  publisher={American Psychological Association},
  doi={10.1037/a0036434}
}

@article{liu2018exploring,
  title={{Exploring causality mechanism in the joint analysis of longitudinal and survival data}},
  author={Liu, Lei and Zheng, Cheng and Kang, Joseph},
  journal={Statistics in Medicine},
  volume={37},
  number={26},
  pages={3733--3744},
  year={2018},
  publisher={Wiley Online Library}
}

@article{zheng2022quantifying,
  title={{Quantifying direct and indirect effect for longitudinal mediator and survival outcome using joint modeling approach}},
  author={Zheng, Cheng and Liu, Lei},
  journal={Biometrics},
  volume={78},
  number={3},
  pages={1233--1243},
  year={2022},
  publisher={Oxford University Press}
}

@article{zhang2021mediation,
  title={{Mediation analysis for survival data with high-dimensional mediators}},
  author={Zhang, Haixiang and Zheng, Yinan and Hou, Lifang and Zheng, Cheng and Liu, Lei},
  journal={Bioinformatics},
  volume={37},
  number={21},
  pages={3815--3821},
  year={2021},
  publisher={Oxford University Press}
}

@article{caubet2023bayesian,
  title={{Bayesian joint modeling for causal mediation analysis with a binary outcome and a binary mediator: Exploring the role of obesity in the association between cranial radiation therapy for childhood acute lymphoblastic leukemia treatment and the long-term risk of insulin resistance}},
  author={Caubet, Miguel and Samoilenko, Mariia and Drouin, Simon and Sinnett, Daniel and Krajinovic, Maja and Laverdi{\`e}re, Caroline and Marcil, Val{\'e}rie and Lefebvre, Genevi{\`e}ve},
  journal={Computational Statistics \& Data Analysis},
  volume={177},
  pages={107586},
  year={2023},
  publisher={Elsevier}
}

@article{bhandari2025bayesian,
  title={A Bayesian semi-parametric approach to causal mediation for longitudinal mediators and time-to-event outcomes with application to a cardiovascular disease cohort study},
  author={Bhandari, Saurabh and Daniels, Michael J and Josefsson, Maria and Lloyd-Jones, Donald M and Siddique, Juned},
  journal={Biostatistics},
  volume={26},
  number={1},
  pages={kxaf027},
  year={2025},
  publisher={Oxford University Press},
  doi={10.1093/biostatistics/kxaf027}
}

@article{liu2015can,
  title={Can antihypertensive treatment restore the risk of cardiovascular disease to ideal levels? The Coronary Artery Risk Development in Young Adults (CARDIA) Study and the Multi-Ethnic Study of Atherosclerosis (MESA)},
  author={Liu, Kiang and Colangelo, Laura A and Daviglus, Martha L and Goff, David C and Pletcher, Mark and Schreiner, Pamela J and Sibley, Christopher T and Burke, Gregory L and Post, Wendy S and Michos, Erin D and others},
  journal={Journal of the American Heart Association},
  volume={4},
  number={9},
  pages={e002275},
  year={2015},
  publisher={Wolters Kluwer Health},
  doi={10.1161/JAHA.115.002275}
}

@article{sun2022estimating,
  title={Estimating mean potential outcome under adaptive treatment length strategies in continuous time},
  author={Sun, Hao and Ertefaie, Ashkan and Johnson, Brent A},
  journal={Biometrics},
  volume={78},
  number={4},
  pages={1503--1514},
  year={2022},
  publisher={Wiley Online Library}
}

@article{kateline2025continuous,
  title={Continuous-time mediation analysis for repeatedly measured mediators and outcomes},
  author={Le Bourdonnec, Kateline and Valeri, Linda and Proust-Lima, C{\'e}cile},
  journal={Biometrics},
  volume={81},
  number={2},
  pages={ujaf062},
  year={2025},
  publisher={Oxford University Press},
  doi={10.1093/biomtc/ujaf062}
}

@article{de2017programming,
  title={Programming with models: writing statistical algorithms for general model structures with NIMBLE},
  author={de Valpine, Perry and Turek, Daniel and Paciorek, Christopher J and Anderson-Bergman, Clifford and Lang, Duncan Temple and Bodik, Rastislav},
  journal={Journal of Computational and Graphical Statistics},
  volume={26},
  number={2},
  pages={403--413},
  year={2017},
  publisher={Taylor \& Francis}
}

@article{domingo2026path,
  title={A Path-Specific Effect Approach to Mediation Analysis With Time-Varying Mediators and Time-to-Event Outcomes Accounting for Competing Risks},
  author={Domingo-Relloso, Arce and Zhang, Yuchen and Wang, Ziqing and Suchy-Dicey, Astrid M and Buchwald, Dedra S and Navas-Acien, Ana and Schwartz, Joel and Berhane, Kiros and Coull, Brent A and Valeri, Linda},
  journal={Statistics in Medicine},
  volume={45},
  number={3-5},
  pages={e70425},
  year={2026},
  publisher={Wiley Online Library}
}

@article{vo2026recanting,
  title={Recanting twins: Addressing intermediate confounding in mediation analysis},
  author={Vo, Tat-Thang and Williams, Nicholas and Liu, Richard and Rudolph, Kara E and D{\'\i}az, Iv{\'a}n},
  journal={Statistics in Medicine},
  volume={45},
  number={3-5},
  pages={e70432},
  year={2026},
  publisher={Wiley Online Library}
}

@article{yiu2025semiparametric,
  title={Semiparametric posterior corrections},
  author={Yiu, Andrew and Fong, Edwin and Holmes, Chris and Rousseau, Judith},
  journal={Journal of the Royal Statistical Society Series B: Statistical Methodology},
  volume={87},
  number={4},
  pages={1025--1054},
  year={2025},
  publisher={Oxford University Press UK}
}

@article{wang2025targeted,
  title={Targeted maximum likelihood based estimation for longitudinal mediation analysis},
  author={Wang, Zeyi and Laan, Lars van der and Petersen, Maya and Gerds, Thomas and Kvist, Kajsa and Laan, Mark van der},
  journal={Journal of Causal Inference},
  volume={13},
  number={1},
  pages={20230013},
  year={2025},
  publisher={De Gruyter}
}

@incollection{van2012targeted,
  title={Targeted learning},
  author={Van der Laan, Mark J and Petersen, Maya L},
  booktitle={Ensemble machine learning},
  pages={117--156},
  year={2012},
  publisher={Springer}
}

@misc{NCHS_NDI,
  author       = {{National Center for Health Statistics}},
  title        = {National {Death} {Index}},
  year         = {2025},
  howpublished = {\url{https://www.cdc.gov/nchs/ndi/index.html}},
  institution  = {Centers for Disease Control and Prevention},
  address      = {Hyattsville, MD},
  note         = {Accessed: 2026-05-17}
}

\newpage

\appendix
% --- Restore default LaTeX 12pt display-math spacing for the appendix ---
\setlength{\abovedisplayskip}{12pt plus 3pt minus 7pt}
\setlength{\belowdisplayskip}{12pt plus 3pt minus 7pt}
\setlength{\abovedisplayshortskip}{0pt plus 3pt}
\setlength{\belowdisplayshortskip}{7pt plus 3pt minus 4pt}

\section{Appendix A: Proofs}\label{AppendixA}

\begin{lemma}\label{lemma1}
Let $\textbf{z}, \boldsymbol{\ell}, \textbf{m} : [0,\mathcal{A}] \rightarrow \mathbb{R}$ be Lipschitz continuous functions representing the longitudinal paths of treatment, confounder, and mediator processes, respectively. That is, there exist constants $C^Z, C^L, C^M > 0$ such that for all $t_1, t_2 \in [0, \mathcal{A}]$,
\begin{align*}
    \big\|\textbf{z}(t_1) - \textbf{z}(t_2)\big\| &\leq C^Z \big|t_1 - t_2\big|, \\
    \big\|\boldsymbol{\ell}(t_1) - \boldsymbol{\ell}(t_2)\big\| &\leq C^L \big|t_1 - t_2\big|, \\
    \big\|\textbf{m}(t_1) - \textbf{m}(t_2)\big\| &\leq C^M \big|t_1 - t_2\big|.
\end{align*}
Then, for any such function, its step function approximation defined on an equally spaced partition of $[0,\mathcal{A}]$ converges in $L^2$ norm to the original function at a rate of $O(W^{-2})$, where $W$ is the number of partition intervals.
\end{lemma}

\begin{proof}
The proof of this lemma builds upon the arguments presented in the proof of nonparametric identification of the causal parameter in \cite{zeng2022causal}. Let $\mathcal{A}_W = \big\{a_0 = 0, a_1 = \tfrac{\mathcal{A}}{W}, \ldots, a_W = \mathcal{A} \big\}$ be an equally spaced partition of $[0,\mathcal{A}]$ into $W$ subintervals. Define the step function approximation $\textbf{m}_{\mathcal{A}_W} : [0,\mathcal{A}] \to \mathbb{R}^d$ to $\textbf{m}$ by
\begin{align*}
    \textbf{m}_{\mathcal{A}_{W}}(t) = \begin{cases}
        \textbf{m}(a_0), \quad 0 \leq t < \mathcal{A}/W \\
        \textbf{m}(a_1), \quad \mathcal{A}/W \leq t < 2\mathcal{A}/W \\
        \ldots \\
        \textbf{m}(a_{W-1}), \quad (W-1)\mathcal{A}/W \leq t \leq \mathcal{A}.
    \end{cases}
\end{align*}
In other words, we define:
\[
\textbf{m}_{\mathcal{A}_W}(t) = \textbf{m}(a_{w-1}) \quad \text{for } t \in [a_{w-1}, a_w), \quad w = 1, \ldots, W.
\]
Similarly define $\boldsymbol{\ell}_{\mathcal{A}_W}$ and $\textbf{z}_{\mathcal{A}_W}$ for $\boldsymbol{\ell}$ and $\textbf{z}$.

We now quantify the approximation error in $L^2$ norm. Fix $\textbf{m}$ and observe:
\begin{align*}
\big\|\textbf{m} - \textbf{m}_{\mathcal{A}_W}\big\|_2^2 &= \int_0^{\mathcal{A}} \big\|\textbf{m}(t) - \textbf{m}_{\mathcal{A}_W}(t)\big\|^2 dt 
= \sum_{w=1}^W \int_{a_{w-1}}^{a_w} \big\|\textbf{m}(t) - \textbf{m}(a_{w-1})\big\|^2 dt.
\end{align*}
Since $\textbf{m}$ is Lipschitz with constant $C^M$, we have for $t \in [a_{w-1}, a_w)$,
\[
\big\|\textbf{m}(t) - \textbf{m}(a_{w-1})\big\| \leq C^M (t - a_{w-1}).
\]
Thus,
\begin{align*}
\int_{a_{w-1}}^{a_w} \big\|\textbf{m}(t) - \textbf{m}(a_{w-1})\big\|^2 dt &\leq (C^M)^2 \int_{a_{w-1}}^{a_w} (t - a_{w-1})^2 dt \\
&= (C^M)^2 \cdot \frac{(a_w - a_{w-1})^3}{3}.
\end{align*}
Since $a_w - a_{w-1} = \frac{\mathcal{A}}{W}$ for all $w$,
\[
\big\|\textbf{m} - \textbf{m}_{\mathcal{A}_W}\big\|_2^2 \leq (C^M)^2 \cdot \sum_{w=1}^W \frac{\mathcal{A}^3}{3W^3} = \frac{(C^M)^2 \mathcal{A}^3}{3W^2} = O(W^{-2}).
\]

The same analysis holds for $\boldsymbol{\ell}$ and $\textbf{z}$ using their respective Lipschitz constants $C^L$ and $C^Z$:
\[
\big\|\boldsymbol{\ell} - \boldsymbol{\ell}_{\mathcal{A}_W}\big\|_2^2 \leq \frac{(C^L)^2 \mathcal{A}^3}{3W^2}, \quad
\big\|\textbf{z} - \textbf{z}_{\mathcal{A}_W}\big\|_2^2 \leq \frac{(C^Z)^2 \mathcal{A}^3}{3W^2}.
\]
Hence, each path is closely approximated by a step function with $L^2$ convergence rate of $O(W^{-2})$ as $W \to \infty$.
\end{proof}

\begin{lemma}\label{lemma2}
Let $(\Omega,\mathcal{F},\mathbb{P})$ be a probability space, and let 
$Y, X_1, X_2, L$ be random variables taking values in measurable spaces. 
Define $D := X_2 - X_1$, assuming subtraction is well-defined and measurable. 
If
\[
Y \indep D \mid (X_1, L),
\]
then
\[
Y \indep X_2 \mid (X_1, L).
\]
\end{lemma}

\begin{proof}
Let $\sigma(\cdot)$ denote the $\sigma$-algebra generated by a random variable (or collection of random variables). 
Since $D = X_2 - X_1$ and $X_2 = X_1 + D$, it follows that
\[
\sigma(D, X_1) = \sigma(X_2, X_1),
\]
because each pair $(D, X_1)$ and $(X_2, X_1)$ is a measurable one-to-one transformation of the other.

By the assumption $Y \indep D \mid (X_1, L)$, we have that for every bounded measurable function $g$,
\[
\mathbb{E}\big[g(Y)\mid \sigma(D, X_1, L)\big]
=
\mathbb{E}\big[g(Y)\mid \sigma(X_1, L)\big]
\quad \text{a.s.}
\]

Using the equality of $\sigma$-algebras $\sigma(D, X_1) = \sigma(X_2, X_1)$, we obtain
\[
\sigma(D, X_1, L) = \sigma(X_2, X_1, L).
\]
Therefore,
\[
\mathbb{E}\big[g(Y)\mid \sigma(X_2, X_1, L)\big]
=
\mathbb{E}\big[g(Y)\mid \sigma(X_1, L)\big]
\quad \text{a.s.}
\]

This is precisely the definition of conditional independence:
\[
Y \indep X_2 \mid (X_1, L).
\]
\end{proof}

\subsection{Non-parametric identification proof (Proposition 1)}

We introduce the following regularity condition:

\begin{assumption}\label{Assump6:L2Continuity}
    $L^2$-continuity of the potential survival functional): \\
    The potential survival probability $\Pr\big[T_{(\textbf{z}(a), \mathbfcal{G}_{\textbf{z}_{*}(a)})} > a \mid \textbf{L}_{0}\big]$, viewed as a functional of the exposure, confounder, and mediator paths $(\textbf{z}, \boldsymbol{\ell}, \textbf{m})$ on $[0, a]$, is continuous with respect to the $L^2$ metric on path space. That is, if step-function approximations $(\textbf{z}_{\mathcal{A}_W}, \boldsymbol{\ell}_{\mathcal{A}_W}, \textbf{m}_{\mathcal{A}_W})$ converge in $L^2$ to $(\textbf{z}, \boldsymbol{\ell}, \textbf{m})$, the corresponding survival probabilities converge to those evaluated at the original continuous paths.
\end{assumption}
\noindent Unlike the causal identification assumptions in Section~\ref{Sec43} (positivity, consistency, ignorability, sequential ignorability, and independent censoring), Assumption~\ref{Assump6:L2Continuity} is a purely technical smoothness condition on the model class as it does not restrict the underlying data-generating mechanism in any causal sense. It is automatically satisfied by smoothly parameterized hazard models, including the cluster-specific Cox model with piecewise constant baseline hazard in equation~\eqref{Eq:CoxModel} of our EDPM specification.

\noindent Under the assumptions specified in Section~\ref{Sec43}, together with Assumption~\ref{Assump6:L2Continuity}, we will show that $\mathcal{S}_{\textbf{z},\textbf{z}_{*}}(a)$ can be identified from the observed data distribution by Equation~\eqref{Eq:NonParamIdentification} of Proposition 1.
\begin{proof}
By the definition of $\mathcal{S}_{\textbf{z},\textbf{z}_{*}}(a)$ , we have:
    \begin{align*}
        & \mathcal{S}_{\textbf{z},\textbf{z}_{*}}(a)  \\
         =& Pr\big[T_{\bigl(\textbf{Z}(a) = \textbf{z}(a), \mathbfcal{G}_{\textbf{z}_{*}(a)} = \textbf{m}(a)\bigr)} > a\big]
        \\
        =&
        \int_{\textbf{m}(a)} \int_{\boldsymbol{\ell}(a)} \int_{\boldsymbol{\ell}_{0}}  
        Pr \Bigl[T_{\bigl(\textbf{z}(a), \textbf{m}(a)\bigr)} > a \big|  \textbf{L}_{\bigl(\textbf{z}(a),\textbf{m}(a^{-})\bigr)}(a) = \boldsymbol{\ell}(a), \mathbfcal{G}_{\textbf{z}_{*}(a)}(a) =  \textbf{m}(a), 
        \textbf{L}_{0}= \boldsymbol{\ell}_{0} \Bigr] \times
        \\&
         f\Bigl(\mathbfcal{G}_{z_{*}(a)}(a) = \textbf{m}(a)\big| \textbf{L}_{\bigl(\textbf{z}(a),\textbf{m}(a^{-})\bigr)}(a) = \boldsymbol{\ell}(a), \mathbfcal{G}_{\textbf{z}_{*}(a^{-})}(a^{-}) =  \textbf{m}(a^{-}),
         \textbf{L}_{0}= \boldsymbol{\ell}_{0}\Bigr)
          \times
          \\& 
        f\Bigl(\textbf{L}_{\bigl(\textbf{z}(a),\textbf{m}(a^{-})\bigr)}(a) = \boldsymbol{\ell}(a)\big| \textbf{L}_{\bigl(\textbf{z}(a^{-}),\textbf{m}(a^{-})\bigr)}(a^{-}) = \boldsymbol{\ell}(a^{-}),
          \mathbfcal{G}_{\textbf{z}_{*}(a^{-})}(a^{-}) =  \textbf{m}(a^{-}),
        \textbf{L}_{0}= \boldsymbol{\ell}_{0}\Bigr)
         \times
         \\&
         f(\textbf{L}_{0})
         d\textbf{m}(a) d \boldsymbol{\ell}(a) d\boldsymbol{\ell}_{0} \quad \quad \big\{\text{law of total probability}\big\}
         \\
         = &
         \int_{\textbf{m}(a)} \int_{\boldsymbol{\ell}(a)} \int_{\boldsymbol{\ell}_{0}}  
        Pr \Bigl[T_{\bigl(\textbf{z}(a), \textbf{m}(a)\bigr)} > a \big|  \textbf{L}_{\bigl(\textbf{z}(a),\textbf{m}(a^{-})\bigr)}(a) = \boldsymbol{\ell}(a), \mathbfcal{G}_{\textbf{z}_{*}(a)}(a) =  \textbf{m}(a), 
        \textbf{L}_{0}= \boldsymbol{\ell}_{0} \Bigr] \times
        \\&
         f\Bigl(\textbf{M}_{\big(z_{*}(a)\big)}(a) = \textbf{m}(a)\big| \textbf{L}_{\big(z_{*}(a)\big)}(a) = \boldsymbol{\ell}(a), \textbf{M}_{\big(\textbf{z}_{*}(a^{-})\big)}(a^{-}) =  \textbf{m}(a^{-}),
         \textbf{L}_{0}= \boldsymbol{\ell}_{0}\Bigr)
          \times
          \\& 
        f\Bigl(\textbf{L}_{\bigl(\textbf{z}(a),\textbf{m}(a^{-})\bigr)}(a) = \boldsymbol{\ell}(a)\big| \textbf{L}_{\bigl(\textbf{z}(a^{-}),\textbf{m}(a^{-})\bigr)}(a^{-}) = \boldsymbol{\ell}(a^{-}),
          \mathbfcal{G}_{\textbf{z}_{*}(a^{-})}(a^{-}) =  \textbf{m}(a^{-}),
        \textbf{L}_{0}= \boldsymbol{\ell}_{0}\Bigr)
         \times
         \\&
         f(\textbf{L}_{0})
         d\textbf{m}(a) d \boldsymbol{\ell}(a) d\boldsymbol{\ell}_{0} \quad \quad \big\{\text{ by the definition of } \mathbfcal{G}_{\textbf{z}_{*}(a)} \big\}
         \\
         = & 
         \int_{\textbf{m}(a)} \int_{\boldsymbol{\ell}(a)} \int_{\boldsymbol{\ell}_{0}}  
        Pr \Bigl[T_{\bigl(\textbf{z}(a), \textbf{m}(a)\bigr)} > a \big|  \textbf{L}_{\bigl(\textbf{z}(a),\textbf{m}(a^{-})\bigr)}(a) = \boldsymbol{\ell}(a),  
        \textbf{L}_{0}= \boldsymbol{\ell}_{0} \Bigr] \times
        \\&
         f\Bigl(\textbf{M}_{\big(z_{*}(a)\big)}(a) = \textbf{m}(a)\big| \textbf{L}_{\big(z_{*}(a)\big)}(a) = \boldsymbol{\ell}(a), \textbf{M}_{\big(\textbf{z}_{*}(a^{-})\big)}(a^{-}) =  \textbf{m}(a^{-}),
         \textbf{L}_{0}= \boldsymbol{\ell}_{0}\Bigr)
          \times
          \\& 
        f\Bigl(\textbf{L}_{\bigl(\textbf{z}(a),\textbf{m}(a^{-})\bigr)}(a) = \boldsymbol{\ell}(a)\big| \textbf{L}_{\bigl(\textbf{z}(a^{-}),\textbf{m}(a^{-})\bigr)}(a^{-}) = \boldsymbol{\ell}(a^{-}),
        \textbf{L}_{0}= \boldsymbol{\ell}_{0}\Bigr)
         \times
         \\&
         f(\textbf{L}_{0})
         d\textbf{m}(a) d \boldsymbol{\ell}(a) d\boldsymbol{\ell}_{0}
         \\&
        \big\{ \text{by the definition of the randomized intervention } \mathbfcal{G}_{\textbf{z}_{*}(a)} \text{ in Section}~\ref{Sec41} \text{ together with}\\&
        \text{Assumption}~\ref{Assump4:MediatorIgnorability} \text{ (sequential ignorability), which together imply that, given the past}\\&
        \text{risk factor history, the random draw } \mathbfcal{G}_{\textbf{z}_{*}(a)}(a) \text{ is independent of the potential}\\&
       \text{confounder } \textbf{L}_{(\textbf{z}(a),\textbf{m}(a^{-}))}(a) \text{ and the potential time-to-event outcome } T_{(\textbf{z}(a),\textbf{m}(a))},\\&
       \text{conditional on } \textbf{L}_{(\textbf{z}(a),\textbf{m}(a^{-}))}(a) \text{ and } \textbf{L}_0 \big\}
    \end{align*}
Let $\mathcal{A}_W = \{a_0 = 0, a_1, \ldots, a_W = a\}$ be an equally spaced partition of $[0, a]$ with mesh size $\Delta = a/W$. For $W$ sufficiently large such that $\Delta < \epsilon$ (where $\epsilon$ is from Assumptions \ref{Assump3:TrtIgnorability} and \ref{Assump4:MediatorIgnorability}), we can express the conditional survival probability using the discrete-time analog.

\noindent Define the quantities:
\begin{itemize}
    \item $\textbf{z}_W = (z(a_0), z(a_1), \ldots, z(a_W))$: treatment values at partition points
    \item $\boldsymbol{\ell}_W = (\ell(a_0), \ell(a_1), \ldots, \ell(a_W))$: confounder values at partition points
    \item $\textbf{m}_W = (m(a_0), m(a_1), \ldots, m(a_W))$: mediator values at partition points
\end{itemize}
By Lemma \ref{lemma1}, the approximation error in the survival probability is $O(W^{-2})$ as $W \to \infty$.
Therefore, we can approximate the survival probability given continuous mediator, confounder, and treatment assignment processes with the values on the corresponding jump functions. In other words, as the values of steps function
$\textbf{m}_{\mathcal{A}_{W}}$, $\boldsymbol{\ell}_{\mathcal{A}_{W}}$, and $\textbf{z}_{\mathcal{A}_{W}}$ are completely determined by the values on finite jumps, we can approximate the
conditional survival probability as:
\begin{align*}
         & \mathcal{S}_{\textbf{z},\textbf{z}_{*}}(a)  \\
        = &
        \int_{\textbf{m}_{\mathcal{A}_{W}}(a_W)} \int_{\boldsymbol{\ell}_{\mathcal{A}_{W}}(a_W)} \int_{\boldsymbol{\ell}_{0}}
        \\&
        Pr \Bigl[T_{\bigl(\textbf{z}_{\mathcal{A}_{W}}(a_W), \mathbf{m}_{\mathcal{A}_{W}}(a_W)\bigr)} > a \big|  \textbf{L}_{\bigl(\textbf{z}_{\mathcal{A}_{W}}(a_W),\textbf{m}_{\mathcal{A}_{W}}(a_{W-1})\bigr)}(a_W) = \boldsymbol{\ell}_{\mathcal{A}_{W}}(a_W),
        \textbf{L}_{0}= \boldsymbol{\ell}_{0} \Bigr] \times
        \\&
         f\Bigl(\textbf{M}_{\big(z_{*,\mathcal{A}_{W}}(a_W)\big)}(a_W) = \textbf{m}_{\mathcal{A}_{W}}(a_W)\big| \textbf{L}_{\big(z_{*,\mathcal{A}_{W}}(a_W)\big)}(a_W) = \boldsymbol{\ell}_{\mathcal{A}_{W}}(a_W),
         \mathbf{M}_{\big(\textbf{z}_{*,\mathcal{A}_{W}}(a_{W-1})\big)}(a_{W-1}) =  \textbf{m}_{\mathcal{A}_{W}}(a_{W-1}),
         \\&
         \textbf{L}_{0}= \boldsymbol{\ell}_{0}\Bigr)
          \times  \\&
        f\Bigl(\textbf{L}_{\bigl( \textbf{z}_{\mathcal{A}_{W}}(a_W),\textbf{m}_{\mathcal{A}_{W}} (a_{W-1})\bigr)}(a_W) = \boldsymbol{\ell}_{\mathcal{A}_{W}}(a_W)\big|
         \textbf{L}_{\bigl(\textbf{z}_{\mathcal{A}_{W}}(a_{W-1}),\textbf{m}_{\mathcal{A}_{W}}(a_{W-2})\bigr)}(a_{W-1}) = \boldsymbol{\ell}_{\mathcal{A}_{W}}(a_{W-1}),
        \textbf{L}_{0}= \boldsymbol{\ell}_{0}\Bigr)
         \times
         \\&
         f(\textbf{L}_{0})
         d\textbf{m}_{\mathcal{A}_{W}}(a_W) d \boldsymbol{\ell}_{\mathcal{A}_{W}}(a_W) d\boldsymbol{\ell}_{0} + O\big(W^{-2} \big)
         \\& \big\{\text{by approximation of the continuous processes using the corresponding step functions} \}
         \\
        = &
        \int_{\big(m(a_0), m(a_1), \ldots, m(a_W) \big)} \int_{\big(\ell(a_0), \ell(a_1), \ldots, \ell(a_W) \big)} \int_{\boldsymbol{\ell}_{0}}
        \\&
        Pr \Bigl[T_{\bigl(\textbf{z}_{\mathcal{A}_{W}}(a_W), \textbf{m}_{\mathcal{A}_{W}}(a_W)\bigr)} > a \big|  \textbf{L}_{\bigl(\textbf{z}_{\mathcal{A}_{W}}(a_W),\textbf{m}_{\mathcal{A}_{W}}(a_{W-1})\bigr)}(a_W) = \big(\ell(a_0), \ell(a_1), \ldots, \ell(a_W) \big),
        \textbf{L}_{0}= \boldsymbol{\ell}_{0} \Bigr] \times
        \\&
         f\Bigl(\mathbf{M}_{\big(z_{*,\mathcal{A}_{W}}(a_W)\big)}(a_W) = \big(m(a_0), m(a_1), \ldots, m(a_W) \big)\big|  \textbf{L}_{\bigl(\textbf{z}_{*,\mathcal{A}_{W}}(a_W)\bigr)}(a_W) = \big(\ell(a_0), \ell(a_1), \ldots, \ell(a_W) \big), \\&
          \mathbf{M}_{\big(z_{*,\mathcal{A}_{W}}(a_{W-1})\big)}(a_{W-1}) = \big(m(a_0), m(a_1), \ldots, m(a_{W-1}) \big),
          \textbf{L}_{0}= \boldsymbol{\ell}_{0}\Bigr)
          \times
        f\Bigl(\textbf{L}_{\bigl( \textbf{z}_{\mathcal{A}_{W}}(a_W),\textbf{m}_{\mathcal{A}_{W}} (a_{W-1})\bigr)}(a_W) = \\& \big(\ell(a_0), \ell(a_1), \ldots, \ell(a_W) \big)\big|
         \textbf{L}_{\bigl(\textbf{z}_{\mathcal{A}_{W}}(a_{W-1}),\textbf{m}_{\mathcal{A}_{W}}(a_{W-2})\bigr)}(a_{W-1}) =
         \big(\ell(a_0), \ell(a_1), \ldots, \ell(a_{W-1}) \big),
        \textbf{L}_{0}= \boldsymbol{\ell}_{0}\Bigr)
         \times
         \\&
         f(\textbf{L}_{0})
         d\big(m(a_0), m(a_1), \ldots, m(a_W) \big)
         d\big(\ell(a_0), \ell(a_1), \ldots, \ell(a_W) \big) d\boldsymbol{\ell}_{0} + O\big(W^{-2} \big)
         \\& \big\{\text{since the values of steps function are completely determined by the values on finite jumps} \big\},
\end{align*}
where $\mathbf{M}_{\big(z_{*,\mathcal{A}_{W}}(a_W)\big)}(a_W)$, and  $\textbf{L}_{\bigl( \textbf{z}_{\mathcal{A}_{W}}(a_W),\textbf{m}_{\mathcal{A}_{W}} (a_{W-1})\bigr)}(a_W)$ denote the potential step functions induced
by the original potential process $\textbf{M}_{\big(z_{*}(a)\big)}(a)$ and $\textbf{L}_{\bigl(\textbf{z}(a),\textbf{m}(a)\bigr)}(a)$, respectively. The step above invokes Assumption~\ref{Assump6:L2Continuity}, under which the potential survival probability is continuous in the exposure, confounder, and mediator paths with respect to the $L^2$ metric. Combined with Lemma~\ref{lemma1}, this guarantees that the step-function approximations converge to the continuous-path values at rate $O(W^{-2})$. In summary, we have shown that:
\begin{equation}\label{Eq:SzztarApprox1}
    \begin{aligned}
        & S_{\textbf{z}, \textbf{z}_{*}}(a) \\
        =  & 
        \int_{\big(m(a_0), m(a_1), \ldots, m(a_W) \big)} \int_{\big(\ell(a_0), \ell(a_1), \ldots, \ell(a_W) \big)} \int_{\boldsymbol{\ell}_{0}}  
        \\& 
        Pr \Bigl[T_{\bigl(\textbf{z}_{\mathcal{A}_{W}}(a_W), \textbf{m}_{\mathcal{A}_{W}}(a_W)\bigr)} > a \big|  \textbf{L}_{\bigl(\textbf{z}_{\mathcal{A}_{W}}(a_W),\textbf{m}_{\mathcal{A}_{W}}(a_{W-1})\bigr)}(a_W) = \big(\ell(a_0), \ell(a_1), \ldots, \ell(a_W) \big),  
        \textbf{L}_{0}= \boldsymbol{\ell}_{0} \Bigr] \times
        \\&
         f\Bigl(\mathbf{M}_{\big(z_{*,\mathcal{A}_{W}}(a_W)\big)}(a_W) = \big(m(a_0), m(a_1), \ldots, m(a_W) \big)\big|  \textbf{L}_{\bigl(\textbf{z}_{*,\mathcal{A}_{W}}(a_W)\bigr)}(a_W) = \big(\ell(a_0), \ell(a_1), \ldots, \ell(a_W) \big), \\&
          \mathbf{M}_{\big(z_{*,\mathcal{A}_{W}}(a_{W-1})\big)}(a_{W-1}) = \big(m(a_0), m(a_1), \ldots, m(a_{W-1}) \big),
          \textbf{L}_{0}= \boldsymbol{\ell}_{0}\Bigr)
          \times 
        f\Bigl(\textbf{L}_{\bigl( \textbf{z}_{\mathcal{A}_{W}}(a_W),\textbf{m}_{\mathcal{A}_{W}} (a_{W-1})\bigr)}(a_W) = \\& \big(\ell(a_0), \ell(a_1), \ldots, \ell(a_W) \big)\big| 
         \textbf{L}_{\bigl(\textbf{z}_{\mathcal{A}_{W}}(a_{W-1}),\textbf{m}_{\mathcal{A}_{W}}(a_{W-2})\bigr)}(a_{W-1}) =
         \big(\ell(a_0), \ell(a_1), \ldots, \ell(a_{W-1}) \big),   
        \textbf{L}_{0}= \boldsymbol{\ell}_{0}\Bigr) 
         \times 
         \\& 
         f(\textbf{L}_{0}) 
         d\big(m(a_0), m(a_1), \ldots, m(a_W) \big)
         d\big(\ell(a_0), \ell(a_1), \ldots, \ell(a_W) \big) d\boldsymbol{\ell}_{0} + O\big(W^{-2}\big) ,
    \end{aligned}
\end{equation}
where $\mathcal{A}_{W}$ is defined as the set of ages where the finite jumps of the step functions are possible. Now, for $w = 0, \ldots, W$, let $z_{\mathcal{A}_{W}}(a_w)$ denote the fixed values of treatments (static intervention), and let $m_{\mathcal{A}_{W}}(a_{w})$ denote the potential values of $\mathcal{G}_{\textbf{z}_{*}(a_w)}(a_w)$ (randomized intervention) evaluated at $a_w = (a\times w)/W$.
Note that under assumptions \ref{Assump3:TrtIgnorability} and \ref{Assump4:MediatorIgnorability}, we can choose a large $W$ such that $a/W \leq \epsilon$. Then, extending the ideas from 
\cite{zeng2022causal}, assumptions \ref{Assump3:TrtIgnorability} and \ref{Assump4:MediatorIgnorability} are equivalent to the following:
\begin{assumptionp}{\ref{Assump3:TrtIgnorability}$'$}\label{Assump3prime:TrtIgnorability}
    Ignorability:
    \begin{enumerate}[label=(\roman*)]
        \item Assumption \ref{Assump3:TrtIgnorability} (i) is equivalent to the following conditional independence conditions:
        \begin{equation}\label{Assump:AppendixAssump3a}
        \resizebox{\textwidth}{!}{$
            \begin{aligned}
                \Big\{\textbf{L}_{\big( \textbf{z}_{*,\mathcal{A}_{W}}(a_{0}: a_W)\big)},\textbf{M}_{\big( \textbf{z}_{*,\mathcal{A}_{W}}(a_{0}: a_W)\big)},\textbf{T}_{\big( \textbf{z}_{*,\mathcal{A}_{W}}(a_W)\big)} \Big\}  &\indep z_{*}(a_0)  \big|\boldsymbol{\ell}_0
                \\
                \Big\{\textbf{L}_{\big( \textbf{z}_{*,\mathcal{A}_{W}}(a_{1}: a_W)\big)},\textbf{M}_{\big( \textbf{z}_{*,\mathcal{A}_{W}}(a_{1}: a_W)\big)},\textbf{T}_{\big( \textbf{z}_{*,\mathcal{A}_{W}}(a_W)\big)} \Big\}  &\indep \big( z_{*}(a_1) - z_{*}(a_0) \big)  \big|z_{*}(a_0), m(a_0), \boldsymbol{\ell}_0
                \\
                \Big\{\textbf{L}_{\big( \textbf{z}_{*,\mathcal{A}_{W}}(a_{2}: a_W)\big)},\textbf{M}_{\big( \textbf{z}_{*,\mathcal{A}_{W}}(a_{2}: a_W)\big)},\textbf{T}_{\big( \textbf{z}_{*,\mathcal{A}_{W}}(a_W)\big)} \Big\}  &\indep \big( z_{*}(a_2) - z_{*}(a_1) \big)  \big|z_{*}(a_0), z_{*}(a_1),  \ell(a_1), m(a_0), m(a_1), \textbf{L}_0
                \\
                &\ldots
                \\
                \Big\{\textbf{L}_{\big( \textbf{z}_{*,\mathcal{A}_{W}}( a_W)\big)},\textbf{M}_{\big( \textbf{z}_{*,\mathcal{A}_{W}}( a_W)\big)},\textbf{T}_{\big( \textbf{z}_{*,\mathcal{A}_{W}}(a_W)\big)} \Big\}  &\indep \big( z_{*}(a_W) - z_{*}(a_{W-1}) \big)  \big|z_{*}(a_0), z_{*}(a_1), \ldots, z_{*}(a_{W-1}), \\& \ell(a_1),\ldots, \ell(a_{W-1}), m(a_0), m(a_1), \ldots, m(a_{W-1}),\boldsymbol{\ell}_0.
            \end{aligned}
        $}
        \end{equation}
        
        \item Assumption \ref{Assump3:TrtIgnorability} (ii) is equivalent to the following conditional independence conditions:
        \begin{equation}\label{Assump:AppendixAssump3b}
        \resizebox{\textwidth}{!}{$
            \begin{aligned}
                \Big\{\textbf{L}_{\big( \textbf{z}_{\mathcal{A}_{W}}(a_{0}: a_W), \textbf{m}_{\mathcal{A}_{W}}(a_{0}: a_W)\big)},\textbf{M}_{\big( \textbf{z}_{\mathcal{A}_{W}}(a_{0}: a_W), \textbf{m}_{\mathcal{A}_{W}}(a_{0}: a_W)\big)},\textbf{T}_{\big( \textbf{z}_{\mathcal{A}_{W}}(a_W), \textbf{m}_{\mathcal{A}_{W}}(a_W)\big)} \Big\}  &\indep z(a_0)  \big|\boldsymbol{\ell}_0
                \\
                \Big\{\textbf{L}_{\big( \textbf{z}_{\mathcal{A}_{W}}(a_{1}: a_W), \textbf{m}_{\mathcal{A}_{W}}(a_{1}: a_W)\big)},\textbf{M}_{\big( \textbf{z}_{\mathcal{A}_{W}}(a_{1}: a_W), \textbf{m}_{\mathcal{A}_{W}}(a_{1}: a_W)\big)},\textbf{T}_{\big( \textbf{z}_{\mathcal{A}_{W}}(a_W), \textbf{m}_{\mathcal{A}_{W}}(a_W)\big)} \Big\}  &\indep \big( z(a_1) - z(a_0) \big)  \big|z(a_0), m(a_0), \boldsymbol{\ell}_0
                \\
                \Big\{\textbf{L}_{\big( \textbf{z}_{\mathcal{A}_{W}}(a_{2}: a_W), \textbf{m}_{\mathcal{A}_{W}}(a_{2}: a_W)\big)},\textbf{M}_{\big( \textbf{z}_{\mathcal{A}_{W}}(a_{2}: a_W), \textbf{m}_{\mathcal{A}_{W}}(a_{2}: a_W)\big)},\textbf{T}_{\big( \textbf{z}_{\mathcal{A}_{W}}(a_W), \textbf{m}_{\mathcal{A}_{W}}(a_W)\big)} \Big\}  &\indep \big( z(a_2) - z(a_1) \big)  \big|z(a_0), z(a_1),  \ell(a_1), \\& m(a_0), m(a_1), \boldsymbol{\ell}_0
                \\
                &\ldots
                \\
                \Big\{\textbf{L}_{\big( \textbf{z}_{\mathcal{A}_{W}}( a_W), \textbf{m}_{\mathcal{A}_{W}}( a_W)\big)},\textbf{M}_{\big( \textbf{z}_{\mathcal{A}_{W}}( a_W), \textbf{m}_{\mathcal{A}_{W}}( a_W)\big)},\textbf{T}_{\big( \textbf{z}_{\mathcal{A}_{W}}(a_W), \textbf{m}_{\mathcal{A}_{W}}( a_W)\big)} \Big\}  &\indep \big( z(a_W) - z(a_{W-1}) \big)  \big|z(a_0), z(a_1), \ldots, z(a_{W-1}), \\&  \ell(a_1),\ldots, \ell(a_{W-1}), m(a_0), m(a_1), \ldots, m(a_{W-1}),\boldsymbol{\ell}_0
            \end{aligned}
        $}
        \end{equation}
    \end{enumerate}
\end{assumptionp}

\begin{assumptionp}{\ref{Assump4:MediatorIgnorability}$'$}\label{Assump4prime:MediatorIgnorability}
Sequential Ignorability: \\
Assumption \ref{Assump4:MediatorIgnorability}  is equivalent to the following conditional independence conditions:
\begin{equation}\label{Assump:AppendixAssump4}
        \resizebox{\textwidth}{!}{$
            \begin{aligned}
                \Big\{\textbf{L}_{\big( \textbf{z}_{\mathcal{A}_{W}}(a_{1}: a_W), \textbf{m}_{\mathcal{A}_{W}}(a_{0}: a_W)\big)},\textbf{M}_{\big( \textbf{z}_{\mathcal{A}_{W}}(a_{1}: a_W), \textbf{m}_{\mathcal{A}_{W}}(a_{0}: a_W)\big)},\textbf{T}_{\big( \textbf{z}_{\mathcal{A}_{W}}(a_W), \textbf{m}_{\mathcal{A}_{W}}(a_W)\big)} \Big\}  &\indep m(a_0)  \big|z(a_0),\boldsymbol{\ell}_0
                \\
                \Big\{\textbf{L}_{\big( \textbf{z}_{\mathcal{A}_{W}}(a_{2}: a_W), \textbf{m}_{\mathcal{A}_{W}}(a_{1}: a_W)\big)},\textbf{M}_{\big( \textbf{z}_{\mathcal{A}_{W}}(a_{2}: a_W), \textbf{m}_{\mathcal{A}_{W}}(a_{1}: a_W)\big)},\textbf{T}_{\big( \textbf{z}_{\mathcal{A}_{W}}(a_W), \textbf{m}_{\mathcal{A}_{W}}(a_W)\big)} \Big\}  &\indep \big( m(a_1) - m(a_0) \big)  \big|z(a_0), z(a_1), \ell(a_1), m(a_0), \boldsymbol{\ell}_0
                \\
                &\ldots
                \\
                \Big\{\textbf{M}_{\big( \textbf{z}_{\mathcal{A}_{W}}( a_W), \textbf{m}_{\mathcal{A}_{W}}( a_W)\big)},\textbf{T}_{\big( \textbf{z}_{\mathcal{A}_{W}}(a_W), \textbf{m}_{\mathcal{A}_{W}}( a_W)\big)} \Big\}  &\indep \big( m(a_W) - m(a_{W-1}) \big)  \big|z(a_0), z(a_1), \ldots, z(a_{W}), \\& \ell(a_1),\ldots, \ell(a_{W}),  m(a_0), m(a_1), \ldots, m(a_{W-1}),\boldsymbol{\ell}_0
            \end{aligned}
        $}
        \end{equation}
\end{assumptionp}
\noindent With these established conditional independence expressions, we can rewrite Equation (\ref{Eq:SzztarApprox1}) as follows: 
\[
\resizebox{\textwidth}{!}{$
\begin{aligned}
        & S_{\textbf{z}, \textbf{z}_{*}}(a) \\
        =  &
        \int_{\big(m(a_0), m(a_1), \ldots, m(a_W) \big)} \int_{\big(\ell(a_0), \ell(a_1), \ldots, \ell(a_W) \big)} \int_{\boldsymbol{\ell}_{0}}
        \\&
        Pr \Bigl[T_{\bigl(\textbf{z}_{\mathcal{A}_{W}}(a_W), \textbf{m}_{\mathcal{A}_{W}}(a_W)\bigr)} > a \big|  \textbf{L}_{\bigl(\textbf{z}_{\mathcal{A}_{W}}(a_W),\textbf{m}_{\mathcal{A}_{W}}(a_{W-1})\bigr)}(a_W) = \big(\ell(a_0), \ell(a_1), \ldots, \ell(a_W) \big),
        \textbf{L}_{0}= \boldsymbol{\ell}_{0} \Bigr] \times
        \\&
         f\Bigl(\mathbf{M}_{\big(z_{*,\mathcal{A}_{W}}(a_W)\big)}(a_W) = \big(m(a_0), m(a_1), \ldots, m(a_W) \big)\big|  \textbf{L}_{\bigl(\textbf{z}_{*,\mathcal{A}_{W}}(a_W)\bigr)}(a_W) = \big(\ell(a_0), \ell(a_1), \ldots, \ell(a_W) \big), \\&
          \mathbf{M}_{\big(z_{*,\mathcal{A}_{W}}(a_{W-1})\big)}(a_{W-1}) = \big(m(a_0), m(a_1), \ldots, m(a_{W-1}) \big),
          \textbf{L}_{0}= \boldsymbol{\ell}_{0}\Bigr)
          \times
        f\Bigl(\textbf{L}_{\bigl( \textbf{z}_{\mathcal{A}_{W}}(a_W),\textbf{m}_{\mathcal{A}_{W}} (a_{W-1})\bigr)}(a_W) = \\& \big(\ell(a_0), \ell(a_1), \ldots, \ell(a_W) \big)\big|
         \textbf{L}_{\bigl(\textbf{z}_{\mathcal{A}_{W}}(a_{W-1}),\textbf{m}_{\mathcal{A}_{W}}(a_{W-2})\bigr)}(a_{W-1}) =
         \big(\ell(a_0), \ell(a_1), \ldots, \ell(a_{W-1}) \big),
        \textbf{L}_{0}= \boldsymbol{\ell}_{0}\Bigr)
         \times
         \\&
         f(\textbf{L}_{0})
         d\big(m(a_0), m(a_1), \ldots, m(a_W) \big)
         d\big(\ell(a_0), \ell(a_1), \ldots, \ell(a_W) \big) d\boldsymbol{\ell}_{0} + O\big(W^{-2}\big).
\end{aligned}
$}
\]

\noindent Sequentially applying Assumption~\ref{Assump3prime:TrtIgnorability}~(i) at $w = 0, \ldots, W$ together with Lemma~\ref{lemma2} yields:
\[
\resizebox{\textwidth}{!}{$
\begin{aligned}
        & S_{\textbf{z}, \textbf{z}_{*}}(a) \\
        =  &
        \int_{\big(m(a_0), m(a_1), \ldots, m(a_W) \big)} \int_{\big(\ell(a_0), \ell(a_1), \ldots, \ell(a_W) \big)} \int_{\boldsymbol{\ell}_{0}}
        \\&
        Pr \Bigl[T_{\bigl(\textbf{z}_{\mathcal{A}_{W}}(a_W), \textbf{m}_{\mathcal{A}_{W}}(a_W)\bigr)} > a \big|  \textbf{L}_{\bigl(\textbf{z}_{\mathcal{A}_{W}}(a_W),\textbf{m}_{\mathcal{A}_{W}}(a_{W-1})\bigr)}(a_W) = \big(\ell(a_0), \ell(a_1), \ldots, \ell(a_W) \big),
        \textbf{L}_{0}= \boldsymbol{\ell}_{0} \Bigr] \times
        \\&
         f\Bigl(\mathbf{M}_{\big(z_{*,\mathcal{A}_{W}}(a_W)\big)}(a_W) = \big(m(a_0), m(a_1), \ldots, m(a_W) \big)\big|  \big(z_{*}(a_0), z_{*}(a_1), \ldots, z_{*}(a_W) \big), \textbf{L}_{\bigl(\textbf{z}_{*,\mathcal{A}_{W}}(a_W)\bigr)}(a_W) = \\&
         \big(\ell(a_0), \ell(a_1), \ldots, \ell(a_W) \big),
          \mathbf{M}_{\big(z_{*,\mathcal{A}_{W}}(a_{W-1})\big)}(a_{W-1}) = \big(m(a_0), m(a_1), \ldots, m(a_{W-1}) \big),
          \textbf{L}_{0}= \boldsymbol{\ell}_{0}\Bigr)
          \times \\&
        f\Bigl(\textbf{L}_{\bigl( \textbf{z}_{\mathcal{A}_{W}}(a_W),\textbf{m}_{\mathcal{A}_{W}} (a_{W-1})\bigr)}(a_W) =  \big(\ell(a_0), \ell(a_1), \ldots, \ell(a_W) \big)\big|
         \textbf{L}_{\bigl(\textbf{z}_{\mathcal{A}_{W}}(a_{W-1}),\textbf{m}_{\mathcal{A}_{W}}(a_{W-2})\bigr)}(a_{W-1}) = \\&
         \big(\ell(a_0), \ell(a_1), \ldots, \ell(a_{W-1}) \big),
        \textbf{L}_{0}= \boldsymbol{\ell}_{0}\Bigr)
         \times
         f(\textbf{L}_{0})
         d\big(m(a_0), m(a_1), \ldots, m(a_W) \big)
         d\big(\ell(a_0), \ell(a_1), \ldots, \ell(a_W) \big) d\boldsymbol{\ell}_{0} + O\big(W^{-2}\big).
\end{aligned}
$}
\]

\noindent Applying Consistency Assumption~\ref{Assump2:Consistency}~(i) (since the values of step functions are completely determined by the values on the finite jumps):
\[
\resizebox{\textwidth}{!}{$
\begin{aligned}
        & S_{\textbf{z}, \textbf{z}_{*}}(a) \\
        =  &
        \int_{\big(m(a_0), m(a_1), \ldots, m(a_W) \big)} \int_{\big(\ell(a_0), \ell(a_1), \ldots, \ell(a_W) \big)} \int_{\boldsymbol{\ell}_{0}}
        \\&
        Pr \Bigl[T_{\bigl(\textbf{z}_{\mathcal{A}_{W}}(a_W), \textbf{m}_{\mathcal{A}_{W}}(a_W)\bigr)} > a \big|  \textbf{L}_{\bigl(\textbf{z}_{\mathcal{A}_{W}}(a_W),\textbf{m}_{\mathcal{A}_{W}}(a_{W-1})\bigr)}(a_W) = \big(\ell(a_0), \ell(a_1), \ldots, \ell(a_W) \big),
        \textbf{L}_{0}= \boldsymbol{\ell}_{0} \Bigr] \times
        \\&
         f\Bigl(\mathbf{M}(a_W) = \big(m(a_0), m(a_1), \ldots, m(a_W) \big)\big|   \big(z_{*}(a_0), z_{*}(a_1), \ldots, z_{*}(a_W) \big), \textbf{L}(a_W) =
         \big(\ell(a_0), \ell(a_1), \ldots, \ell(a_W) \big), \\&
          \mathbf{M}(a_{W-1}) = \big(m(a_0), m(a_1), \ldots, m(a_{W-1}) \big),
          \textbf{L}_{0}= \boldsymbol{\ell}_{0}\Bigr)
          \times \\&
        f\Bigl(\textbf{L}_{\bigl( \textbf{z}_{\mathcal{A}_{W}}(a_W),\textbf{m}_{\mathcal{A}_{W}} (a_{W-1})\bigr)}(a_W) =  \big(\ell(a_0), \ell(a_1), \ldots, \ell(a_W) \big)\big|
         \textbf{L}_{\bigl(\textbf{z}_{\mathcal{A}_{W}}(a_{W-1}),\textbf{m}_{\mathcal{A}_{W}}(a_{W-2})\bigr)}(a_{W-1}) = \\&
         \big(\ell(a_0), \ell(a_1), \ldots, \ell(a_{W-1}) \big),
        \textbf{L}_{0}= \boldsymbol{\ell}_{0}\Bigr)
         \times
         f(\textbf{L}_{0})
         \\&
         d\big(m(a_0), m(a_1), \ldots, m(a_W) \big)
         d\big(\ell(a_0), \ell(a_1), \ldots, \ell(a_W) \big) d\boldsymbol{\ell}_{0} + O\big(W^{-2}\big).
\end{aligned}
$}
\]

\noindent Sequentially applying Assumptions~\ref{Assump3prime:TrtIgnorability}~(ii) and~\ref{Assump4prime:MediatorIgnorability} at $w = 0, \ldots, W$ together with Lemma~\ref{lemma2}:
\[
\resizebox{\textwidth}{!}{$
\begin{aligned}
        & S_{\textbf{z}, \textbf{z}_{*}}(a) \\
        =  &
        \int_{\big(m(a_0), m(a_1), \ldots, m(a_W) \big)} \int_{\big(\ell(a_0), \ell(a_1), \ldots, \ell(a_W) \big)} \int_{\boldsymbol{\ell}_{0}}
        \\&
        Pr \Bigl[T_{\bigl(\textbf{z}_{\mathcal{A}_{W}}(a_W), \textbf{m}_{\mathcal{A}_{W}}(a_W)\bigr)} > a \big|  \big(z(a_0), z(a_1), \ldots, z(a_W) \big),  \big(m(a_0), m(a_1), \ldots, m(a_W) \big),\textbf{L}_{\bigl(\textbf{z}_{\mathcal{A}_{W}}(a_W),\textbf{m}_{\mathcal{A}_{W}}(a_{W-1})\bigr)}(a_W) =
        \\& \big(\ell(a_0), \ell(a_1), \ldots, \ell(a_W) \big),
        \textbf{L}_{0}= \boldsymbol{\ell}_{0} \Bigr] \times
         f\Bigl(\mathbf{M}(a_W) = \big(m(a_0), m(a_1), \ldots, m(a_W) \big)\big|   \big(z_{*}(a_0), z_{*}(a_1), \ldots, z_{*}(a_W) \big), \\&
         \textbf{L}(a_W) =
         \big(\ell(a_0), \ell(a_1), \ldots, \ell(a_W) \big),
          \mathbf{M}(a_{W-1}) = \big(m(a_0), m(a_1), \ldots, m(a_{W-1}) \big),
          \textbf{L}_{0}= \boldsymbol{\ell}_{0}\Bigr)
          \times \\&
        f\Bigl(\textbf{L}_{\bigl( \textbf{z}_{\mathcal{A}_{W}}(a_W),\textbf{m}_{\mathcal{A}_{W}} (a_{W-1})\bigr)}(a_W) =  \big(\ell(a_0), \ell(a_1), \ldots, \ell(a_W) \big)\big| \big(z(a_0), z(a_1), \ldots, z(a_W) \big),  \big(m(a_0), m(a_1), \ldots, m(a_{W-1}) \big), \\&
         \textbf{L}_{\bigl(\textbf{z}_{\mathcal{A}_{W}}(a_{W-1}),\textbf{m}_{\mathcal{A}_{W}}(a_{W-2})\bigr)}(a_{W-1}) =
         \big(\ell(a_0), \ell(a_1), \ldots, \ell(a_{W-1}) \big),
        \textbf{L}_{0}= \boldsymbol{\ell}_{0}\Bigr)
         \times
         f(\textbf{L}_{0}) 
         \\&
         d\big(m(a_0), m(a_1), \ldots, m(a_W) \big)
         d\big(\ell(a_0), \ell(a_1), \ldots, \ell(a_W) \big) d\boldsymbol{\ell}_{0} + O\big(W^{-2}\big).
\end{aligned}
$}
\]

\noindent Finally, applying Consistency Assumptions~\ref{Assump2:Consistency}~(ii) and~(iii):
\[
\resizebox{\textwidth}{!}{$
\begin{aligned}
        & S_{\textbf{z}, \textbf{z}_{*}}(a) \\
        =  &
        \int_{\big(m(a_0), m(a_1), \ldots, m(a_W) \big)} \int_{\big(\ell(a_0), \ell(a_1), \ldots, \ell(a_W) \big)} \int_{\boldsymbol{\ell}_{0}}
        \\&
        Pr \Bigl[T> a \big|  \big(z(a_0), z(a_1), \ldots, z(a_W) \big),  \big(m(a_0), m(a_1), \ldots, m(a_W) \big),\textbf{L}(a_W) =
         \big(\ell(a_0), \ell(a_1), \ldots, \ell(a_W) \big),
        \textbf{L}_{0}= \boldsymbol{\ell}_{0} \Bigr] \\& \times
         f\Bigl(\mathbf{M}(a_W) = \big(m(a_0), m(a_1), \ldots, m(a_W) \big)\big|   \big(z_{*}(a_0), z_{*}(a_1), \ldots, z_{*}(a_W) \big),
         \textbf{L}(a_W) =
         \big(\ell(a_0), \ell(a_1), \ldots, \ell(a_W) \big), \\&
          \mathbf{M}(a_{W-1}) = \big(m(a_0), m(a_1), \ldots, m(a_{W-1}) \big),
          \textbf{L}_{0}= \boldsymbol{\ell}_{0}\Bigr)
          \times \\&
        f\Bigl(\textbf{L}(a_W) =  \big(\ell(a_0), \ell(a_1), \ldots, \ell(a_W) \big)\big| \big(z(a_0), z(a_1), \ldots, z(a_W) \big),  \big(m(a_0), m(a_1), \ldots, m(a_{W-1}) \big), \\&
         \textbf{L}(a_{W-1}) =
         \big(\ell(a_0), \ell(a_1), \ldots, \ell(a_{W-1}) \big),
        \textbf{L}_{0}= \boldsymbol{\ell}_{0}\Bigr)
         \times
         f(\textbf{L}_{0})
         \\&
         d\big(m(a_0), m(a_1), \ldots, m(a_W) \big)
         d\big(\ell(a_0), \ell(a_1), \ldots, \ell(a_W) \big) d\boldsymbol{\ell}_{0} + O\big(W^{-2}\big).
\end{aligned}
$}
\]

\noindent Therefore, as $W$ goes to infinity, we have:
\begin{equation}
    \begin{aligned}
        & S_{\textbf{z}, \textbf{z}_{*}}(a) \\
        &=   
        \int_{\textbf{m}_i(a)} \int_{\boldsymbol{\ell}_i(a)} \int_{\boldsymbol{\ell}_{0}}  
        Pr \Bigl[T_i> a \big| \textbf{Z}_i(a) = \textbf{z}_i(a), \textbf{L}_i(a) = \boldsymbol{\ell}_i(a),   \textbf{M}_i(a) = \textbf{m}_i(a),   
        \textbf{L}_{i,0}= \boldsymbol{\ell}_{i,0} \Bigr] 
        \times \\& 
         f\Bigl(\mathbf{M}_i(a) = \textbf{m}_i(a)\big|   \textbf{Z}_i(a) = \textbf{z}_{i,*}(a), 
         \textbf{L}(a) = 
         \boldsymbol{\ell}_i(a), \mathbf{M}_i(a^{-}) = \textbf{m}_i(a^{-}), \textbf{L}_{i,0}= \boldsymbol{\ell}_{i,0} \Bigr)
          \times \\&
        f\Bigl(\mathbf{L}_i(a) = \boldsymbol{\ell}_i(a)\big| \textbf{Z}_i(a) = \textbf{z}_{i}(a), \textbf{L}(a^{-}) = 
         \boldsymbol{\ell}_i(a^{-}), \mathbf{M}_i(a^{-}) = \textbf{m}_i(a^{-}), \textbf{L}_{i,0}= \boldsymbol{\ell}_{i,0}\Bigr) 
         \times  
         f(\textbf{L}_{0}) \\&
         d\textbf{m}_i(a)
         d\boldsymbol{\ell}_i(a) d\boldsymbol{\ell}_{i,0}
         \\
         &=   
        \int_{\textbf{m}_i(a)} \int_{\boldsymbol{\ell}_i(a)} \int_{\boldsymbol{\ell}_{0}} \int_{\textbf{b}_i}
        Pr \Bigl[T_i> a \big| \textbf{Z}_i(a) = \textbf{z}_i(a), \textbf{L}_i(a) = \boldsymbol{\ell}_i(a),   \textbf{M}_i(a) = \textbf{m}_i(a),   
        \textbf{L}_{i,0}= \boldsymbol{\ell}_{i,0} \Bigr] 
        \times \\& 
         f\Bigl(\mathbf{M}_i(a) = \textbf{m}_i(a)\big|  \textbf{Z}_i(a) = \textbf{z}_{i,*}(a), 
         \textbf{L}(a) = 
         \boldsymbol{\ell}_i(a), \mathbf{M}_i(a^{-}) = \textbf{m}_i(a^{-}), \textbf{L}_{i,0}= \boldsymbol{\ell}_{i,0}, b_i^M \Bigr)
          \times \\&
        f\Bigl(\mathbf{L}_i(a) = \boldsymbol{\ell}_i(a) \big| \textbf{Z}_i(a) = \textbf{z}_{i}(a), \textbf{L}(a^{-}) =
         \boldsymbol{\ell}_i(a^{-}), \mathbf{M}_i(a^{-}) = \textbf{m}_i(a^{-}), \textbf{L}_{i,0}= \boldsymbol{\ell}_{i,0}, b_i^L\Bigr)
         \times  
         f(\textbf{L}_{0})  \\&
         d\textbf{m}_i(a)
         d\boldsymbol{\ell}_i(a) d\boldsymbol{\ell}_{i,0} d\textbf{b}_i
         \\&
         \big\{ \text{by the random-effects structure of the EDPM in equation}~\eqref{eq:EDPMmodel}\big\}
    \end{aligned}
\end{equation}

\noindent Finally, under Assumption~\ref{Assump5:IndepCensoring} (independent censoring), the event indicator and observed time $\tilde{T}_i = \min(T_i, C_i)$ satisfy $\Pr[T_i > a \mid \cdot] = \Pr[\tilde{T}_i > a \mid \cdot]$ for $a \leq C_i$, and the discrete-time factorization of the conditional survival function on the age grid yields
\begin{equation}
    \begin{aligned}   
        & \int_{\textbf{m}_i(a)} \int_{\boldsymbol{\ell}_i(a)} \int_{\boldsymbol{\ell}_{0}} \int_{\textbf{b}_i}  \\
          & 
        \prod_{k = 1}^{K} \bigg\{ Pr \Bigl[\tilde{T}_i> a_{k} \big|  \tilde{T}_i \geq a_{k-1},\textbf{Z}_i(a_{k}) = \textbf{z}_i(a_{k}), \textbf{L}_i(a_{k}) = \boldsymbol{\ell}_i(a_{k}),   \textbf{M}_i(a_{k}) = \textbf{m}_i(a_{k}),   
        \textbf{L}_{i,0}= \boldsymbol{\ell}_{i,0} \Bigr] 
        \times \\& 
         f\Bigl(\mathbf{M}_i(a_{k}) = \textbf{m}_i(a_{k})\big| \tilde{T}_i \geq a_{k-1},   \textbf{Z}_i(a_{k}) = \textbf{z}_{i,*}(a_{k}),
         \textbf{L}(a_{k}) =
         \boldsymbol{\ell}_i(a_{k}), \mathbf{M}_i(a_{k-1}) = \textbf{m}_i(a_{k-1}), \textbf{L}_{i,0}= \boldsymbol{\ell}_{i,0}, b_i^M \Bigr)
          \times \\&
        f\Bigl(\mathbf{L}_i(a_{k}) = \boldsymbol{\ell}_i(a_{k})\big|\tilde{T}_i \geq a_{k-1}, \textbf{Z}_i(a_{k}) = \textbf{z}_{i}(a_{k}), \textbf{L}(a_{k-1}) = 
         \boldsymbol{\ell}_i(a_{k-1}), \mathbf{M}_i(a_{k-1}) = \textbf{m}_i(a_{k-1}), \textbf{L}_{i,0}= \boldsymbol{\ell}_{i,0}, b_i^L\Bigr) \bigg\}
         \times  
         \\&
         f(\textbf{L}_{0})  
         d\textbf{m}_i(a)
         d\boldsymbol{\ell}_i(a) d\boldsymbol{\ell}_{i,0} d\textbf{b}_i
    \end{aligned}
\end{equation}
where $a_1 < a_2 < \ldots < a_k < \ldots < a_K < \ldots$ is the age grid where at least one event occurs within each interval $(a_{k-1}, a_k]$, and the fixed age $a$ lies between $a_K$ and $a_{K+1}$. The product-over-$k$ form is the Kaplan-Meier-type factorization $\Pr[\tilde{T}_i > a_K \mid \cdot] = \prod_{k=1}^{K} \Pr[\tilde{T}_i > a_k \mid \tilde{T}_i \geq a_{k-1}, \cdot]$, valid under Assumption~\ref{Assump5:IndepCensoring}.

\end{proof}

\section{Appendix B: EDP base measures and prior specification} \label{AppendixB}
The notation $EDP(\alpha^{\beta}, \alpha^{\theta|\beta}, H_0 )$ means that $H_{\beta} \sim DP(\alpha^{\beta}, H_{0\beta})$ and, conditionally on $\boldsymbol{\beta}$, $H_{\theta|\beta} \sim DP(\alpha^{\theta|\beta}, H_{0\theta|\beta})$, where $\alpha^{\beta}$ and $\alpha^{\theta|\beta}$ are positive valued parameters and $H_0 = H_{0\beta} \times H_{0\theta|\beta}$ is the base distribution with the parameters $\boldsymbol{\beta}$ and $\boldsymbol{\theta}$ independent of each other. 

\subsection{Priors for the $\beta$-level parameters}
For the coefficient of the $p^{th}$ predictor, $\beta_{i,p}$, in the local Cox regression model, we assume the prior:
\begin{align*}
    \beta_{i,p}\stackrel{ind}{\sim}N\big(\beta_{0,p}, c\sigma^{2,T}_{0,p} \big).
\end{align*}
We define $\beta_{0,p}$ and $\sigma^{2,T}_{0,p}$ as the estimates for the $p^{th}$ coefficient obtained from fitting a Cox model on the regression predictors in the dataset. In other words, our prior guess for the cluster-specific coefficients of the Cox regression model correspond to the coefficients from a Cox model fitted to the entire dataset, with the associated uncertainty represented by a constant $c > 1$ multiplied by the variance estimates. Following the arguments of \cite{roy2018bayesian}, we set $c = n/5$.

\noindent For $b \in \{1,\ldots, B-1\}$, we set $\lambda_b \stackrel{ind}{\sim} \text{Gamma}\big( [v_{b+1}- v_{b}]\lambda_{*}w, [v_{b+1}- v_{b}]w\big)$. For $b = B$, we set $\lambda_B \stackrel{ind}{\sim} \text{Gamma}\big( \lambda_{*}w_B, w_B\big)$. The prior involves two tuning constants, $\lambda_{*}$ and $w$ (as well as $w_B$). For large $w$ (and $w_B$), the Gamma distributions will concentrate on $\lambda_{*}$ and the distribution of baseline hazard will mimic that of an Exp($\lambda_{*}$) distribution. On the other hand, small $w$ gives a larger prior variance on $\lambda_b$,
allowing for substantial uncertainty. A large value of $\lambda_{*}$ corresponds to a prior expectation that events happen at a high rate across
all intervals, while a small value indicates that events are somewhat rare. Hence the value of $\lambda_{*}$ can be set by eliciting a percentile for the median event time from the data. In summary, we assume:
\begin{align*}
            H_{0\beta} &\sim  \prod_{b=1}^{B-1}\Big\{\underbrace{\text{Gamma}\big( [v_{b+1}- v_{b}]\lambda_{*}w, [v_{b+1}- v_{b}]w\big)}_{\lambda_b, b = 1, \ldots, B-1}\Big\} \times \underbrace{\text{Gamma}\big( \lambda_{*}w_B, w_B\big)}_{\lambda_B}  \times 
             \prod_{p=1}^{P} \underbrace{N\big(\beta_{0,p}, c\sigma^{2,T_i}_{0,p} \big)}_{\beta_{i,p}, p = 1, \ldots, P},
        \end{align*}
where $P$ denotes the number of predictors, including both baseline and time-varying variables, in the local Cox regression model.

\subsection{Priors for the $\theta$-level parameters}
At the $\theta-$level, we have regression $\bigl(\boldsymbol{\theta}_i^{M}, \boldsymbol{\theta}_i^{L}, \boldsymbol{\theta}_i^{Z}\bigr) $ and spline $\bigl( \boldsymbol{\eta}^{M}_{i}, \boldsymbol{\eta}^{L}_{i}, \boldsymbol{\eta}^{Z}_{i}\bigr)$ coefficients for modeling the time-varying variables and regression $\big(\boldsymbol{\theta}^{\textbf{L}_0}\big)$ coefficients for modeling the baseline confounders. Without loss of generality, let us assume that we have a binary time-varying confounder and a continuous time-varying mediator. We assume the following base measures for these parameters:
\begin{equation}\label{Eq:thetaLevelPriors}
        \begin{aligned}
             H_{0\theta|\beta} &\sim   \prod_{q}^{n_{\boldsymbol{\ell}_{0}}+1}\Big\{ \underbrace{ N\big(\theta^{Z}_{0,q}, c\sigma^{2,\theta^{Z}}_{0,q}\big)}_{\theta_{i,q}^{Z}}  \times \underbrace{N\big(\theta^{L}_{0,q}, c\sigma^{2,\theta^{L}}_{0,q}\big)}_{\theta_{i,q}^{L}} 
             \times \underbrace{N\big(\theta^{M}_{0,q}, c\sigma^{2,\theta^{M}}_{0,q}\big)}_{\theta_{i,q}^{M}}\Big\} \times \underbrace{\text{Inv-Gamma}\big(a^{M},b^{M} \big)}_{\sigma_{i}^{2,M}}   \times 
             \\&
            \prod_{d=1}^{D}\Big\{    \underbrace{N\big(\eta_{0,d}^{Z}, c\sigma_{0,d}^{2,\eta^{Z}}\big)}_{\eta^{Z}_{i,d}}  \times \underbrace{N\big(\eta_{0,d}^{L}, c\sigma_{0,d}^{2,\eta^{L}}\big)}_{\eta^{L}_{i,d}}  \times \underbrace{N\big(\eta_{0,d}^{M}, c\sigma_{0,d}^{2,\eta^{M}}\big)}_{\eta^{M}_{i,d}} \Big\}  \times     f_0\big(\theta_i^{\textbf{L}_0}\big),
        \end{aligned}
\end{equation}
where $n_{\boldsymbol{\ell}_{0}}$ denotes the number of baseline confounders in the local regression. In the specification above, $f_0(\theta_i^{\textbf{L}_0}) = \prod_{s = 1}^{n_{\boldsymbol{\ell}_{0}}}f_{0,s}(\theta_{i,s}^{\textbf{L}_0})$ for $n_{\boldsymbol{\ell}_{0}}$ baseline covariates with 
            \begin{align*}
                f_{0,s}\big( \theta_{i,s}^{\textbf{L}_0}\big) = \begin{cases}
                    \underbrace{\text{Inv-Gamma}\big(a_{\boldsymbol{\ell}_{0}},b_{\boldsymbol{\ell}_{0}} \big)}_{\sigma_{i,s}^{2,\textbf{L}_{0}}} \times \underbrace{N\big(\mu_{\boldsymbol{\ell}_{0}},\sigma_{\boldsymbol{\ell}_{0}}^2\big)}_{\theta_{i,s}^{\textbf{L}_0}}  \quad\quad\quad \text{for continuous baseline covariates}\\
                 \underbrace{\text{Beta}\big(a_{\boldsymbol{\ell}_{0}},b_{\boldsymbol{\ell}_{0}}\big)}_{\theta_{i,s}^{\textbf{L}_0}\,\equiv\, p_{i,s}^{\textbf{L}_0}}\quad\quad\quad \quad\quad\quad \quad\quad\quad \quad\quad\quad \text{     for binary baseline covariates}.
                \end{cases}
            \end{align*}
We center and scale the base measures using maximum likelihood estimates from ordinary linear or logistic regressions applied to all of the data, again setting $c = n/5$. We assume conjugate priors for the baseline confounder parameters.

 \subsection{Priors for the EDP mass parameters}
The number of clusters in the EDP depends on the concentration parameters $\alpha^{\beta}$ and $\alpha^{\theta|\beta}$, where lower values of these parameters indicate fewer clusters. Therefore, selecting appropriate values for these concentration parameters is crucial in determining reasonable truncation levels $N$ and $M$ for EDP mixtures. Since the square-breaking weights decay exponentially, $N$ and $M$ can typically be chosen to be relatively small. Chapter 6 of \cite{daniels2023bayesian} provides a simple calculation showing that, on average, when the concentration parameter $\alpha = 1$, the first 20 stick-breaking weights in the case of Dirichlet process priors sum to approximately 1. Among these first 20 weights, the latter 10 are, in expectation, approximately 0 when $\alpha = 1$. Utilizing this observation, we set $\alpha^{\beta} = 1$ in our work and choose $N = 10$. Consequently, the mass parameter $\alpha^{\beta}$ is treated as deterministic in this article.  

\noindent For the mass parameter $\alpha_r^{\theta|\beta}$ (one per outer $\beta$-cluster, $r = 1, \ldots, N$), which is random, we assume independent prior distributions $\alpha_r^{\theta|\beta} \stackrel{iid}{\sim} \mathrm{Gamma}(a_{\theta}, b_{\theta})$ for $r = 1, \ldots, N$.
The choice of the inner-level truncation value $M$ in our work is data-dependent. After evaluating candidate values of $M$ from $1$ to $10$ when fitting the proposed model to the ARIC cohort study, we observe that only $2$--$3$ inner clusters typically have non-negligible membership probabilities. Therefore, we set $M = 4$ as a conservative choice that also offers computational advantages.

\subsection{Random effects}
The random effects $b_i^{M}, b_i^{L},$ and $b_i^{Z}$, which are not included in the EDP prior, are assumed to follow mean-zero normal distributions:
$$b_i^{M} \sim N\big(0, \tau_{M}^{2}\big), b_i^{L} \sim N\big(0, \tau_{L}^{2}\big), \text{ and } b_i^{Z} \sim N\big(0, \tau_{Z}^{2}\big).$$
The corresponding variances of these random intercepts, $\tau_{M}^{2}, \tau_{L}^{2}$, and  $\tau_{Z}^{2}$, are assigned conjugate inverse-gamma priors:
\begin{align*}
    \tau_{M}^{2} &\sim IG\big(\alpha_{\tau_{M}^{2}}, \beta_{\tau_{M}^{2}} \big),
    \\
    \tau_{L}^{2} &\sim IG\big(\alpha_{\tau_{L}^{2}}, \beta_{\tau_{L}^{2}} \big), \text{ and }
    \\
    \tau_{Z}^{2} &\sim IG\big(\alpha_{\tau_{Z}^{2}}, \beta_{\tau_{Z}^{2}} \big).
\end{align*}

\section{Appendix C: EDPM truncation approximation and posterior computation}\label{AppendixC}

\subsection{Truncation approximation}
Let
\begin{equation}\label{Eq:EDPTruncApprox}
\resizebox{\textwidth}{!}{$
    \begin{aligned}
        T_i, \delta_i \big| T_i > a_{i,n_i}, \textbf{M}_{i}(a_{i,n_i}), \textbf{L}_{i}(a_{i,n_i}), \textbf{Z}_{i}(a_{i,n_i}), \textbf{L}_{i,0}; \boldsymbol{\beta}_i &\sim F_{t_i,\delta_i}\big(. \big| \textbf{m}_{i}(a_{i,n_i}), \boldsymbol{\ell}_{i}(a_{i,n_i}), \textbf{z}_{i}(a_{i,n_i}), \boldsymbol{\ell}_{i,0}; \boldsymbol{\beta}_i \big)
       \\
        M_{i}(a_{i,j})\big| T_i> a_{i,j}, \textbf{M}_{i}(a_{i,j-1}), \textbf{L}_{i} (a_{i,j}), \textbf{Z}_{i}(a_{i,j}), \textbf{L}_{i,0}, b_i^{M}; \boldsymbol{\theta}^{M}_{i}, \boldsymbol{\eta}^{M}_{i} &\sim F_{m(a_{i,j})} \big(.\big|a_{i,j}, \textbf{m}_{i}(a_{i,j-1}), \boldsymbol{\ell}_{i} (a_{i,j}), \textbf{z}_{i}(a_{i,j}), \boldsymbol{\ell}_{i,0},b_i^{M}; \boldsymbol{\theta}^{M}_{i}, \boldsymbol{\eta}^{M}_{i}\big)
        \\
        L_{i}(a_{i,j})\big|T_i> a_{i,j}, \textbf{M}_{i}(a_{i,j-1}), \textbf{L}_{i} (a_{i,j-1}), \textbf{Z}_{i}(a_{i,j}), \textbf{L}_{i,0},b_i^{L}; \boldsymbol{\theta}^{L}_{i}, \boldsymbol{\eta}^{L}_{i} &\sim F_{\ell(a_{i,j})} \big(.\big|a_{i,j}, \textbf{m}_{i}(a_{i,j-1}), \boldsymbol{\ell}_{i} (a_{i,j-1}), \textbf{z}_{i}(a_{i,j}), \boldsymbol{\ell}_{i,0},b_i^{L}; \boldsymbol{\theta}^{L}_{i}, \boldsymbol{\eta}^{L}_{i}\big)
        \\
        Z_{i}(a_{i,j})\big|T_i> a_{i,j}, \textbf{M}_{i}(a_{i,j-1}), \textbf{L}_{i} (a_{i,j-1}), \textbf{Z}_{i}(a_{i,j-1}), \textbf{L}_{i,0}, b_i^{Z}; \boldsymbol{\theta}^{Z}_{i}, \boldsymbol{\eta}^{Z}_{i} &\sim F_{z(a_{i,j})} \big(.\big|a_{i,j}, \textbf{m}_{i}(a_{i,j-1}), \boldsymbol{\ell}_{i} (a_{i,j-1}), \textbf{z}_{i}(a_{i,j-1}), \boldsymbol{\ell}_{i,0}, b_i^{Z}; \boldsymbol{\theta}^{Z}_{i}, \boldsymbol{\eta}^{Z}_{i}\big)
        \\ \textbf{L}_{i,0};\boldsymbol{\theta}^{\textbf{L}_0}_{i} &\sim F_{\boldsymbol{\ell}_{0}}\big(.\big| \boldsymbol{\theta}^{\textbf{L}_0}_{i} \big), \\
        (\boldsymbol{\beta}_i, \boldsymbol{\theta}_{i})|H &\sim H  \\
        H &\sim \mathcal{H}_{NM}
    \end{aligned}
$}
\end{equation}
In Equation (\ref{Eq:EDPTruncApprox}), $H \sim \mathcal{H}_{NM}$ implies 
    $ H = \sum_{r=1}^{N} \sum_{s=1}^{M}\xi_{r} \xi_{s|r} \delta_{\beta_r^{*},\theta_{s|r}^{*}},$ where:
    \begin{align*}
        \xi_r =& \xi_r^{'}\prod_{t <r} (1- \xi_t^{'}), \quad   \xi_t^{'} \stackrel{iid}{\sim} \mathrm{Beta}(1, \alpha^{\beta}) \text{ for } t = 1, \ldots, N-1, \quad \xi_N^{'} = 1, \quad \beta^{*}_r \mathop{\sim}\limits^{iid} H_{0\beta} \\
        \xi_{s|r} =& \xi_{s|r}^{'}\prod_{t <s} (1- \xi_{t|r}^{'}), \quad  \xi_{t|r}^{'} \stackrel{iid}{\sim} \mathrm{Beta}(1, \alpha_r^{\theta|\beta}) \text{ for } t = 1, \ldots, M-1,  \quad \xi_{M|r}^{'} = 1, \quad \theta^{*}_{s|r} \mathop{\sim}\limits^{iid} H_{0\theta|\beta}.
    \end{align*}
    Now suppose we have a fixed age grid $a_1 < a_2 < \ldots <  a_{k} < \ldots < a_{K} < \ldots$ such that $a_{K} \leq t_i < a_{K+1}$. With this specification, assuming local independence among ages in the grid within the inner cluster, we can write the model for the joint density using the EDPM as a finite mixture at $k = K$:
\begin{equation}
        \begin{aligned}
            f_H\big(t_i, \textbf{m}_i(a_{K}),\boldsymbol{\ell}_i(a_{K}),\textbf{z}_i(a_{K}),\boldsymbol{\ell}_{i,0}\big)
            &  = \sum_{r=1}^{N}\Bigl\{\xi_r\, p\bigl(t_i\big| \boldsymbol{\ell}_{i,0};\beta_r^{*}\bigr) \times
            \sum_{s=1}^{M} \xi_{s|r}
            \prod_{k=1}^{K} \big\{ p\bigl(m_i(a_{k})\big| \boldsymbol{\ell}_{i,0},b_i^{M} ;\theta_{s|r}^{*}\bigr) \times
            \\&
            p\bigl(\ell_i(a_{k})\big| \boldsymbol{\ell}_{i,0},b_i^{L} ; \theta_{s|r}^{*}\bigr) \times  p\bigl(z_i(a_{k})\big| \boldsymbol{\ell}_{i,0},b_i^{Z} ;\theta_{s|r}^{*}\bigr)\big\} \times   p\bigl(\boldsymbol{\ell}_{i,0};\theta_{s|r}^{*}\bigr)\Bigr\},
        \end{aligned}
\end{equation}
where $p(.)$ denotes the corresponding density associated with distributions $F(.)$ in (\ref{Eq:EDPTruncApprox}).

\noindent \cite{burns2023truncation} show that $\mathcal{H}_{NM}$ converges almost surely to an enriched Dirichlet process with base distribution $H_{0\beta} \times H_{0\theta|\beta}$ and precision parameters  $\alpha^{\beta}$ and $\alpha^{\theta|\beta}$, respectively. We obtain draws from the posterior distribution of all parameters through blocked Gibbs sampler using the truncation approximation introduced in \cite{burns2023truncation}.  We use the term  $\beta$-cluster to indicate top-level clusters based on the parameters of the outcome model. Similarly, a  $\theta$-cluster denotes a subcluster (based on the variables other than the outcome) nested within a $\beta$-cluster. The subjects from a $\beta$-cluster share the value of $\beta$ and the subjects from a $\theta$-cluster within the $\beta$-cluster share the value of $\theta$ among each other. For subject $i$, we define the cluster membership index $V_i = (V_i^{\beta}, V_i^{\theta} )$. Here, the value of $V_i^{\theta}$ is only meaningful in conjunction with $V_i^{\beta}$, as it describes the subcluster assignment within the $\beta$-cluster.

\subsection{Metropolis-Hastings within blocked Gibbs sampler}
At each iteration, we sample from the following conditional distributions iteratively:
\begin{enumerate}
            \item Update cluster membership, i.e., sample from the conditional distribution of \\ $\textbf{V}^{\beta}, \textbf{V}^{\theta}|\boldsymbol{\xi}_{r}, \boldsymbol{\xi}_{s|r}, \boldsymbol{\beta}^{*}, \boldsymbol{\theta}^{*}, \textbf{b}^{M}, \textbf{b}^{L}, \textbf{b}^{Z}, \textbf{T}, \boldsymbol{\delta}, \textbf{M},\textbf{L}, \textbf{Z}, \textbf{L}_0$: \\
            For each subject $i$, $V^{\beta}_i, V^{\theta}_i|\boldsymbol{\xi}_{r}, \boldsymbol{\xi}_{s|r}, \boldsymbol{\beta}^{*}, \boldsymbol{\theta}^{*}, b_i^{M}, b_i^{L}, b_i^{Z}, t_i, \delta_i,  \textbf{m}(a_{i,n_i}),\boldsymbol{\ell}(a_{i,n_i}), \textbf{z}(a_{i,n_i}), \boldsymbol{\ell}_{i,0}$ is sampled from a multinomial  distribution with the  probability that subject $i$ is assigned to $\beta$-cluster $r$ (of the $N$ $\beta$-clusters) and $\theta$-cluster $s$ (of the $M$ $\theta$-clusters) is proportional to
            \begin{align*}
                p_{i,r, s|r}   &\propto \xi_r \xi_{s|r} f\bigl(t_i, \delta_i\big| \boldsymbol{\ell}_{i,0};\xi_r, \xi_{s|r}, \beta_r^{*}, \theta_{s|r}^{*}\bigr) \times 
                 \Big\{ \prod_{j =1}^{n_i} f\bigl(m(a_{ij})\big|\boldsymbol{\ell}_{i,0},b_i^{M}; \boldsymbol{\theta}_{s|r}^{M,*}, \sigma_{s|r}^{M,2,*}, \boldsymbol{\eta}^{M,*}_{s|r}\bigr)  \\& f\bigl(\ell(a_{ij})\big|\boldsymbol{\ell}_{i,0},b_i^{L}; \boldsymbol{\theta}_{s|r}^{L,*} \boldsymbol{\eta}^{L,*}_{s|r}\bigr)  f\bigl(z(a_{ij})\big|\boldsymbol{\ell}_{i,0},b_i^{Z}; \boldsymbol{\theta}_{s|r}^{Z,*} \boldsymbol{\eta}^{Z,*}_{s|r}\bigr) \Big\} \times f(\boldsymbol{\ell}_{i,0};  \theta_{s|r}^{L_0,*})
            \end{align*}
            \item Update parameters at the $\beta$-level cluster using Metropolis-Hastings update, i.e., sample from the conditional distribution of $\boldsymbol{\beta}^{*}|\textbf{V}^{\beta}, \alpha^{\beta}, \textbf{T}, \boldsymbol{\delta}$: \\
            For each $r = 1, \ldots, N$ (indexing $\beta$-level parameters), we update the parameters $\lambda_{b}, b = 1,\ldots B$ (piecewise constant baseline hazard parameters), and $\beta_{p}, p = 1,\ldots P$ (regression parameters in the local Cox model), using a random-walk Metropolis-Hastings algorithm with Normal proposal distributions.

            \item Update regression and spline parameters at the $\theta$-level cluster, i.e., sample from the conditional distribution of $\boldsymbol{\theta}^{*}|\textbf{V}^{\beta}, \textbf{V}^{\theta}, \alpha^{\theta|\beta},\textbf{b}^{M}, \textbf{b}^{L}, \textbf{b}^{Z}, \textbf{M}, \textbf{L}, \textbf{Z}, \textbf{L}_0$:
            \begin{equation}\label{Eq:ThetaLevelRegPost}
                \resizebox{\textwidth}{!}{$
                \begin{aligned}  
                    & f\bigl(\boldsymbol{\theta}_{s|r}^{M,*}, \boldsymbol{\theta}_{s|r}^{L,*},\boldsymbol{\theta}_{s|r}^{Z,*}, \boldsymbol{\eta}^{M,*}_{s|r}, \boldsymbol{\eta}^{L,*}_{s|r}, \boldsymbol{\eta}^{Z,*}_{s|r},  \sigma_{s|r}^{2,M,*}, \boldsymbol{\theta}_{s|r}^{L_{0},*}\big|  \textbf{b}^{M}, \textbf{b}^{L}, \textbf{b}^{Z}, \textbf{V}^{\beta}, \textbf{V}^{\theta}, \alpha^{\theta|\beta}, \textbf{M},\textbf{L},\textbf{Z},\textbf{L}_0 \bigr)
                    \\& \propto \underbrace{H_{0\theta|\beta}(\mathrm{d}\boldsymbol{\theta}_{s|r}^{M,*}) \times H_{0\theta|\beta}(\mathrm{d}\boldsymbol{\theta}_{s|r}^{L,*})
                    \times H_{0\theta|\beta}(\mathrm{d}\boldsymbol{\theta}_{s|r}^{Z,*}) \times  H_{0\theta|\beta}(\mathrm{d}\boldsymbol{\eta}^{M,*}_{s|r}) \times H_{0\theta|\beta}(\mathrm{d}\boldsymbol{\eta}^{L,*}_{s|r}) \times H_{0\theta|\beta}(\mathrm{d}\boldsymbol{\eta}^{Z,*}_{s|r})
                    \times H_{0\theta|\beta}(\mathrm{d}\sigma_{s|r}^{2,M,*})
                    \times H_{0\theta|\beta}(\mathrm{d}\boldsymbol{\theta}_{s|r}^{L_{0},*})}_{\text{priors}} \times
                    \\&
                     \underbrace{\prod_{i: V_i^{\beta} = r, V_i^{\theta} = s}
                     f\big(\textbf{m}(a_{i,n_i}), \boldsymbol{\ell}(a_{i,n_i}), \textbf{z}(a_{i,n_i}), \boldsymbol{\ell}_{i,0}\big| b_i^{M}, b_i^{L}, b_i^{Z}; \boldsymbol{\theta}_{s|r}^{M,*}, \boldsymbol{\theta}_{s|r}^{L,*}, \boldsymbol{\theta}_{s|r}^{Z,*},
                    \boldsymbol{\theta}_{s|r}^{L_{0},*}, \sigma_{s|r}^{2,M,*}, \boldsymbol{\eta}^{M,*}_{s|r}, \boldsymbol{\eta}^{L,*}_{s|r}, \boldsymbol{\eta}^{Z,*}_{s|r},
                    \alpha^{\theta|\beta}\big)}_{ \text{likelihood contribution}}.
                    \end{aligned}
                    $}
        \end{equation}
                For each $r = 1, \ldots, N$ (indexing $\beta$-level parameters), $s = 1, \ldots, M$ (indexing $\theta$-level parameters within $r$), $j = 1, \ldots, n_i$ (indexing visit ages for subject $i$), and $k = 1, \ldots, K $ (indexing number of baseline covariates), the parameters in the regression models for the longitudinal variables are updated using a Gibbs sampler if the corresponding data is continuous, or a Metropolis-Hastings step if the data is binary.
                Assuming, without loss of generality, a binary time-varying confounder and a continuous time-varying mediator, the regression and spline parameters $\big(\boldsymbol{\theta}_{s|r}^{M}, \boldsymbol{\eta}_{s|r}^{M}\big)$,  along with the variance parameter $\big(\sigma_{s|r}^{2,M}\big)$ from the mediator model are updated from conjugate normal and inverse-gamma distributions, respectively. For binary time-varying confounders and exposures, the parameters from the corresponding probit regressions are updated using a random-walk Metropolis-Hastings algorithm with Normal proposal distributions.

                The baseline covariate model parameters are updated using conjugate normal distributions for continuous baseline covariate model parameters and beta distributions for binary baseline covariate model parameters, with the posterior distribution proportional to:
                \begin{align*}
                    & f\bigl( \boldsymbol{\theta}_{s|r}^{L_{0},*}  \big|  \textbf{V}^{\beta}, \textbf{V}^{\theta}, \alpha^{\theta|\beta}, \textbf{L}_0 \bigr)
                      \propto    \underbrace{H_{0\theta|\beta}(\mathrm{d}\boldsymbol{\theta}_{s|r}^{L_{0},*} )}_{\text{prior}}  \times
                      \underbrace{\prod_{i: V_i^{\beta} = r, V_i^{\theta} = s}
                     f\big(\boldsymbol{\ell}_{i,0}\big|  \theta_{s|r}^{L_0,*},  \alpha^{\theta|\beta}\big)}_{\text{likelihood contribution}},
                \end{align*}
                where the corresponding likelihood contributions are from all subjects $i$ in outer cluster $r$ and inner cluster $s$. 
            \item Update the weights at the $\beta$-level, i.e., sample from the conditional distribution of $\boldsymbol{\xi}_{r}|\textbf{V}^{\beta}, \alpha^{\beta}$  using $\xi_1= \xi_1^{'}, \xi_r= \xi^{'}_r\prod_{t =1}^{r-1} (1- \xi_t^{'}), r = 2, \ldots, N$ where 
            $$ \xi_{\tilde{r}}^{'}|\textbf{V}^{\beta}, \alpha^{\beta} \sim \text{Beta}\Bigl(n_{\tilde{r}}+1, \alpha^{\beta} + \sum_{w = \tilde{r}+1}^{N}n_w\Bigr) $$
             with $n_{\tilde{r}}$ denoting the number of subjects currently in the $\tilde{r}^{th}$ $\beta-$cluster for $\tilde{r} = 1, \ldots, N-1$ and $\xi_N^{'} = 1$.
            \item Update the weights at the $\theta$-level, i.e., sample from the conditional distribution of $\boldsymbol{\xi}_{s|r}|\textbf{V}^{\beta}, \textbf{V}^{\theta}, \alpha^{\theta|\beta}$ , using  $\xi_{1|r} = \xi_{1|r}^{'}, \xi_{s|r} = \xi_{s|r}^{'}\prod_{t =1}^{s-1} (1- \xi_{t|r}^{'}) ,s = 2,\ldots, M$  for every $r = 1,\ldots, N$ where
            $$\xi_{\tilde{s}|r}^{'}|\textbf{V}^{\theta}, \alpha_{\tilde{s}}^{\theta|\beta} \sim \text{Beta}\Bigl(n_{r\tilde{s}} +1, \alpha_{\tilde{s}}^{\theta|\beta} + \sum_{w = \tilde{s}+1}^{M} n_{rw} \Bigr) $$
            with $n_{r\tilde{s}}$ denoting the number of subjects currently in the $\tilde{s}^{th}$ $\theta-$cluster within the $r^{th}$ $\beta-$cluster for $\tilde{s} = 1, \ldots, M-1$ and $\xi_{M|r}^{'} = 1$.

            \item Update the concentration parameter at the $\theta$-level (the outer-level mass parameter $\alpha^{\beta} = 1$ is treated as deterministic, as specified in Appendix~\ref{AppendixB}, so no update is required for it): \\
            For a $\mathrm{Gamma}(a_{\theta},b_{\theta} )$ prior on $\alpha_r^{\theta|\beta}$, sample from the conditional distribution $$\alpha_r^{\theta|\beta}\big|\boldsymbol{\xi}_{s|r} \sim \mathrm{Gamma}\Bigl(a_{\theta} + M_r^{\mathrm{occ}} - 1,\; b_{\theta} - \sum_{s =1}^{M-1} \log(1-\xi_{s|r}^{'}) \Bigr),$$
            for each $r = 1,\ldots, N$ corresponding to $N$ inner level mass parameters, where $M_r^{\mathrm{occ}}$ denotes the number of occupied $\theta$-clusters within outer cluster $r$ (under the truncation approximation, all $M$ slots are pre-instantiated, so $M_r^{\mathrm{occ}} = M$ in our implementation).
            
             \item Update the random effects ($ b_i^{M},b_i^{L},b_i^{Z}$): \\
            Recall that the random effects are not included in the EDPM model. At each MCMC iteration, conditional on the data, the cluster memberships $(V_i^{\beta}, V_i^{\theta})$, and the regression parameters $(\boldsymbol{\theta}_{s|r}^{M,*}, \boldsymbol{\eta}_{s|r}^{M,*}, \sigma_{s|r}^{2,M,*})$,  the new random intercept $b_i^{M,*}$ from the mediator (continuous data) model for subject $i$ is updated from the conjugate Normal distribution after taking the residuals from the current fit. Note that if we do not condition on the cluster memberships $(V_i^{\beta}, V_i^{\theta})$, $b_i^{M}$ must instead be updated using a Metropolis-Hastings step, since marginally over cluster, membership the mediator follows a finite mixture of Normal distributions. The random intercept variance parameter $\tau_{M}^{2,*}$ is updated from conjugate Inverse-Gamma distribution with shape $\alpha_{\tau_{M}^{2}} + \frac{n}{2}$ and rate $\beta_{\tau_{M}^{2}} + \frac{1}{2} \sum_{i=1}^{n}\big(b_i^{M,*}\big)^2$, where $\alpha_{\tau_{M}^{2}}$ and $\beta_{\tau_{M}^{2}}$ are the shape and rate parameters, respectively, from the prior distribution of $\tau_{M}^{2}$.

            Similarly, for subject $i$, the new random intercepts $ b_i^{L,*}$ and $b_i^{Z,*}$ from the confounder and exposure (binary data) models, respectively,  are updated using a Metropolis-Hastings step, conditional on the data, the cluster memberships $(V_i^{\beta}, V_i^{\theta})$, and the regression parameters $(\boldsymbol{\theta}_{s|r}^{L}, \boldsymbol{\eta}_{s|r}^{L}, \boldsymbol{\theta}_{s|r}^{Z}, \boldsymbol{\eta}_{s|r}^{Z})$. This is because the binary data are modeled using probit regressions. Given these random intercepts, 
            the variances $\tau_{L}^{2}$, and  $\tau_{Z}^{2}$ are updated using a Metropolis-Hastings step as well.
 \end{enumerate}

\section{Appendix D: Spline construction for the longitudinal data and Poisson approximation trick for the survival model}\label{AppendixD}

\subsection{Thin-plate spline construction}

\noindent The observed data models for the time-varying variables $Z_i(a_{ij})$, $L_i(a_{ij})$, and $M_i(a_{ij})$ utilize splines to allow for nonlinearities across time and to allow for prediction at any age $a$. In this article, we use splines with $D$ pre-specified knots at $q_1 \leq \ldots \leq q_D$. Following the modeling framework from \cite{zeldow2021functional}, we implement penalized, thin-plate splines that have good mixing properties in Bayesian analysis \citep{crainiceanu2005bayesian}. To construct thin-plate splines, let $\mathcal{B}_{i,D}$ be the $n_i \times D$ matrix which contains basis functions evaluated at the $n_i$ visit ages for individual $i$. The $j$th row of $\mathcal{B}_{i,D}$ is populated by
$$\Bigl\{|a_{ij} - q_1|^3, \ldots, |a_{ij} - q_D|^3\Bigr\} , \quad j \in \{1, \ldots, n_i\}.$$
This means, we have:
    \begin{align*}
        \mathcal{B}_{i,D}  =
        \begin{bmatrix}
            |a_{i1} - q_1|^3 & |a_{i1} - q_2|^3 & \ldots & |a_{i1} - q_
            {D-1}|^3  & |a_{i1} - q_D|^3 \\
            |a_{i2} - q_1|^3 & |a_{i2} - q_2|^3 & \ldots & |a_{i2} - q_
            {D-1}|^3  & |a_{i2} - q_D|^3 \\
            \ldots  & \ldots  & \ldots  & \ldots  & \ldots   \\
            |a_{i,n_i} - q_1|^3 & |a_{i,n_i} - q_2|^3 & \ldots & |a_{i,n_i} - q_
            {D-1}|^3  &  |a_{i,n_i} - q_D|^3
        \end{bmatrix}_{n_i \times D}.
    \end{align*}
The penalty matrix $\Omega_D$ is a $D\times D$ matrix where the $(f,g)^{th}$ entry is $|q_f - q_g|^{3}$. The penalty matrix prevents overfitting by penalizing the coefficients of $\mathbfcal{B}_{D}$, and is given by:
    \begin{align*}
        \Omega_{D}  =&
        \begin{bmatrix}
            0 & |q_1 - q_2|^{3} &  \ldots & |q_1 - q_{D-1}|^{3} & |q_1 - q_{D}|^{3}  \\
            |q_2 - q_1|^{3} & 0 &  \ldots & |q_2 - q_{D-1}|^{3} & |q_2 - q_{D}|^{3} \\
            \ldots & \ldots &    0  & \ldots & \ldots \\
            |q_{D-1} - q_1|^{3} & |q_{D-1} - q_2|^{3} & \ldots & 0 & |q_{D-1} - q_D|^{3} \\
            |q_{D} - q_1|^{3} & |q_{D} - q_2|^{3} & \ldots  & |q_{D} - q_{D-1}|^{3} & 0 \\
        \end{bmatrix}_{D\times D}.
    \end{align*}
Finally, we have the $n_i \times D$ matrix $\mathbfcal{B}_i= \mathcal{B}_{i,D}\Omega_{D}^{-1/2}$. For the $j^{th}$ visit age, we may slice the matrix $\mathbfcal{B}_i$ as $\mathbfcal{B}_i(a_{ij})= \mathcal{B}_{i,D}[j, :]\Omega_{D}^{-1/2}$. For instance, $\mathbfcal{B}_i(a_{ij})$ is equal to:
\begin{align*}
        \begin{bmatrix}
            |a_{ij} - q_1|^3 & |a_{ij} - q_2|^3 &  \ldots &  |a_{ij} - q_D|^3
        \end{bmatrix}_{1\times D}
        \begin{bmatrix}
            0 & |q_1 - q_2|^{3} &  \ldots & |q_1 - q_{D-1}|^{3} & |q_1 - q_{D}|^{3}  \\
            |q_2 - q_1|^{3} & 0 &  \ldots & |q_2 - q_{D-1}|^{3} & |q_2 - q_{D}|^{3} \\
            \ldots & \ldots &    0  & \ldots & \ldots \\
            |q_{D-1} - q_1|^{3} & |q_{D-1} - q_2|^{3} & \ldots & 0 & |q_{D-1} - q_D|^{3} \\
            |q_{D} - q_1|^{3} & |q_{D} - q_2|^{3} & \ldots  & |q_{D} - q_{D-1}|^{3} & 0 \\
        \end{bmatrix}_{D\times D}^{-1/2}.
    \end{align*}

\subsection{Poisson approximation for the piecewise exponential survival model}

\noindent The survival data is denoted by $(t_i, \delta_i)$ where $t_i$ is the observed event time for subject $i$ and $\delta_i =1$ if $t_i$ is the true death time and $0$ if the death time is right censored at $t_i$. Under the Cox regression model in (\ref{Eq:CoxModel}), the likelihood of the survival data is given by:
\begin{equation}
    \begin{aligned}
        L_{surv}(\lambda_0, \boldsymbol{\beta}_i) =& \prod_{i=1}^{n} f(T_i,\delta_i, \boldsymbol{\ell}_{i,0}; \lambda_0, \boldsymbol{\beta}_i)   
        \\
        =&  \prod_{i=1}^{n} \Big\{\lambda_0\bigl(t\bigr)exp\bigl\{ \boldsymbol{\beta}_i \boldsymbol{\ell}_{i,0}^\top  \bigr\} \Big\}^{\delta_i}exp\Big\{ -exp\bigl\{ \boldsymbol{\beta}_i \boldsymbol{\ell}_{i,0}^\top  \bigr\}\Lambda_i\bigl(t\big|\boldsymbol{\ell}_{i,0}, \boldsymbol{\beta}_i\bigr)\Big\}, 
    \end{aligned}
\end{equation}
where $\Lambda_i\bigl(t\big|\boldsymbol{\ell}_{i,0}, \boldsymbol{\beta}_i\bigr) = \int_{0}^{t} \lambda_i\bigl(u\big|\boldsymbol{\ell}_{i,0}, \boldsymbol{\beta}_i\bigr) du $ denotes the cumulative baseline hazard. In this article, we assume the baseline hazard, $\lambda_0\bigl(t\bigr)$, to be piecewise constant on a partition composed of $B$ disjoint intervals, yielding the piecewise exponential model. In other words, we set:
\begin{align*}
    \lambda_0\bigl(t\bigr) =& \sum_{b=1}^{B}\lambda_b I_{\{ v_{b-1}\leq t < v_{b}\}},
\end{align*}
where $v_{0} = 0$ and $v_{B} = max(t_i)$. Suppose $\lambda = \{\lambda_1, \ldots, \lambda_B\}$ where $\lambda_b \geq 0$ for $b \in \{1, \ldots, B\}$. Following the ideas from \cite{christensen2010bayesian}, the likelihood can be re-expressed as:
\begin{equation}\label{poiss_lik}
    \begin{aligned}
        L_{surv}(\lambda_0, \boldsymbol{\beta}_i) =&   \prod_{i=1}^{n} \prod_{b=1}^{B}\Big\{\Theta \big[i,b\big] \Big\}^{N\big[i,b\big]} \times exp\Big\{-\Theta \big[i,b\big] \Big\},
    \end{aligned}
\end{equation}
where, for $b \in \{1, \ldots, B\}$, we define $N\big[i,b\big] = \begin{cases}
    1 \quad \text{if } t_i \in [v_{b},v_{b+1} ) \text{ and } \delta_i = 1 \\
    0 \quad \text{otherwise}
\end{cases}$ .  
Similarly, we define 
\begin{equation}
    \Theta\big[i,b\big] = exp\bigl\{ \boldsymbol{\beta}_i \boldsymbol{\ell}_{i,0}^\top  \bigr\} H\big[i,b\big] \lambda_b,
\end{equation} where $H\big[i,b\big] = \begin{cases}
    v_{b+1} - v_{b} \quad \text{if } t_i > v_{b+1}  \\
    t_i - v_{b} \quad \text{if } t_i \in [v_{b},v_{b+1} )\\
    0 \quad t_i < v_{b}
\end{cases}$ . The likelihood in (\ref{poiss_lik}) yields the following log-likelihood component for each subject $i$:
\begin{equation}\label{poiss_log_lik}
    \begin{aligned}
        l_{i,surv} =&   \sum_{b=1}^{B}log\Big[\Big\{\Theta[i,b] \Big\}^{N[i,b]} \times exp\Big\{-\Theta[i,b] \Big\}\Big].
    \end{aligned}
\end{equation}
Each summand in (\ref{poiss_log_lik}) can be recognized as the log-likelihood for $N\big[i,b\big]\stackrel{ind}{\sim}\text{Poisson}\big(\Theta\big[i,b\big]\big)$ where all the Poisson data happen to be either $0$ or $1$. This enables us to approximate the parameters in the Cox regression model given in (\ref{Eq:CoxModel}) using the parameter estimates obtained from fitting a Poisson model to $N\big[i,b\big]$  with rate parameter $\Theta\big[i,b\big]$.

\section{Appendix E: details on the EDP mixture model}\label{AppendixE}

In this appendix, we provide the derivations referenced in Section~\ref{Sec52} of the main text such as the implied joint and conditional densities under the EDPM in equation~\eqref{eq:EDPMmodel}, the truncation approximation, the interval-specific conditional survival probabilities used by the G-computation algorithm in Section~\ref{Sec6}, and the explicit form of the cluster-membership weights.

\noindent Let $t_i$ be an observation of the random variable $T_i$ and $\textbf{b}_i = \big\{ b_i^{M}, b_i^{L}, b_i^{Z}\big\}$ denote a vector of random effects for subject $i$. For any arbitrary age of interest $a \in \mathbb{R}^{+}$ with $a \leq t_i$, we can write the model for the joint density using the EDPM specified in Section \ref{Sec52} as an infinite mixture:
\begin{equation}\label{Eq:EDPMixture}
\resizebox{\textwidth}{!}{$
        \begin{aligned}
            & f_H\big(t_i, \textbf{m}_i(a),\boldsymbol{\ell}_i(a),\textbf{z}_i(a),\boldsymbol{\ell}_{i,0}, \textbf{b}_{i}
            \big)
            \\ = & f_H\big(t_i, \textbf{m}_i(a),\boldsymbol{\ell}_i(a),\textbf{z}_i(a),\boldsymbol{\ell}_{i,0}\big| \textbf{b}_{i}
            \big) \times f_H\big( \textbf{b}_{i}
            \big)
            \\=& \sum_{r=1}^{\infty}\Bigl\{\gamma_r p\bigl(t_i\big| \textbf{m}_i(a),\boldsymbol{\ell}_i(a),\textbf{z}_i(a), \boldsymbol{\ell}_{i,0},\textbf{b}_{i};\beta_r\bigr) \times \sum_{s=1}^{\infty} \gamma_{s|r}p\bigl(\textbf{m}_i(a), \boldsymbol{\ell}_i(a),\textbf{z}_i(a), \boldsymbol{\ell}_{i,0} \big| \textbf{b}_{i};\theta_{s|r}\bigr) \Bigr\} \times f_{H} \big(\textbf{b}_i\big)
            \\=& \sum_{r=1}^{\infty}\Bigl\{\gamma_r p\bigl(t_i\big|  \boldsymbol{\ell}_{i,0};\beta_r\bigr) \times \sum_{s=1}^{\infty} \gamma_{s|r} p\bigl(\textbf{m}_i(a)\big| \boldsymbol{\ell}_{i,0},b_i^{M} ; \theta_{s|r}\bigr)   p\bigl(\boldsymbol{\ell}_i(a)\big| \boldsymbol{\ell}_{i,0},b_i^{L} ; \theta_{s|r}\bigr)
            p\bigl(\textbf{z}_i(a)\big| \boldsymbol{\ell}_{i,0},b_i^{Z} ; \theta_{s|r}\bigr)
            \\& p\bigl(\boldsymbol{\ell}_{i,0};\theta_{s|r}\bigr)\Bigr\} \times f_{H} \big(\textbf{b}_i\big)
        \end{aligned}
$}
    \end{equation}
where $p(.)$ denotes the corresponding density associated with distributions $F(.)$ in (\ref{eq:EDPMmodel}). Note that the presence of infinite mixtures in (\ref{Eq:EDPMixture}) complicates the posterior computation of EDP mixtures. In our work, we address this issue by truncating the outer and inner sums in (\ref{Eq:EDPMixture}) at finite values, denoted by $N$ and $M$, respectively. With the truncation approximation, we can obtain closed-form solutions for the posterior of all (global) distributions that appear in the identification of the causal parameter $\mathcal{S}_{\textbf{z},\textbf{z}_{*}}(a)$ in (\ref{Eq:NonParamIdentification}). \cite{burns2023truncation} provide a comprehensive discussion on the truncation approximation of the EDPM and present a blocked Gibbs sampler algorithm for posterior sampling. In this article, we adapt their algorithm to our proposed model, with further details provided in Appendix C.

\noindent Now suppose we have a fixed age grid $a_1 < a_2 < \ldots <  a_{k} < \ldots < a_{K} < \ldots$ such that $a_{K} \leq t_i < a_{K+1}$. Assuming local independence among ages in the grid within the inner cluster, we obtain the following joint densities at $k = K$:
\begin{equation}
\resizebox{\textwidth}{!}{$
        \begin{aligned}
            f_H\big(t_i, \textbf{m}_i(a_{K}),\boldsymbol{\ell}_i(a_{K}),\textbf{z}_i(a_{K}),\boldsymbol{\ell}_{i,0},\textbf{b}_i\big)
            &  = \sum_{r=1}^{\infty}\Bigl\{\gamma_r p\bigl(t_i\big| \boldsymbol{\ell}_{i,0};\beta_r\bigr) \times
            \sum_{s=1}^{\infty} \gamma_{s|r}
            \prod_{k=1}^{K} \big\{ p\bigl(m_i(a_{k})\big| \boldsymbol{\ell}_{i,0},b_i^{M} ;\theta_{s|r}\bigr) \times
            \\&
            p\bigl(\ell_i(a_{k})\big| \boldsymbol{\ell}_{i,0},b_i^{L} ; \theta_{s|r}\bigr) \times  p\bigl(z_i(a_{k})\big| \boldsymbol{\ell}_{i,0},b_i^{Z} ;\theta_{s|r}\bigr)\big\} \times   p\bigl(\boldsymbol{\ell}_{i,0};\theta_{s|r}\bigr)\Bigr\} \times f_{H} \big(\textbf{b}_i\big),
            \\
            f_H\big( \textbf{m}_i(a_{K}), \boldsymbol{\ell}_i(a_{K}),\textbf{z}_i(a_{K}),\boldsymbol{\ell}_{i,0},\textbf{b}_i \big)
            &
            = \sum_{r=1}^{\infty}\Bigl\{\gamma_r  \sum_{s=1}^{\infty} \gamma_{s|r}
            \prod_{k=1}^{K} \big\{ p\bigl(m_i(a_{k})\big| \boldsymbol{\ell}_{i,0},b_i^{M} ; \theta_{s|r}\bigr) \times
            p\bigl(\ell_i(a_{k})\big| \boldsymbol{\ell}_{i,0},b_i^{L} ; \theta_{s|r}\bigr) \times
            \\&
            p\bigl(z_i(a_{k})\big| \boldsymbol{\ell}_{i,0},b_i^{Z} ;\theta_{s|r}\bigr)\big\} \times
            p\bigl(\boldsymbol{\ell}_{i,0}; \theta_{s|r}\bigr)\Bigr\} \times f_{H} \big(\textbf{b}_i\big), \text{ and}
            \\
            f_H\big( \boldsymbol{\ell}_i(a_{K}), \textbf{z}_i(a_{K}), \textbf{m}_i(a_{K-1}),\boldsymbol{\ell}_{i,0} ,\textbf{b}_i \big)
            & = \sum_{r=1}^{\infty}\Bigl\{\gamma_r  \sum_{s=1}^{\infty} \gamma_{s|r} \prod_{k=1}^{K} \big\{
            p\bigl(\ell_i(a_{k})\big| \boldsymbol{\ell}_{i,0},b_i^{L} ; \theta_{s|r}\bigr)  \times
            p\bigl(z_i(a_{k})\big| \boldsymbol{\ell}_{i,0},b_i^{Z} ; \theta_{s|r}\bigr)\big\} \times
            \\&
            \prod_{k=1}^{K-1} p\bigl(m_i(a_{k})\big| \boldsymbol{\ell}_{i,0},b_i^{M} ; \theta_{s|r}\bigr)  \times p\bigl(\boldsymbol{\ell}_{i,0};\theta_{s|r}\bigr)\Bigr\} \times f_{H} \big(\textbf{b}_i\big).
        \end{aligned}
$}
\end{equation}
The EDPM induces the following conditional density for $f_H\bigl(t_i\big| \textbf{m}_i(a_{K}),\boldsymbol{\ell}_i(a_{K}),\textbf{z}_i(a_{K}),\boldsymbol{\ell}_{i,0},\textbf{b}_i\bigr)$:
    \begin{equation}\label{Eq:TconditionalDens}
\resizebox{\textwidth}{!}{$
        \begin{aligned}
            & f_H\bigl(t_i\big| \textbf{m}_i(a_{K}),\boldsymbol{\ell}_i(a_{K}),\textbf{z}_i(a_{K}),\boldsymbol{\ell}_{i,0},\textbf{b}_i\bigr)
            \\
            =& \frac{f_H\bigl(t_i, \textbf{m}_i(a_{K}),\boldsymbol{\ell}_i(a_{K}),\textbf{z}_i(a_{K}),\boldsymbol{\ell}_{i,0}, \textbf{b}_i\bigr)}{f_H\bigl( \textbf{m}_i(a_{K}),\boldsymbol{\ell}_i(a_{K}),\textbf{z}_i(a_{K}),\boldsymbol{\ell}_{i,0},\textbf{b}_i\bigr)}
            \\
            =& \frac{\splitfrac{\sum_{u=1}^{\infty}\Bigl\{\gamma_u p\bigl(t_i\big|  \boldsymbol{\ell}_{i,0};\beta_u\bigr) \times  \sum_{v=1}^{\infty} \gamma_{v|u}
            \prod_{k=1}^{K} \big\{ p\bigl(m_i(a_{k})\big|  \boldsymbol{\ell}_{i,0},b_i^{M} ; \theta_{v|u}\bigr) p\bigl(\ell_i(a_{k})\big| \boldsymbol{\ell}_{i,0},b_i^{L} ; \theta_{v|u}\bigr) }{
            p\bigl(z_i(a_{k})\big| \boldsymbol{\ell}_{i,0},b_i^{Z} ; \theta_{v|u}\bigr)\big\}  p\bigl(\boldsymbol{\ell}_{i,0};\theta_{v|u}\bigr)\Bigr\} \times f_{H} \big(\textbf{b}_i\big)}}
            {\splitfrac{\sum_{h=1}^{\infty}\Bigl\{\gamma_h  \sum_{w=1}^{\infty} \gamma_{w|h}
            \prod_{k=1}^{K} \big\{ p\bigl(m_i(a_{k})\big|  \boldsymbol{\ell}_{i,0},b_i^{M} ; \theta_{w|h}\bigr) p\bigl(\ell_i(a_{k})\big| \boldsymbol{\ell}_{i,0},b_i^{L} ; \theta_{w|h}\bigr)
            p\bigl(z_i(a_{k})\big| \boldsymbol{\ell}_{i,0},b_i^{Z} ; \theta_{w|h}\bigr)\big\} }{ p\bigl(\boldsymbol{\ell}_{i,0};\theta_{w|h}\bigr)\Bigr\} \times f_{H} \big(\textbf{b}_i\big)}}
            \\
            =& \sum_{u=1}^{\infty}\frac{\Bigl\{\gamma_u  \sum_{v=1}^{\infty} \gamma_{v|u}
            \prod_{k=1}^{K} \big\{ p\bigl(m_i(a_{k})\big| \boldsymbol{\ell}_{i,0},b_i^{M} ; \theta_{v|u}\bigr)  p\bigl(\ell_i(a_{k})\big| \boldsymbol{\ell}_{i,0},b_i^{L} ;\theta_{v|u}\bigr)
            p\bigl(z_i(a_{k})\big| \boldsymbol{\ell}_{i,0},b_i^{Z} ; \theta_{v|u}\bigr)\big\}  p\bigl(\boldsymbol{\ell}_{i,0};\theta_{v|u}\bigr)\Bigr\}}
            {\sum_{h=1}^{\infty}\Bigl\{\gamma_h  \sum_{w=1}^{\infty} \gamma_{w|h}
            \prod_{k=1}^{K} \big\{ p\bigl(m_i(a_{k})\big| \boldsymbol{\ell}_{i,0},b_i^{M} ; \theta_{w|h}\bigr)  p\bigl(\ell_i(a_{k})\big| \boldsymbol{\ell}_{i,0},b_i^{L} ; \theta_{w|h}\bigr)
            p\bigl(z_i(a_{k})\big| \boldsymbol{\ell}_{i,0},b_i^{Z} ; \theta_{w|h}\bigr)\big\}  p\bigl(\boldsymbol{\ell}_{i,0};\theta_{w|h}\bigr)\Bigr\}} \times \\&
            p\bigl(t_i\big|  \boldsymbol{\ell}_{i,0};\beta_u\bigr)
            \\
            =& \sum_{u = 1}^{\infty} w_u\bigl( \textbf{m}_i(a_{K}),\boldsymbol{\ell}_i(a_{K}),\textbf{z}_i(a_{K}),\boldsymbol{\ell}_{i,0}\bigr) \times
            p\bigl(t_i\big|  \boldsymbol{\ell}_{i,0};\beta_u\bigr),
        \end{aligned}
$}
    \end{equation}
    \[
\resizebox{\textwidth}{!}{$
        \begin{aligned}
        \text{ where }
        & w_u\bigl( \textbf{m}_i(a_{K}),\boldsymbol{\ell}_i(a_{K}),\textbf{z}_i(a_{K}),\boldsymbol{\ell}_{i,0}\bigr)
        \\=& \frac{\Bigl\{\gamma_u  \sum_{v=1}^{\infty} \gamma_{v|u}
            \prod_{k=1}^{K} \big\{ p\bigl(m_i(a_{k})\big| \boldsymbol{\ell}_{i,0},b_i^{M} ; \theta_{v|u}\bigr)  p\bigl(\ell_i(a_{k})\big| \boldsymbol{\ell}_{i,0},b_i^{L} ; \theta_{v|u}\bigr)
            p\bigl(z_i(a_{k})\big| \boldsymbol{\ell}_{i,0},b_i^{Z} ; \theta_{v|u}\bigr)\big\}  p\bigl(\boldsymbol{\ell}_{i,0};\theta_{v|u}\bigr)\Bigr\}}
            {\sum_{h=1}^{\infty}\Bigl\{\gamma_h  \sum_{w=1}^{\infty} \gamma_{w|h}
            \prod_{k=1}^{K} \big\{ p\bigl(m_i(a_{k})\big| \boldsymbol{\ell}_{i,0},b_i^{M} ; \theta_{w|h}\bigr)  p\bigl(\ell_i(a_{k})\big| \boldsymbol{\ell}_{i,0},b_i^{L} ; \theta_{w|h}\bigr)
            p\bigl(z_i(a_{k})\big| \boldsymbol{\ell}_{i,0},b_i^{Z} ; \theta_{w|h}\bigr)\big\}  p\bigl(\boldsymbol{\ell}_{i,0};\theta_{w|h}\bigr)\Bigr\}}.
        \end{aligned}
$}
    \]
The quantity of interest in this article is the survival function rather than the density function. A critical aspect of the G-computation algorithm (Section \ref{Sec6}) is the computation of interval-specific conditional survival probabilities. At each age $a_k$ in the age grid, we compute the probability of surviving from the previous age $a_{k-1}$ to the current age $a_k$, conditional on having survived to $a_{k-1}$ and the observed predictors up to $a_k$. The overall survival probability to age $a_K$ is then obtained by sequentially multiplying these interval-specific conditional survival probabilities. Specifically, at the $r^{\text{th}}$ outer-cluster of EDPM, we have:
\begin{equation}\label{Eq:IntervalSurvivalProduct}
    \begin{aligned}
        S_r(a_K) & \coloneqq P\bigl(T > a_K \big| \boldsymbol{\ell}_{0};\beta_r \bigr)
        = P\bigl(T > a_K \big| T \geq a_{K-1}, \boldsymbol{\ell}_{0};\beta_r \bigr) \times S_r(a_{K-1})
        \\
        &= \prod_{k=1}^{K} P\bigl(T > a_k \big| T \geq a_{k-1}, \boldsymbol{\ell}_{0};\beta_r \bigr),
    \end{aligned}
\end{equation}
where $a_0 = 0$ represents the baseline age.

\noindent Using (\ref{Eq:TconditionalDens}), we can compute the interval-specific conditional survival probability from age $a_{k-1}$ to age $a_k$ as:
\begin{equation}\label{Eq:TconditionalSurv}
                    \begin{aligned}
                        & P\bigl(T > a_k \big| T \geq a_{k-1}, \textbf{m}_i(a_{k}), \boldsymbol{\ell}_i(a_{k}),\textbf{z}_i(a_{k}),\boldsymbol{\ell}_{0}, \textbf{b}_i \bigr)
                        \\
                        =& 1 - P\bigl(a_{k-1} \leq T \leq a_k \big| T \geq a_{k-1}, \textbf{m}_i(a_{k}), \boldsymbol{\ell}_i(a_{k}),\textbf{z}_i(a_{k}),\boldsymbol{\ell}_{0} , \textbf{b}_i \bigr)
                        \\
                        =&  1- \int_{a_{k-1}}^{a_k}f_H\bigl(u\big| \textbf{m}_i(a_{k}),\boldsymbol{\ell}_i(a_{k}),\textbf{z}_i(a_{k}),\boldsymbol{\ell}_{i,0}, \textbf{b}_i\bigr) du
                        \\
                        =& 1 -
                        \sum_{r=1}^{\infty} w_r\bigl( \textbf{m}_i(a_{k}),\boldsymbol{\ell}_i(a_{k}),\textbf{z}_i(a_{k}),\boldsymbol{\ell}_{i,0}\bigr) \times
                        \int_{a_{k-1}}^{a_k} p\bigl(u\big| \boldsymbol{\ell}_{i,0};\beta_r\bigr)
                        du
                        \\
                        =& 1 -
                        \sum_{r=1}^{\infty} w_r\bigl( \textbf{m}_i(a_{k}),\boldsymbol{\ell}_i(a_{k}),\textbf{z}_i(a_{k}),\boldsymbol{\ell}_{i,0}\bigr) \times \Big[1-  exp\big\{- \Lambda_i\big(a_{k-1}, a_k\big| \boldsymbol{\ell}_{i,0}, \beta_r\big) \big\}\Big],
                    \end{aligned}
                \end{equation}
where $\Lambda_i\big(a_{k-1}, a_k\big| \boldsymbol{\ell}_{i,0}, \beta_r\big)$ represents the interval-specific cumulative hazard from age $a_{k-1}$ to age $a_k$:
\begin{equation}\label{Eq:IntervalPieceExpCumHazard}
    \begin{aligned}
    \Lambda_i\big(a_{k-1}, a_k\big| \boldsymbol{\ell}_{i,0}, \beta_r\big) &= \Lambda_i\big(a_k\big| \boldsymbol{\ell}_{i,0}, \beta_r\big) - \Lambda_i\big(a_{k-1}\big| \boldsymbol{\ell}_{i,0}, \beta_r\big) \\
    &= \int_{a_{k-1}}^{a_k} \lambda_0(u) exp\bigl\{ \boldsymbol{\beta}_r \boldsymbol{\ell}_{i,0}^\top \bigr\} du.
    \end{aligned}
\end{equation}

\noindent For the piecewise exponential model with $B$ disjoint intervals $[v_0, v_1), [v_1, v_2), \ldots, [v_{B-1}, v_B]$ where $v_0 = 0$ and constant baseline hazard $\lambda_b$ in interval $[v_{b-1}, v_b)$, the interval-specific cumulative hazard is computed as follows. Let $b_{start} = \min\{b: v_b > a_{k-1}\}$ and $b_{end} = \min\{b: v_b > a_k\}$ denote the piecewise interval indices containing $a_{k-1}$ and $a_k$, respectively. Then:

\textbf{Case 1:} If $b_{start} = b_{end}$ (both ages in the same piecewise interval):
\begin{equation}\label{Eq:IntervalHazardSameInterval}
    \Lambda_i\big(a_{k-1}, a_k\big| \boldsymbol{\ell}_{i,0}, \beta_r\big) = \lambda_{b_{start}} exp\bigl\{ \boldsymbol{\beta}_r \boldsymbol{\ell}_{i,0}^\top \bigr\}\big(a_k - a_{k-1}\big).
\end{equation}

\textbf{Case 2:} If $b_{start} < b_{end}$ (ages span multiple piecewise intervals):
\begin{equation}\label{Eq:IntervalHazardMultipleIntervals}
    \begin{aligned}
    \Lambda_i\big(a_{k-1}, a_k\big| \boldsymbol{\ell}_{i,0}, \beta_r\big) =& \lambda_{b_{start}} exp\bigl\{ \boldsymbol{\beta}_r \boldsymbol{\ell}_{i,0}^\top \bigr\}\big(v_{b_{start}} - a_{k-1}\big) \\
    &+ \sum_{c=b_{start}+1}^{b_{end}-1}\lambda_c exp\bigl\{ \boldsymbol{\beta}_r \boldsymbol{\ell}_{i,0}^\top \bigr\}\big(v_{c} - v_{c-1}\big) \\
    &+ \lambda_{b_{end}} exp\bigl\{ \boldsymbol{\beta}_r \boldsymbol{\ell}_{i,0}^\top \bigr\}\big(a_k - v_{b_{end}-1}\big).
    \end{aligned}
\end{equation}

\noindent The first term represents the partial contribution from the interval containing $a_{k-1}$, the sum represents full contributions from complete intervals between $a_{k-1}$ and $a_k$, and the final term represents the partial contribution from the interval containing $a_k$. This formulation ensures that the cumulative hazard (and hence mortality) is computed only for the specific age interval $[a_{k-1}, a_k)$, not from baseline age 0. 

\noindent Next, we derive the conditional density $f_H\bigl( m_i(a_K)\big|\boldsymbol{\ell}_i(a_{K}),\textbf{z}_i(a_{K}), \textbf{m}_i(a_{K-1}),\boldsymbol{\ell}_{i,0},\textbf{b}_i\bigr)$ below:
    \begin{equation}
\resizebox{\textwidth}{!}{$
        \begin{aligned}
            & f_H\bigl(m_i(a_K)\big|\boldsymbol{\ell}_i(a_{K}),\textbf{z}_i(a_{K}), \textbf{m}_i(a_{K-1}),\boldsymbol{\ell}_{i,0},\textbf{b}_i\bigr) \\
            =& \frac{f_H\bigl( \textbf{m}_i(a_{K}),\boldsymbol{\ell}_i(a_{K}),\textbf{z}_i(a_{K}),\boldsymbol{\ell}_{i,0},\textbf{b}_i\bigr)}{f_H\bigl(\textbf{m}_i(a_{K-1}),\boldsymbol{\ell}_i(a_{K}),\textbf{z}_i(a_{K}), \boldsymbol{\ell}_{i,0} ,\textbf{b}_i\bigr)}
            \\
            =& \frac{\splitfrac{\sum_{u=1}^{\infty}\Bigl\{\gamma_u  \sum_{v=1}^{\infty} \gamma_{v|u}
            \prod_{k=1}^{K} \big\{ p\bigl(m_i(a_{k})\big| \boldsymbol{\ell}_{i,0},b_i^{M} ; \theta_{v|u}\bigr)  p\bigl(\ell_i(a_{k})\big| \boldsymbol{\ell}_{i,0},b_i^{L} ; \theta_{v|u}\bigr)  p\bigl(z_i(a_{k})\big| \boldsymbol{\ell}_{i,0},b_i^{Z} ; \theta_{v|u}\bigr)\big\}}{  p\bigl(\boldsymbol{\ell}_{i,0};\theta_{v|u}\bigr)\Bigr\}\times f_{H} \big(\textbf{b}_i\big)}}
            {\splitfrac{\sum_{h=1}^{\infty}\Bigl\{\gamma_h  \sum_{w=1}^{\infty} \gamma_{w|h} \prod_{k=1}^{K-1} p\bigl(m_i(a_{k})\big| \boldsymbol{\ell}_{i,0},b_i^{M} ; \theta_{w|h}\bigr) \prod_{k=1}^{K} \big\{
            p\bigl(\ell_i(a_{k})\big| \boldsymbol{\ell}_{i,0},b_i^{L} ; \theta_{w|h}\bigr)
            p\bigl(z_i(a_{k})\big| \boldsymbol{\ell}_{i,0},b_i^{Z} ; \theta_{w|h}\bigr)\big\} }{ p\bigl(\boldsymbol{\ell}_{i,0};\theta_{w|h}\bigr)\Bigr\}\times f_{H} \big(\textbf{b}_i\big)}}
            \\
            =& \sum_{u=1}^{\infty}  \frac{ \splitfrac{\Bigl\{\gamma_u  \sum_{v=1}^{\infty} \gamma_{v|u}\,
            \prod_{k=1}^{K-1}p\bigl(m_i(a_{k})\big| \boldsymbol{\ell}_{i,0},b_i^{M} ; \theta_{v|u}\bigr) \prod_{k=1}^{K} \big\{
            p\bigl(\ell_i(a_{k})\big| \boldsymbol{\ell}_{i,0},b_i^{L} ; \theta_{v|u}\bigr) }{
            p\bigl(z_i(a_{k})\big| \boldsymbol{\ell}_{i,0},b_i^{Z} ; \theta_{v|u}\bigr)\big\}  p\bigl(\boldsymbol{\ell}_{i,0};\theta_{v|u}\bigr)\Bigr\}}}
            {\splitfrac{ \sum_{h=1}^{\infty}\Bigl\{\gamma_h  \sum_{w=1}^{\infty} \gamma_{w|h} \prod_{k=1}^{K-1} p\bigl(m_i(a_{k})\big| \boldsymbol{\ell}_{i,0},b_i^{M} ; \theta_{w|h}\bigr) \prod_{k=1}^{K} \big\{
            p\bigl(\ell_i(a_{k})\big| \boldsymbol{\ell}_{i,0},b_i^{L} ; \theta_{w|h}\bigr) }{
            p\bigl(z_i(a_{k})\big| \boldsymbol{\ell}_{i,0},b_i^{Z} ; \theta_{w|h}\bigr)\big\}   p\bigl(\boldsymbol{\ell}_{i,0};\theta_{w|h}\bigr)\Bigr\}}} \times  p\bigl(m_i(a_{K})\big| \boldsymbol{\ell}_{i,0},b_i^{M};\theta_{v|u}\bigr)
            \\
            =& \sum_{u = 1}^{\infty}  w_{u}\bigl( \textbf{m}_i(a_{K-1}),\boldsymbol{\ell}_i(a_{K}),\textbf{z}_i(a_{K}),\boldsymbol{\ell}_{i,0}\bigr) \times p\bigl(m_i(a_{K})\big| \boldsymbol{\ell}_{i,0},b_i^{M};\theta_{v|u}\bigr)
        \end{aligned}
$}
    \end{equation}
    where
    \[
\resizebox{\textwidth}{!}{$
        \begin{aligned}
     & w_{u}\bigl(\textbf{m}_i(a_{K-1}),\boldsymbol{\ell}_i(a_{K}),\textbf{z}_i(a_{K}), \boldsymbol{\ell}_{i,0}\bigr) \\
        =& \frac{\Bigl\{\gamma_u \sum_{v=1}^{\infty} \gamma_{v|u}\,
            \prod_{k=1}^{K-1} p\bigl(m_i(a_{k})\big| \boldsymbol{\ell}_{i,0},b_i^{M} ; \theta_{v|u}\bigr) \prod_{k=1}^{K} \big\{
            p\bigl(\ell_i(a_{k})\big| \boldsymbol{\ell}_{i,0},b_i^{L} ; \theta_{v|u}\bigr)
            p\bigl(z_i(a_{k})\big| \boldsymbol{\ell}_{i,0},b_i^{Z} ; \theta_{v|u}\bigr)\big\}  p\bigl(\boldsymbol{\ell}_{i,0};\theta_{v|u}\bigr)\Bigr\}}{\sum_{h=1}^{\infty}\Bigl\{\gamma_h  \sum_{w=1}^{\infty} \gamma_{w|h} \prod_{k=1}^{K-1} p\bigl(m_i(a_{k})\big| \boldsymbol{\ell}_{i,0},b_i^{M} ; \theta_{w|h}\bigr) \prod_{k=1}^{K} \big\{ p\bigl(\ell_i(a_{k})\big| \boldsymbol{\ell}_{i,0},b_i^{L} ; \theta_{w|h}\bigr) p\bigl(z_i(a_{k})\big| \boldsymbol{\ell}_{i,0},b_i^{Z} ; \theta_{w|h}\bigr)\big\} p\bigl(\boldsymbol{\ell}_{i,0};\theta_{w|h}\bigr)\Bigr\}}.
        \end{aligned}
$}
    \]

    This implies:
    \begin{equation}
\resizebox{\textwidth}{!}{$
        \begin{aligned}
            & f_H\bigl( m_i(a_{K})\big|\textbf{m}_i(a_{K-1}),\boldsymbol{\ell}_i(a_{K}),\textbf{z}_i(a_{K}), \boldsymbol{\ell}_{i,0},\textbf{b}_i\bigr)
            = \sum_{r=1}^{\infty}  w_{r}\bigl(\textbf{m}_i(a_{K-1}),\boldsymbol{\ell}_i(a_{K}),\textbf{z}_i(a_{K}), \boldsymbol{\ell}_{i,0}\bigr)\times p\bigl(m_i(a_{K})\big| \boldsymbol{\ell}_{i,0},b_i^{M};\theta_{s|r}\bigr),
        \end{aligned}
$}
    \end{equation}
where $p\bigl(m_i(a_{K})\big| \boldsymbol{\ell}_{i,0},b_i^{M};\theta_{s|r}\bigr)$ represents the density function associated with $N\big(\boldsymbol{\theta}^{M}_{s|r}\boldsymbol{\ell}^\top_{i,0} + b_i^{M} + \boldsymbol{\eta}^{M}_{s|r}\,\mathbfcal{B}_i(a_K)^{\top}, \sigma^{2,M}_{s|r} \big)$ if the mediator is a continuous random variable, or the mass function associated with $Bern\big(p_{s|r}^{M(a_{K})}\big)$, where $probit\big(p_{s|r}^{M(a_{K})}\big) = \boldsymbol{\theta}^{M}_{s|r}\boldsymbol{\ell}^\top_{i,0} + b_i^{M} + \boldsymbol{\eta}^{M}_{s|r}\,\mathbfcal{B}_i(a_K)^{\top}$, if the mediator is a binary random variable.
Similarly, for any $k = 1,\ldots, K$, we can derive the conditional densities $f_H\bigl(m_i(a_{k})\big|\textbf{m}_i(a_{k-1}),\boldsymbol{\ell}_i(a_{k}),\textbf{z}_i(a_{k}), \boldsymbol{\ell}_{i,0},\textbf{b}_i\bigr)$ and  $f_H\bigl(\ell_i(a_{k})\big|\textbf{m}_i(a_{k-1}),\boldsymbol{\ell}_i(a_{k-1}),\textbf{z}_i(a_{k}), \boldsymbol{\ell}_{i,0},\textbf{b}_i\bigr)$.

\noindent Algorithm \ref{Alg:GcompAlg} in the main text requires us to evaluate the interval-specific conditional survival probability from age $a_{k-1}$ to age $a_k$, given by:
\begin{equation}\label{Eq:GcompSurvfunc1}
    \begin{aligned}
        p^{c}_{a_{k-1} \to a_k} =&   1 -
        \sum_{r=1}^{N} w_r\bigl(\textbf{z}_1(a_{k}),\tilde{\boldsymbol{\ell}}^{c}(a_{k}), \tilde{\textbf{m}}^{c}(a_{k}),\tilde{\boldsymbol{\ell}}_{0}^{c}\bigr) \times \Big[1-  exp\big\{- \Lambda_i\big(a_{k-1}, a_k\big| \tilde{\boldsymbol{\ell}}_{0}^{c}; \beta_r^q \big) \big\}\Big],
    \end{aligned}
\end{equation}
where $\Lambda_i\big(a_{k-1}, a_k\big| \tilde{\boldsymbol{\ell}}_{0}^{c}; \beta_r^q\big)$ represents the interval-specific cumulative hazard from age $a_{k-1}$ to age $a_k$ evaluated at the $q$-th posterior draw of the outer-cluster Cox coefficients $\beta_r^q$, defined in (\ref{Eq:IntervalHazardSameInterval}) and (\ref{Eq:IntervalHazardMultipleIntervals}). The overall survival probability to age $a_k$ for Monte Carlo sample $c$ is then obtained by sequentially multiplying these interval-specific conditional survival probabilities:
\begin{equation}\label{Eq:GcompCumulativeSurv}
    S^{c}(a_k) = \prod_{j=1}^{k} p^{c}_{a_{j-1} \to a_j}.
\end{equation}

\section{Appendix F: details on the simulation data-generating process and performance metrics}\label{AppendixF}

In this appendix, we provide the implementation details of the simulation studies in the main text.

\subsection{Data-generating process (DGP)}\label{Sec71}

We designed the DGP to reflect the key features of the ARIC cohort study while allowing for known ground truth causal effects. Throughout the section, the \textit{observed visit ages}, $\{a_{i,1}, a_{i,2}, \ldots, a_{i,n_i}\}$, are subject-specific, irregularly spaced ages at which longitudinal measurements (exposure, confounder, mediator) are actually recorded for subject $i$. In contrast, the \textit{age grid}, $\{a_1, a_2, \ldots, a_K\}$, is a common set of ages at which we define and estimate causal effects (IDE, IIE, TE). This grid provides a unified framework for comparing survival probabilities across intervention regimes.

\subsubsection{Parametric Bayesian latent class joint model (BLCJM) estimation from ARIC}\label{Sec:EmpCoef}

To ensure that the data-generating process reflects realistic relationships observed in the ARIC cohort, we estimate all model parameters by fitting a parametric Bayesian latent class joint model (BLCJM) to the ARIC hypertensive population. The BLCJM represents a finite mixture analog of the proposed EDPM framework, capturing population heterogeneity through two discrete latent classes. Unlike the proposed EDPM, which employs a nested two-level clustering structure with outer ($\beta$-level) clusters for survival parameters and inner ($\theta$-level) clusters for longitudinal predictor parameters, the BLCJM uses only a single level of clustering where all model parameters are shared within each latent class. Subjects within each class share common regression parameters but differ from subjects in other classes. The latent class structure induces marginal dependence across all model components while maintaining conditional independence within classes.

\noindent Let $r \in \{1, 2\}$ denote the latent class membership for subject $i$, with class probabilities $\boldsymbol{\pi} = (\pi_1, \pi_2)$ satisfying $\sum_{r=1}^{2} \pi_r = 1$. The joint distribution of observed data for subject $i$, marginalizing over latent class membership, is:
\begin{equation}\label{Eq:SimDGMMarginalJoint}
    \begin{aligned}
        f_{\text{BLCJM}}\bigl(t_i, \delta_i, \textbf{m}_i(a), \boldsymbol{\ell}_i(a), \textbf{z}_i(a), \boldsymbol{\ell}_{i,0}\bigr) & = \sum_{r=1}^{2} \pi_r \Big\{ f^{(r)}\bigl(t_i, \delta_i \big|  \boldsymbol{\ell}_{i,0}; \boldsymbol{\lambda}^{(r)}, \boldsymbol{\beta}^{(r)}\bigr)
        \times  f^{(r)}\bigl (\textbf{m}_i(a) \big| \boldsymbol{\ell}_{i,0}, b_i^M; \boldsymbol{\theta}^{M,(r)}, \boldsymbol{\eta}^{M,(r)}\bigr) \\
        &\times f^{(r)}\bigl(\boldsymbol{\ell}_i(a) \big| \boldsymbol{\ell}_{i,0}, b_i^L; \boldsymbol{\theta}^{L,(r)}, \boldsymbol{\eta}^{L,(r)}\bigr)
        \times f^{(r)}\bigl(\textbf{z}_i(a) \big| \boldsymbol{\ell}_{i,0}, b_i^Z; \boldsymbol{\theta}^{Z,(r)}, \boldsymbol{\eta}^{Z,(r)}\bigr) \\
        &\times f^{(r)}\bigl(\boldsymbol{\ell}_{i,0}; \boldsymbol{\mu}_0^{(r)}, \boldsymbol{\sigma}_0^{(r)}, \textbf{p}_0^{(r)}\bigr)\Big\}.
    \end{aligned}
\end{equation}
This two-class latent mixture model is a special case of the EDPM framework (Section~\ref{Sec52}) with: (i) outer truncation level $N = 2$, (ii) no inner clustering (equivalently, $M = 1$), and (iii) class-specific rather than enriched Dirichlet process-distributed parameters. In~\eqref{Eq:SimDGMMarginalJoint}, the latent class-specific densities correspond to the local distributions specified in~\eqref{Eq:CoxModel} and~\eqref{eq:Localmodels}. Although the general framework includes subject-specific random intercepts $b_i^{M}$, $b_i^{L}$, and $b_i^{Z}$ in the longitudinal submodels, these are set to zero in all simulation studies (see below). In our implementation, we use two baseline covariates $\textbf{L}_{0} = (L_{0,1}, L_{0,2})$, where $L_{0,1}$ and $L_{0,2}$ denote baseline BMI (continuous) and biological sex (binary), respectively.

\noindent  From this parametric BLCJM, we extract the following class-specific posterior mean estimates for $r = 1, 2$:
\begin{enumerate}[label=(\roman*)]
    \item Baseline covariate distribution parameters: $\hat{\mu}_{0,1}^{(r)}$, $\hat{\sigma}_{0,1}^{2,(r)}$ (for continuous BMI) and $\hat{p}_{0,2}^{(r)}$ (for binary sex);
    \item Exposure model coefficients: $\hat{\boldsymbol{\theta}}^{Z,(r)} = (\hat{\theta}_0^{Z,(r)}, \hat{\theta}_1^{Z,(r)}, \hat{\theta}_2^{Z,(r)})^\top$ and spline coefficients $\hat{\boldsymbol{\eta}}^{Z,(r)} = (\hat{\eta}_1^{Z,(r)}, \ldots, \hat{\eta}_4^{Z,(r)})^\top$;
    \item Confounder model coefficients: $\hat{\boldsymbol{\theta}}^{L,(r)} = (\hat{\theta}_0^{L,(r)}, \hat{\theta}_1^{L,(r)}, \hat{\theta}_2^{L,(r)}, \hat{\theta}_Z^{L,(r)})^\top$ and spline coefficients $\hat{\boldsymbol{\eta}}^{L,(r)}$;
    \item Mediator model coefficients: $\hat{\boldsymbol{\theta}}^{M,(r)} = (\hat{\theta}_0^{M,(r)}, \hat{\theta}_1^{M,(r)}, \hat{\theta}_2^{M,(r)}, \hat{\theta}_Z^{M,(r)}, \hat{\theta}_L^{M,(r)})^\top$, spline coefficients $\hat{\boldsymbol{\eta}}^{M,(r)}$, and residual variance $\hat{\sigma}^{2,M,(r)}$;
    \item Survival model parameters: piecewise constant baseline hazards $\hat{\boldsymbol{\lambda}}^{(r)} = (\hat{\lambda}_1^{(r)}, \ldots, \hat{\lambda}_B^{(r)})^\top$ and log-hazard ratio coefficients $\hat{\boldsymbol{\beta}}^{(r)} = (\hat{\beta}_1^{(r)}, \ldots, \hat{\beta}_5^{(r)})^\top$.
\end{enumerate}
These parameter estimates form the basis for generating replicated datasets and computing ground truth causal effects via Monte Carlo integration.

\noindent Although the EDPM framework in Section~\ref{Sec52} and the real data analysis in Section~\ref{Sec8} include subject-specific random intercepts $b_i^M$, $b_i^L$, and $b_i^Z$ in the longitudinal submodels, we exclude them from both the data-generating mechanism and the fitted models in the simulation studies. This is a deliberate simplification. In our setup, the mixture structure already induces substantial between-subject heterogeneity as individuals assigned to different latent classes inherit distinct intercepts, slopes, and spline terms, which capture much of the variability that random intercepts are typically used to model. At the same time, the causal effects we target are quite small (e.g., the true IDE is on the order of $10^{-3}$), and adding random intercepts only increases Monte Carlo variability in the g-computation step. In practice, this extra noise tends to obscure the signal and can lead to inflated bias that reflects simulation instability rather than deficiencies of the estimators. By omitting random intercepts, we keep the focus on how the EDPM and BLCJM handle model misspecification, without introducing an additional source of variability that is not central to the comparison. In contrast, for the real data analysis—where we fit the model once, the effects are larger, and capturing subject-level variation is part of the scientific goal—we retain the random intercepts.

\subsubsection{Sampling visit ages $a_{ij}$}

To ensure realistic visit patterns, we use the Bayesian bootstrap \citep{rubin1981bayesian} to sample visit age sequences from the ARIC dataset. Specifically, let $\big\{ \textbf{a}_i^{\text{ARIC}}\big\}_{i=1}^{n_{\text{ARIC}}}$ denote the  observed visit ages for the $n_{\text{ARIC}}$ subjects in the ARIC cohort, where $\textbf{a}_i^{\text{ARIC}} = (a_{i,1}^{\text{ARIC}}, \ldots, a_{i,n_i}^{\text{ARIC}})$ is the vector of visit ages for subject $i$.

For each simulation replicate, we:
\begin{enumerate}
    \item Draw Bayesian bootstrap weights $\textbf{w} = (w_1, \ldots, w_{n_{\text{ARIC}}}) \sim \text{Dirichlet}(1, \ldots, 1)$.
    \item For each simulated subject $i = 1, \ldots, n$, sample an ARIC subject index $j_i$ with probability $w_{j_i}$, and set $\textbf{a}_i = \textbf{a}_{j_i}^{\text{ARIC}}$.
\end{enumerate}

\noindent This procedure preserves the empirical distribution of visit ages observed in ARIC.
The visit age sequences $\textbf{a}_i = (a_{i,1}, \ldots, a_{i,n_i})$ sampled from ARIC reflect realistic irregularity, with subjects having between 1 and 5 visits at ages spanning approximately 45--90 years.

\subsubsection{Computation of true causal effects}

The true interventional direct effect (IDE), interventional indirect effect (IIE), and total effect (TE) at each age $a \in \{65, 75\}$ on the age grid are computed via Monte Carlo integration with $C^{*} = 10{,}000$ simulated subjects. For each Monte Carlo replicate, we generate potential outcomes under the relevant intervention regimes:
\begin{itemize}
    \item $\mathcal{S}^{\text{true}}_{\textbf{z}, \textbf{z}}(a)$: Survival probability under always-treated ($\textbf{z} = \textbf{1}$) with mediators drawn from the always-treated distribution.
    \item $\mathcal{S}^{\text{true}}_{\textbf{z}_{*}, \textbf{z}_{*}}(a)$: Survival probability under never-treated ($\textbf{z}_{*} = \textbf{0}$) with mediators drawn from the never-treated distribution.
    \item $\mathcal{S}^{\text{true}}_{\textbf{z}, \textbf{z}_{*}}(a)$: Survival probability under always-treated with mediators drawn from the never-treated distribution.
\end{itemize}
The true effects are then: i) $\text{IDE}^{\text{true}}(a) = \mathcal{S}^{\text{true}}_{\textbf{z}, \textbf{z}_{*}}(a) - \mathcal{S}^{\text{true}}_{\textbf{z}_{*}, \textbf{z}_{*}}(a)$, ii) $\text{IIE}^{\text{true}}(a) = \mathcal{S}^{\text{true}}_{\textbf{z}, \textbf{z}}(a) - \mathcal{S}^{\text{true}}_{\textbf{z}, \textbf{z}_{*}}(a)$, and iii) $\text{TE}^{\text{true}}(a) = \mathcal{S}^{\text{true}}_{\textbf{z}, \textbf{z}}(a) - \mathcal{S}^{\text{true}}_{\textbf{z}_{*}, \textbf{z}_{*}}(a)$.

\subsection{Model settings}

Posterior inference is based on four MCMC chains, each with $12{,}500$ total iterations. We discard the first $10{,}000$ iterations of each chain as burn-in and retain every tenth post--burn-in draw, yielding a total of $Q = 1{,}000$ posterior samples altogether for inference. For the EDPM model, we set the outer and inner truncation levels to $N = 10$ and $M = 4$, respectively. The thin-plate splines in the longitudinal models are specified with knots $D = 7$, placed at five-year age intervals between $50$ and $80$. We evaluate causal effects on the age grid $a_k \in \{65, 75\}$ and compute effect estimates via G-computation using $C^{*} = 10{,}000$ Monte Carlo samples at each retained MCMC iteration.

\subsection{Performance metrics}

For each causal estimand $\psi \in \{\text{IDE}(a), \text{IIE}(a), \text{TE}(a)\}$ evaluated at ages $a \in \{65, 75\}$, we summarize performance across $R = 500$ simulation replicates using the following metrics:
\begin{itemize}
    \item \textbf{Bias:} $\text{Bias}(\hat{\psi}) = R^{-1} \sum_{r=1}^{R} (\hat{\psi}^{(r)} - \psi_0)$, where $\hat{\psi}^{(r)}$ denotes the posterior mean estimate from replicate $r$ and $\psi_0$ is the true estimand value.
    \item \textbf{Mean squared error:} $\text{MSE}(\hat{\psi}) = R^{-1} \sum_{r=1}^{R} (\hat{\psi}^{(r)} - \psi_0)^2$.
    \item \textbf{95\% coverage:} The proportion of replicates in which the 95\% equal-tailed credible interval $[L^{(r)}, U^{(r)}]$ contains $\psi_0$.
    \item \textbf{Average interval width:} $R^{-1} \sum_{r=1}^{R} (U^{(r)} - L^{(r)})$.
\end{itemize}

\section{Appendix G: numerical results for causal effect estimates}\label{AppendixG}

This appendix provides the numerical values for the interventional direct effect (IDE), interventional indirect effect (IIE), and total effect (TE) estimates presented in Figures~\ref{Fig:Interventional_Effects} and~\ref{Fig:Posterior_Prop_of_Positive_Effects} of the main text. Posterior means are reported along with 95\% credible intervals.

\begin{table}[H]
\singlespacing
\centering
\caption{Interventional direct, indirect, and total effects on survival probability---Overall hypertensive at baseline population. Results shown are posterior means with 95\% credible intervals.}
\label{Tab:EffectsOverall}
\small
\begin{tabular}{c|c|c|c}
\hline
Age & Direct Effect & Indirect Effect & Total Effect \\
\hline
40 & $-$0.017 ($-$0.064, 0.0025) & $-$0.082 ($-$0.21, $-$0.0097) & $-$0.099 ($-$0.23, $-$0.023) \\
45 & $-$0.047 ($-$0.15, 0.0064) & $-$0.077 ($-$0.20, 0.013) & $-$0.12 ($-$0.26, $-$0.031) \\
50 & $-$0.062 ($-$0.19, 0.014) & $-$0.068 ($-$0.19, 0.039) & $-$0.13 ($-$0.27, $-$0.026) \\
55 & $-$0.067 ($-$0.21, 0.021) & $-$0.060 ($-$0.18, 0.056) & $-$0.13 ($-$0.28, $-$0.021) \\
60 & $-$0.067 ($-$0.21, 0.028) & $-$0.053 ($-$0.17, 0.068) & $-$0.12 ($-$0.27, $-$0.0080) \\
65 & $-$0.063 ($-$0.21, 0.036) & $-$0.047 ($-$0.16, 0.078) & $-$0.11 ($-$0.26, 0.0062) \\
70 & $-$0.058 ($-$0.20, 0.048) & $-$0.041 ($-$0.16, 0.084) & $-$0.099 ($-$0.25, 0.019) \\
75 & $-$0.051 ($-$0.20, 0.057) & $-$0.037 ($-$0.15, 0.088) & $-$0.088 ($-$0.24, 0.034) \\
80 & $-$0.046 ($-$0.19, 0.067) & $-$0.032 ($-$0.14, 0.096) & $-$0.078 ($-$0.23, 0.046) \\
85 & $-$0.042 ($-$0.18, 0.069) & $-$0.027 ($-$0.13, 0.098) & $-$0.069 ($-$0.22, 0.056) \\
\hline
\end{tabular}
\end{table}

\begin{table}[H]
\singlespacing
\centering
\caption{Interventional direct, indirect, and total effects on survival probability---Males only hypertensive at baseline population. Results shown are posterior means with 95\% credible intervals.}
\label{Tab:EffectsMale}
\small
\begin{tabular}{c|c|c|c}
\hline
Age & Direct Effect & Indirect Effect & Total Effect \\
\hline
40 & 0.00010 ($-$0.0074, 0.0065) & $-$0.18 ($-$0.45, $-$0.019) & $-$0.18 ($-$0.45, $-$0.018) \\
45 & 0.00030 ($-$0.018, 0.019) & $-$0.17 ($-$0.43, $-$0.013) & $-$0.17 ($-$0.43, $-$0.012) \\
50 & 0.00060 ($-$0.026, 0.029) & $-$0.16 ($-$0.40, $-$0.0081) & $-$0.16 ($-$0.40, $-$0.0031) \\
55 & 0.00090 ($-$0.032, 0.036) & $-$0.15 ($-$0.37, $-$0.0027) & $-$0.14 ($-$0.37, 0.0048) \\
60 & 0.00070 ($-$0.039, 0.042) & $-$0.14 ($-$0.35, 0.0024) & $-$0.13 ($-$0.35, 0.011) \\
65 & 0.00070 ($-$0.044, 0.049) & $-$0.13 ($-$0.33, 0.011) & $-$0.13 ($-$0.33, 0.020) \\
70 & 0.00060 ($-$0.051, 0.056) & $-$0.12 ($-$0.31, 0.019) & $-$0.12 ($-$0.31, 0.030) \\
75 & 0.00040 ($-$0.056, 0.063) & $-$0.11 ($-$0.30, 0.028) & $-$0.11 ($-$0.30, 0.037) \\
80 & $-$0.00030 ($-$0.060, 0.066) & $-$0.10 ($-$0.28, 0.034) & $-$0.10 ($-$0.29, 0.044) \\
85 & $-$0.0010 ($-$0.062, 0.067) & $-$0.096 ($-$0.27, 0.041) & $-$0.097 ($-$0.27, 0.050) \\
\hline
\end{tabular}
\end{table}

\begin{table}[H]
\singlespacing
\centering
\caption{Interventional direct, indirect, and total effects on survival probability---Females only hypertensive at baseline population. Results shown are posterior means with 95\% credible intervals.}
\label{Tab:EffectsFemale}
\small
\begin{tabular}{c|c|c|c}
\hline
Age & Direct Effect & Indirect Effect & Total Effect \\
\hline
40 & $-$0.012 ($-$0.045, 0.0034) & $-$0.069 ($-$0.21, $-$0.0036) & $-$0.081 ($-$0.21, $-$0.011) \\
45 & $-$0.033 ($-$0.11, 0.0069) & $-$0.068 ($-$0.20, $-$0.00030) & $-$0.10 ($-$0.22, $-$0.017) \\
50 & $-$0.046 ($-$0.15, 0.0099) & $-$0.066 ($-$0.19, 0.0029) & $-$0.11 ($-$0.23, $-$0.020) \\
55 & $-$0.056 ($-$0.18, 0.011) & $-$0.064 ($-$0.19, 0.0052) & $-$0.12 ($-$0.23, $-$0.023) \\
60 & $-$0.062 ($-$0.19, 0.014) & $-$0.062 ($-$0.18, 0.0070) & $-$0.12 ($-$0.24, $-$0.024) \\
65 & $-$0.067 ($-$0.20, 0.016) & $-$0.059 ($-$0.18, 0.0096) & $-$0.13 ($-$0.24, $-$0.023) \\
70 & $-$0.069 ($-$0.21, 0.019) & $-$0.057 ($-$0.17, 0.011) & $-$0.13 ($-$0.24, $-$0.022) \\
75 & $-$0.068 ($-$0.21, 0.021) & $-$0.055 ($-$0.17, 0.014) & $-$0.12 ($-$0.24, $-$0.017) \\
80 & $-$0.065 ($-$0.21, 0.025) & $-$0.053 ($-$0.16, 0.016) & $-$0.12 ($-$0.24, $-$0.014) \\
85 & $-$0.062 ($-$0.20, 0.031) & $-$0.051 ($-$0.16, 0.018) & $-$0.11 ($-$0.23, $-$0.0076) \\
\hline
\end{tabular}
\end{table}

\begin{table}[H]
\singlespacing
\centering
\caption{Posterior proportion of positive interventional direct, indirect, and total effects on survival probability---Overall hypertensive at baseline population.}
\label{Tab:PosteriorProbOverall}
\small
\begin{tabular}{c|c|c|c}
\hline
Age & P(IDE $>$ 0) & P(IIE $>$ 0) & P(TE $>$ 0) \\
\hline
40 & 0.086 & 0.010 & 0.000 \\
45 & 0.064 & 0.049 & 0.0010 \\
50 & 0.073 & 0.083 & 0.0050 \\
55 & 0.094 & 0.12 & 0.0090 \\
60 & 0.13 & 0.15 & 0.017 \\
65 & 0.15 & 0.18 & 0.032 \\
70 & 0.19 & 0.21 & 0.064 \\
75 & 0.22 & 0.23 & 0.10 \\
80 & 0.24 & 0.26 & 0.14 \\
85 & 0.26 & 0.29 & 0.17 \\
\hline
\end{tabular}
\end{table}

\begin{table}[H]
\singlespacing
\centering
\caption{Posterior probabilities of positive interventional direct, indirect, and total effects on survival probability---Males only hypertensive at baseline population.}
\label{Tab:PosteriorProbMale}
\small
\begin{tabular}{c|c|c|c}
\hline
Age & P(IDE $>$ 0) & P(IIE $>$ 0) & P(TE $>$ 0) \\
\hline
40 & 0.52 & 0.0010 & 0.0010 \\
45 & 0.51 & 0.0080 & 0.0070 \\
50 & 0.52 & 0.014 & 0.024 \\
55 & 0.51 & 0.022 & 0.034 \\
60 & 0.50 & 0.030 & 0.047 \\
65 & 0.50 & 0.041 & 0.059 \\
70 & 0.49 & 0.049 & 0.074 \\
75 & 0.49 & 0.060 & 0.088 \\
80 & 0.48 & 0.068 & 0.094 \\
85 & 0.47 & 0.083 & 0.11 \\
\hline
\end{tabular}
\end{table}

\begin{table}[H]
\singlespacing
\centering
\caption{Posterior probabilities of positive interventional direct, indirect, and total effects on survival probability---Females only hypertensive at baseline population.}
\label{Tab:PosteriorProbFemale}
\small
\begin{tabular}{c|c|c|c}
\hline
Age & P(IDE $>$ 0) & P(IIE $>$ 0) & P(TE $>$ 0) \\
\hline
40 & 0.13 & 0.0090 & 0.000 \\
45 & 0.094 & 0.023 & 0.000 \\
50 & 0.087 & 0.035 & 0.0010 \\
55 & 0.076 & 0.042 & 0.0010 \\
60 & 0.079 & 0.050 & 0.0030 \\
65 & 0.077 & 0.059 & 0.0040 \\
70 & 0.075 & 0.067 & 0.0060 \\
75 & 0.099 & 0.076 & 0.0070 \\
80 & 0.12 & 0.089 & 0.0080 \\
85 & 0.15 & 0.091 & 0.016 \\
\hline
\end{tabular}
\end{table}

\end{document}